\documentclass[12pt]{article}

\usepackage{epsfig} 
\usepackage{amsmath,amsthm, amssymb, latexsym, bm}
\usepackage{lipsum}
\usepackage{color,soul}
\usepackage{xspace}
\usepackage{algorithm}
\usepackage{algpseudocode}
\usepackage{graphicx}
\usepackage{natbib}
\usepackage{multirow}
\usepackage{subcaption}

\allowdisplaybreaks
\algnewcommand\algorithmicforeach{\textbf{for each}}
\algdef{S}[FOR]{ForEach}[1]{\algorithmicforeach\ #1\ \algorithmicdo}
\algrenewcommand\algorithmicrequire{\textbf{Input}}
\algrenewcommand\algorithmicensure{\textbf{Output}}
\newtheorem{theorem}{Theorem}
\newtheorem{lemma}{Lemma}
\newtheorem{corollary}{Corollary}

\newtheorem{example}{\quad Example}
\newtheorem*{example2b*}{\quad Example 2 (Continuation)}

\def\*#1{\bm{#1}}
\def\btheta{\mathop{\theta\kern-.45em\hbox{$\theta$}}\nolimits}

\def\bfeta{\mathop{\eta\kern-.5em\hbox{$\eta$}}\nolimits}
\def\bxi{\mathop{\xi\kern-.45em\hbox{$\xi$}}\nolimits}
\def\P{\mathbb P}
\def\A{\mathbb A}
\def\B{\mathbb B}
\def\C{\mathbb C}
\newcommand{\RR}{\mathbb{R}}
\newcommand{\R}{\mathbb{R}}

\newcommand{\diag}{\hbox{diag}}

\newcommand\numberthis{\addtocounter{equation}{1}\tag{\theequation}}

\title{\Large Repro Samples Method for Finite- and Large-Sample Inferences}
\author{Min-ge Xie and Peng Wang}
\footnotetext[1]{\footnotesize Min-ge Xie is a Distinguished Professor, Department of Statistics, Rutgers, The State University of New Jersey, Piscataway, NJ 08854, Email: mxie@stat.rutgers.edu. Peng Wang is an Associate Professor, Department of Operations, Business Analytics and Information System, Email: wangp9@ucmail.uc.edu.  The research is supported in part by NSF grants  DMS1812048, DMS2015373 and DMS2027855.}
\date{}

\begin{document}

\maketitle

\thispagestyle{empty}

\vspace{-6mm}
\begin{abstract}
This article presents a novel, general, and effective simulation-inspired approach, called {\it repro samples method}, to conduct statistical inference. The approach studies the performance of artificial samples, referred to as {\it repro samples}, obtained by mimicking  the true observed sample to achieve uncertainty quantification and construct confidence sets for parameters of interest with guaranteed coverage rates. Both exact and asymptotic inferences are developed. An attractive feature of the general framework developed is that it does not rely on the large sample central limit theorem and is likelihood-free. As such, it is thus effective for complicated inference problems which we can not solve using the large sample central limit theorem. The proposed method is applicable to a wide range of problems, including many open questions where solutions were previously unavailable, for example, those involving discrete or non-numerical parameters. To reduce the large computational cost of such inference problems, we develop a unique matching scheme to obtain a data-driven candidate set. Moreover, we show the advantages of the proposed framework over the classical Neyman-Pearson framework.  We demonstrate the effectiveness of the proposed approach on various models throughout the paper and provide a case study that addresses an open inference question on how to quantify the uncertainty  for the unknown number of components in a normal mixture model.  To evaluate the empirical performance of our repro samples method, we conduct simulations and study real data examples with comparisons to existing approaches. Although the development pertains to the settings where the large sample central limit theorem does not apply, it also has direct extensions to the cases where the central limit theorem does hold. 
\end{abstract}

\newpage
\noindent 
Key words: Artificial samples;  Conformal or matching samples; Discrete and continuous parameters; 
Parametric and nonparametric settings; Mixture model; Nuisance parameters

% \vspace{3mm}

\newpage
\setcounter{page}{1} 

\section{Introduction}

Statistical inference includes quantifying uncertainty inherited from sampling based on an observed copy of data. 
The approach of repeatedly creating similar artificial (Monte-Carlo) data 
to help assess the uncertainty has proved in the literature to be an effective method. 
Examples include bootstrap, approximate Bayesian computing (ABC), generalized fiducial inference (GFI) algorithm,
permutation, among others; see, e.g., \citet{EfroTibs93,fishman1996, Robert2016, Hannig2016}, 
and references therein.  
Most of these works, however, rely on the large sample central limit theorem (CLT) or approximations
to justify their validity of statistical inference. 
Their applicability to some complicated
inference problems, especially those to which the large sample CLT does not apply, is limited. In this article, we develop a novel and wide reaching artificial-sample-based inferential framework, which we support with both finite and large sample theories.
 The proposed framework has the advantage of providing valid inference without relying on likelihood functions or large sample theories. 
It is especially effective for difficult inference problems such as those involving discrete and/or non-numerical parameters (e.g., the unknown number of clusters, model A, B, C, etc.) among others. 
Although the repro samples method pertains to settings in which the large sample CLT does not apply, it also has direct extensions to the cases where the large sample theorem does hold.

Suppose that
the sample data $\*Y = (Y_1, \ldots, Y_n)^\top \in {\cal Y} \subset \RR^n$
are generated from a model 
$
\*Y | {\btheta}_0 \sim M_{{\bf \theta}_0}(\cdot),$ where
${\btheta}_0 \in \Theta$ is the true value of the model parameters that can be either numerical or nonnumerical or a mixture of both. 
We assume (and assume only) that we know how to simulate data from $M_{\bf \theta}$, given $\btheta$. Often,  $\*Y | {\btheta}_0 \sim M_{{\bf \theta}_0}(\cdot)$ 
can be re-expressed in a structure equation form (also referred to as an ``algorithmic model''):
\begin{equation}\label{eq:1}
\*Y = G({\btheta}_0, \*U), 
\end{equation} 
where $G(\cdot, \cdot)$ is a known mapping from $\Theta \times {\cal U} \mapsto {\cal Y}$ and $\*U = (U_1, \ldots U_r)^\top \in {\cal U}\subset \RR^r$, $r \ge n$, is a random vector whose distribution is known or from which it can be simulated.
The reason that we adopt the formulation of \eqref{eq:1} is that this model specification is very general.
For example, it includes the traditional statistical model specification using a density function as its special case:
Let $Y \sim f_{\theta_0}(y)$ be a sample from a density function $f_{\theta}(y)$ with true parameter value $\theta_0$.
We then express $Y$ in the form of \eqref{eq:1} as  $Y = G(\theta_0, U)$ where $G(\theta, U) =  F_\theta^{-1}(U)$, $U \sim U(0,1)$ and
$F_\theta(y)$ is the cumulative distribution function of $f_\theta(y)$.
Additional examples, including more complicated ones, 
will be provided throughout the paper.

The repro samples method  uses a simple yet fundamental idea: study the performance of artificial samples that are obtained by mimicking the sampling mechanism 
of the observed data; the
artificial samples then help to 
quantify the uncertainty in estimation of the associated model and parameters. Specifically,~let 
\begin{equation}
\*y_{obs} = G({\btheta}_0, \*u^{rel}) \nonumber
\end{equation}
be the realized observations, where $\*u^{rel}$ is the corresponding unobserved (thus unknown) realization of $\*U$. 
For any $\btheta \in \Theta$ and 
a copy of $\*u^* \sim \*U$, we can obtain 
a copy of artificial sample $\*y^* = G({\btheta}, \*u^*)$, 
which we refer to in this paper as {\it repro samples}.  
Note that, when $\btheta = \btheta_0$ and $\*u^*$ matches $\*u^{rel}$, the corresponding repro sample $\*y^* = G(\btheta, \*u^*)$ 
will match $\*y_{obs}$.
 However, both $\btheta_0$ and $\*u^{rel}$ are unknown. 
 Since $\*u^{rel}$ and $\*u^*$ are both realizations from  the same $\*U$ distribution, we use this fact, together with a so-called {\it nuclear mapping function} on ${\cal U} \times \Theta \to \RR^q$, 
 $q \leq n$, to quantify how likely a $\*u^*$ is a potential match for $\*u^{rel}$.  For the collection of $\*u^*$'s that are deemed as ``likely" matches, we look for the corresponding $\btheta$ that  produces a matching repro sample $\*y^* = G(\btheta, \*u^*)$ with $\*y^* = \*y_{obs}$.
 We use these $\btheta$ to help recover $\btheta_0$.
We formally formulate the procedure 
and show that a subset of $\Theta$ constructed by this procedure has a desired frequency coverage rate and thus is a confidence set for $\btheta_0$.

In Section 2, we formally describe the repro samples framework. For a given potential value of $\btheta$, we construct a fixed level-$\alpha$ Borel set of $\*u^*$ through a nuclear mapping function, and 
analyze whether any $\*u^*$ in the Borel set produces repro samples $\*y^* = G({\btheta}, {\*u}^*)$ that match $\*y_{obs}$. By controlling 
$\*u^*$ through the fixed Borel set and seeking for the matching, we obtain a theoretically guaranteed level-$\alpha$ confidence set for $\btheta_0$. The use of  nuclear mapping functions is new and general, and it greatly increases the flexibility and utility of the repro samples method. If the distribution of the nuclear mapping has an explicit form, so does the confidence set. When the the confidence set cannot be explicitly expressed, 
we provide a generic algorithm.  Moreover, the repro samples method can also be used to obtain p-values to conduct hypothesis testing 
for a wide range of problems.
In addition, we further extended the framework beyond the parametric model (\ref{eq:1}), where the data generating model $G(\btheta, \*U)$ is not completely known. 
Although the repro samples method shares many insights with some Bayesian and fiducial procedures,
its development is purely  under the frequentist paradigm. It
sidesteps the need to specify a likelihood or a loss function, and also the need to use large sample theories or the Dumpster-Shafer calculus. 
It has a wide range of applications,
 including problems involving any type of parameter spaces, exact and asymptotic inferences, parametric and nonparametric models,  among others.

Section 3 provides 
additional developments we use to facilitate our understanding and broaden the application of  the repro samples method. 
In Section 3.1, we explore  
the relationship of our repro samples method with the classical Neyman-Pearson  hypothesis testing procedure
by considering a special case in which we choose a test statistic as our nuclear mapping function. In general, 
a nuclear mapping function does not need to be a test statistics. Even in special cases in which the nuclear mapping is defined by a typical test statistics, we show that a confidence set obtained by the repro samples method is either the same or smaller than the one obtained by inverting the classical Neyman-Pearson hypothesis testing procedure. 

We also obtain a set of optimality results.
Furthermore, we develop in Section 3.2 a general technique to handle nuisance parameters through profiling.
The method
maintains the desired 
coverage rates for confidence sets constructed for target parameters, even in 
finite sample settings. 
Finally in Section 3.3, we improve the computational efficiency by investigating the use of data-dependent candidate sets to greatly reduce the search region in the parameter spaces. Using a candidate set is particularly effective when the target parameter of the inference is discrete, since we can take advantage of a many-to-one mapping that is uniquely inherent in a repro samples procedure. We provide a Monte-Carlo approach to obtain the candidate set and supporting theories to ensure the coverage rate of the final confidence set.

To demonstrate the utility and performance of the repro samples method, we use a normal mixture model as a case study example in Section 4.
Accounting for uncertainty in an estimation of the unknown number of components in a mixture model, say $\tau_0$, is a 
notoriously difficult and unsolved 
inference problem \citep[e.g.,][]{roeder_graphical_1994,chen_inference_2012}, partly because the parameter space of $\tau_0$ is discrete and the large sample CLT does not apply on any point estimator $\hat \tau$ of $\tau_0$.  
Here, we
construct a finite-sample confidence set for the number of components $\tau_0$ in the normal mixture model.
We demonstrate that the repro samples method can effectively solve this open inference problem while handling a large and varying number of nuisance parameters.
We perform simulation and real data analyses, and provide comparisons to existing frequentist and Bayesian approaches. The analyses is revealing. First, the regular criterion-based point estimators (even consistent estimators such as the BIC estimators) are often biased for smaller $\tau$ and they rarely recover the true $\tau_0$. Second, existing (frequentist) likelihood-based procedures do not provide any upper bounds, and they have low power and work only asymptotically. Third, Bayesian credible sets obtained using the  posteriors of $\tau$ perform poorly in terms of covering $\tau_0$ in repeated runs, even when a uniform flat prior of $\tau$ (over a range including the true value $\tau_0$) is used. In addition, credible sets for $\tau_0$ are very sensitive to the prior choices of both the target parameter $\tau$ and also (surprisingly) the nuisance parameters. By contrast, the proposed repro sample confidence sets always recover $\tau_0$ with the desired coverage rate. They also provide a full picture of the sampling uncertainty pertaining to the estimation of $\tau_0$. 

The repro samples method stems
from 
several existing artificial-sample-based inference procedures across Bayesian, frequentist and fiducial (BFF) paradigms, particularly,  approximate Bayesian computing (ABC), generalized fiducial inference (GFI) and bootstrap methods, in which artificial samples are used to help quantify inference uncertainty. It also extends a key idea of the inferential model (IM) \citep{Martin2015} that the inference uncertainty is quantified through the auxiliary $\*U$. 
We view the repro samples approach as an improvement over the current existing inference methods and provide comparisons throughout the paper when the topic arises,
At the outset, we emphasize that the repro samples method side-steps many 
restrictions encountered in the aforementioned BFF approaches. First, bootstrap,  GFI and ABC methods rely on large sample CLT or some version of Bernstein von Mises-type theorem
(thus CLT) to justify their (frequentist) frequency performance, therefore does not guarantee Small sample frequency performance. Second, both ABC and GFI methods are approximation methods that require a predetermined threshold $\epsilon$ to judge whether an artificial sample is sufficiently close to the observed sample. Their supporting theories need $\epsilon \to 0$ at a fast rate as the sample size $n \to \infty$. However, when $\epsilon \to 0$  at a fast rate, the sample rejection rate
in their algorithms
tends to $1$. Thus, how to choose an appropriate $\epsilon$ to balance the procedures' computational efficiency and inference validity
is an unsolved problem \citep{Li2016}. 
Finally, IM method is designed for epistemic probabilistic inference, and its complexity and reliance on Dempster-Shafer calculus make it hard to implement under complex settings. The IM method is not easily adaptable for asymptotic inference, and the restriction of using an ancillary to quantify uncertainty also limits its use in applications.
The repro samples method, developed fully following the frequentist setting, addresses or sidesteps all these issues.

The remainder of this article is arranged as follows. Section 2 develops the general framework of the repro samples method with supporting theories and associated algorithms. Sections 3 considers several important aspects that further support the general development, including connection to the 
Neyman-Pearson procedure, handling of nuisance parameters, and a unique approach of getting a data-driven candidate set to improve computational efficiency. Section 4 is a case study on a normal mixture model that addresses a long standing open problem in statistical inference. Section 5 contains further remarks and discussions.

\section{A general inference framework by repro samples}
\label{sec:general}

Let $(\Omega, {\cal F}, \P)$ be a probability space and $\*U \in {\cal U}$ is a measurable $r \times1$ random vector on $\Omega$ with $\{ \*U \leq \*u\} = \{\omega \in \Omega: \*U(\omega) \leq \*u\}$. 
Suppose $\*Y$ is the sample data from (\ref{eq:1}) and the
 $m\times 1$ vector $\*Z= {\bf m}(\*Y)$,  $m \leq n$, is a function (or a summary) of the sample data $\*Y$.
A slightly more general model  than (\ref{eq:1}) is to assume that $\*Z$ is directly generated from 
\begin{equation}
\label{eq:Z}
\*Z = G_z({\btheta}_0, \*U), 
\end{equation} 
where $G_z(\cdot, \cdot) = {\bf m}(G(\cdot, \cdot))$ is a given mapping from $\Theta \times {\cal U} \to 
{\cal Z} \subset 
R^m$.
In the special case that ${\bf m}(\cdot)$ is the identity mapping function,
then $G_z = G$, $\*Z= \*Y$ and model (\ref{eq:Z}) reduces to model (\ref{eq:1}); thus (\ref{eq:1}) is a special case of (\ref{eq:Z}). 

Let $\*z_{obs}$ be the observed $\*Z$ and $\*u^{rel}$ be the (unknown) realization of $\*U$ associated with  $\*z_{obs}$. By (\ref{eq:Z}), we have
\begin{equation*}
\*z_{obs} = G_z({\btheta}_0, \*u^{rel}).  
\end{equation*} 
For any~given $\btheta^*$ and a simulated $\*u^* \sim \*U$, we can thereafter construct an artificial (`fake') copy of $\*z_{obs}$: 
\begin{equation*}
\*z^* = G_z({\btheta}^*, \*u^*). 
\end{equation*} 
We call $\*z^*$ a {\it repro sample} of $\*z_{obs}$ and $({\btheta}^*, \*u^*)$ a {\it repro copy} of $({\btheta}_0, \*u^{rel})$. 

The goal of Section 2.1-2.3 is to present how to use repro samples to construct a confidence set, based on the observed $\*z_{obs}$,  for the unknown parameter $\btheta_0$ with a guaranteed coverage rate. In Section 2.1 we consider the case of giving the algorithmic model (\ref{eq:Z}) where $G_z(\cdot, \cdot)$ is completely known. Section 2.2 contains a generic algorithm
when an explicit mathematical expression of the confidence set is not available. 
Section 2.3 extends the developments to a more general case in which the algorithmic model $G_z(\cdot, \cdot)$ is not completely given. 

\subsection{A general formulation of a basic repro samples method}
 
Under the algorithmic model (\ref{eq:Z}) with a fixed $\btheta_0$, the sampling uncertainty of $\*Z$ is solely determined by the uncertainty of $\*U$, whose distribution is free of the unknown model parameters $\btheta_0$. 
Let us first focus on the easier task 
of quantifying the uncertainty in $\*U$. 
Specifically, let $B_{\alpha} \subset {\cal U}$ be a level-$\alpha$ Borel set such that $\P\{\*U \in B_{ \alpha}\}\ge \alpha.$ 
Although $\*u^{rel}$ is unknown, we have more than $\alpha$-level confidence that $\*u^{rel} \in B_\alpha$. If $\*u^*$ is confined within the same set $B_{\alpha}$ (i.e., requiring $\*u^*  \in B_{\alpha}$) and we 
collect any potential $\btheta$ values that can create a repro sample $\*z^* = G_z(\btheta, \*u^*)$ that matches $\*z^* = \*z_{obs}$, then it leads to a subset on $\Theta$: 
\begin{equation}
\label{eq:G0}
\Gamma_{\alpha}(\*z_{obs}) = \big\{\btheta: \exists \,
\*u^*  \in B_{\alpha} \mbox{ s.t. }  \*z_{obs} = 
G_z({\btheta}, \*u^*)  \big\} \subset \bm \Theta,
\end{equation}
where s.t. is short for such that. In another words, for a potential value $\btheta$, if there exists a $\*u^* \in B_\alpha$ such that the repro sample ${\*z}^* = G( {\btheta}, {\*u}^*)$ matches $\*z_{obs}$ with $\*z^* = \*z_{obs}$, then we keep this $\btheta$ in the set. 
Since $\*z_{obs} = G_z({\btheta}_0, \*u^{rel})$, if 
$\*u^{rel}  \in B_{\alpha}$ then $\btheta_0 \in \Gamma_{\alpha}(\*z_{obs})$. 
Similarly,
under the model $\*Z = G_z({\btheta}_0, \*U)$, we have $\left\{\*U \in B_\alpha\right\} \subseteq \left\{\btheta_0 \in \Gamma_{\alpha}(\*Z)\right\}$. Thus, by construction, 
$$\P\left\{\btheta_0 \in \Gamma_\alpha(\*Z) \right\} \ge \P\big\{\*U\in B_\alpha\big\}  \ge \alpha.
$$ 
The set $\Gamma_{\alpha}(\*z_{obs})$ constructed in (\ref{eq:G0}) 
is a $100\alpha\%$ confidence set for $\btheta_0$. Here, $B_\alpha$ is any single fixed set. As long as  $\P(\*U \in B_\alpha) \ge \alpha$, the validity statement holds. 

To expand the scope of this development, 
the requirement $\P(\*U \in B_\alpha) \ge \alpha$ is replaced~by 
\begin{equation}
    \label{eq:B}
    \P \left\{T(\*U, \btheta)  \in B_{ \alpha}(\btheta) \right\}\ge \alpha,
\end{equation}
for any $\btheta$, where $T(\cdot, \cdot)$ is a mapping function from ${\cal U} \times \Theta \to$ ${\cal T} \subseteq \RR^{q}$, for some $q \leq n$. 
Then, the constructed confidence set in (\ref{eq:G0}) becomes
\begin{equation}
\label{eq:G1}
\Gamma_{\alpha}(\*z_{obs}) = \big\{\btheta: \exists \, \*u^* \in {\cal U} \mbox{ s.t. }  \*z_{obs} = 
G_z({\btheta}, \*u^*),
\,  
T(\*u^*, \btheta)  \in B_{\alpha}(\btheta)  \big\} \subset \Theta. 
\end{equation}
The function $T(\cdot, \cdot)$, which we refer to as a {\it nuclear mapping}, adds much flexibility to our repro-sampling method, as we will show throughout the paper. For the moment,~we assume that the nuclear mapping $T$ is given; further discussions are provided in later sections.

From now on the set $\Gamma_{\alpha}(\*z_{obs})$ refers to the general expression of (\ref{eq:G1}) instead of the special case in (\ref{eq:G0}),  unless specified otherwise. 
In the construction of $\Gamma_{\alpha}(\*z_{obs})$ in (\ref{eq:G1}), for any potential value $\btheta$, we impose two constraints: (a) The repro error $\*u^*$ is confined by the restriction $T(\*u^*, \btheta)  \in B_{\alpha}(\btheta)$; (b) The repro sample $\*z^*(\btheta) = G_z({\btheta}, \*u^*)$ generated using $\bm \theta$ matches the observed sample, i.e.,  $\*z_{obs} =  \*z^*(\btheta)$.  The quantification of inference uncertainty is achieved only through constraint (a) requiring $T(\*u^*, \btheta)  \in B_{\alpha}(\btheta)$ for the given $\btheta$,
and it does not involve the unknown true $\btheta_0$ that generates $\*z_{obs}$. 
The input of the observed sample $\*z_{obs}$ is discernible by matching the artificial repro sample $\*z^*(\btheta) = G_z({\btheta}, \*u^*)$ with $\*z_{obs}$ in constraint~(b).

As an illustration of the method, we consider below a simple binomial example: 

\begin{example} [Binomial sample] 
Assume $Y \sim Binomial(r, \theta_0)$, $0 < \theta_0 < 1$, and we observe $y_{obs}$.
We can re-express the model as $Y = \sum_{i =1}^r {\bf 1}{(U_i \leq \theta_0)}$ for $U_1, \ldots, U_r \sim U(0,1)$.  
It follows that
$y_{obs} =  \sum_{i =1}^r {\bf 1}{(u_i^{rel} \leq \theta_0)},$
where $\*u^{rel} = (u_1^{rel} , \ldots, u_r^{rel} )$ is the realization of $\*U = (U_1, \ldots, U_r)$ corresponding to 
$y_{obs}$. 

For a given vector $\*u = (u_1, \ldots, u_r)$,  we consider the nuclear mapping function
$T(\*u,  \theta) =  \sum_{i =1}^r {\bf 1}{(u_i \leq \theta)}$, and we have $T(\*U,  \theta) =   \sum_{i =1}^r {\bf 1}{(U_i \leq \theta)} \sim Binomial(r, \theta)$, for each given $\theta \in \Theta = (0,1)$. Let $B_\alpha(\theta) = \{\*u \big| a_L(\theta) \leq T(\*u,  \theta) \leq a_U(\theta) \}$, where  
\begin{equation}
    \label{eq:B-bounds}
    \left(a_L(\theta), a_U(\theta)\right) = \underset{\left\{(i,j): 
    \sum^j_{k=i} {r \choose k} \theta^k (1 -  \theta)^{(r-k)} \geq \alpha \right\}}{\arg\min} |j-i|,
\end{equation}
is the pair of $(i,j)$ in $\{(i,j): 
    \sum^j_{k=i} {r \choose k} \theta^k (1 - \theta)^{(r-k)} \geq \alpha\}$ that makes the shortest interval $\left[a_L(\theta), a_U(\theta)\right]$. 
By a direct calculation, it follows that $P\left\{T(\*U,  \theta) \in [a_L(\theta), a_U(\theta)]\right\} \ge  \alpha$, for any $\theta \in (0,1)$.
The confidence set in (\ref{eq:G1}) is then 
\begin{eqnarray*}
\Gamma_{\alpha}(y_{obs}) &=& \bigg\{ \theta \big |  \exists \, \*u^* \in {\cal U} \, \hbox{s.t.} \,\, y_{obs} =  \sum_{i =1}^r {\bf 1}{(u_i^* \leq \theta)},   
 \sum_{i =1}^r {\bf 1}{(u_i^* \leq \theta)}
\in [a_L(\theta), a_U(\theta)] \bigg\}  \nonumber \\  
& = & \left\{ \theta \big |y_{obs} = \sum_{i =1}^r {\bf 1}{(u_{i}^* \leq \theta)}, y_{obs}\in [a_L(
\theta), a_U(\theta)],  \exists \, \*u^* \in {\cal U} \right\}
\nonumber \\ 
& = & \left\{ \theta \big |y_{obs} = \sum_{i =1}^r {\bf 1}{(u_{i}^* \leq \theta)},  \exists \, \*u^* \in {\cal U} \right\} \cap \left\{ 
\theta\big | y_{obs} \in [a_L(\theta), a_U(\theta)]\right\}
\nonumber \\ 
& = & \left\{ \theta \big |a_L(\theta) \leq y_{obs} \leq  a_U(\theta) \right\}. \nonumber 
\end{eqnarray*}
The last equation holds, because for any given $\theta \in \Theta = (0,1)$ there always exists at least a $\*u^* \in {\cal U}$ such that $y_{obs} = \sum_{i =1}^r {\bf 1}{(u_{i}^* \leq \theta)}$, i.e., $\{ \theta \big |y_{obs} = \sum_{i =1}^r {\bf 1}{(u_{i}^* \leq \theta)}, \exists \, \*u^* \in {\cal U} \} = \Theta$. 

Table~\ref{tab:simulation_binary_example}  provides a numerical study comparing the empirical performance of $\Gamma_{.95}(y_{obs})$ (Repro) against 95\% confidence intervals obtained using the traditional Wald method and the fiducial approach (GFI) of \cite{Hannig2009}, both of which are asymptotic methods that can ensure coverage only when $n$ is large. The rerpo sampling method is an exact method that improves the performance of existing methods, especially when $n \theta_0 < 5$ (a rule of thumb on asymptotic approximation of binomial data; cf., \citet[][p.106]{tCAS90a}). 
\end{example}

\begin{table}[ht]
   \centering 
\resizebox{\textwidth}{!}{\begin{tabular}{cccccccccc}
\hline\hline
 && \multicolumn{2}{c}{$n=20, \theta_0=0.1$} && \multicolumn{2}{c}{$n=20, \theta_0=0.4$} && \multicolumn{2}{c}{$n=20, \theta_0=0.8$}\\ \cline{3-4} \cline{6-7} \cline{9-10}
  && Coverage           &    Width      & & Coverage &Width   & &  Coverage    & Width         \\ \hline
 Repro & &     0.949(0.007)      &  0.281(0.059)   &      & 0.963(0.006)           &       0.408(0.026)   & &    0.959(0.006)       &    0.342(0.045)  \\ \hline
 Wald & &        0.877(0.010)  &     0.236(0.052)      & &  0.927(0.008)         &    0.418(0.028)       &  &  0.915(0.009)       &0.332(0.071)\\ \hline
 GFI & &   0.988(0.003)       &  0.293(0.065)        & &          0.963(0.006) &      0.438(0.022)     &  &    0.973(0.004)     &      0.365(0.056)     \\ \hline\hline
\end{tabular}}
    \caption{Comparison of $95\%$ confidence intervals by the repro, Wald and GFI methods in Binomial$(n, \theta_0)$ data; repetitions $= 1000$; standard errors are enclosed in brackets.}
    \label{tab:simulation_binary_example}
\end{table}

The next theorem states that $\Gamma_{\alpha}(\*z_{obs})$ in (\ref{eq:G1}) is a level-$\alpha$ confidence set.

\begin{theorem}\label{thm:1} Assume model (\ref{eq:Z}) holds. If inequality (\ref{eq:B}) holds exactly, then the following inequality holds exactly 
\begin{equation}
\label{eq:v1}
\P\left\{ {\btheta}_0 \in \Gamma_{\alpha}(\*Z)\right\} \ge \alpha 
\,\,\, \hbox{for} \,\, 0< \alpha <1.
\end{equation}
Furthermore, if the inequality (\ref{eq:B}) holds approximately 
 with $\P \big\{T(\*U, \btheta)  \in B_{ \alpha}(\btheta) \big\}\ge \alpha\{1 + o(\delta^{'})\}$, then (\ref{eq:v1}) holds approximately 
 with $\P\left\{ {\btheta}_0 \in \Gamma_{\alpha}(\*Z)\right\} \ge \alpha\{1 + o(\delta^{'})\}$, for $0< \alpha <1$, where $\delta^{'} > 0$ is 
a small value 
that may 
depend on sample size $n$. 
\end{theorem}
We provide a proof of the theorem is in the Appendix. 
The $\delta^{'}$ in Theorem~\ref{thm:1} 
may or may not depend on sample size $n$. In examples involving large sample 
approximations, $\delta^{'}$ is often a function of $n$ with $\delta^{'} \to 0$ as $n \to 0$. However, there are also examples in which $\delta^{'}$ does not involve $n$.
For example, suppose that $Y | \theta = \lambda \sim \text{Poisson}(\lambda)$. Then, $U = (Y -  \lambda)/\sqrt{\lambda} \to N(0,1)$, when $\lambda$ is large. So, 
if we take $T(U, \lambda) = U$, $\P \{T(U, \lambda)  \in B \} =  \int_{{t} \in B}
\phi({t}) d {t} \{1 + o(\lambda^{-1})\}$, for any Borel set $B$, thus $\delta^{'} = o(\lambda^{-1})$. 
Here, 
 $\phi(t)$ is the density function of $N(0,1)$.

Moreover, the repro samples method can be used to do hypothesis testing. From Theorem~1, we have the following corollary. A proof is given in Appendix.

\begin{corollary}\label{col:test}
Suppose we are interested in testing a hypothesis $H_0: \theta \in \Theta_0$ vs $H_1: \theta \not\in \Theta_0$. In that case, we can define a p-value $p(\*z_{obs})$ as
\begin{align*} 
p(\*z_{obs})=   1- \inf_{\theta \in \Theta_0}\left[\inf\left\{\alpha': \theta \in \Gamma_{\alpha'}(\*z_{obs})\right\}\right],
\end{align*}
where $\Gamma_{\alpha'}(\*z_{obs})$ is the level-$\alpha'$ repro samples confidence set defined in \eqref{eq:G1}.
Rejecting $H_0$ when
$p(\*z_{obs}) \leq \gamma$ leads to a size $\gamma$ test, for any $0 < \gamma < 1$. 
\end{corollary}

Although the repro samples method is developed following the frequentist principles, thus making it a completely  frequentist approach, it is also closely related to several existing approaches across the Bayesian, fiducial, and frequentist (BFF) paradigms. 
In Remarks 2.1 - 2.3 below, we provide discussions to connect and compare these approaches. See also Appendix II for further and more elaborated discussions. 

{\bf Remark 2.1} (comparison with Neyman-Pearson testing method). Corollary~\ref{col:test} 
suggests that we can conduct hypothesis testing using the repro samples method. 
Section~\ref{sec:T-1} 
further investigates the relationship between the repro samples method and the classical Neyman-Pearson (N-P) hypothesis testing procedure
under a special case where a nuclear mapping $T(\*u, \theta)$ in \eqref{eq:G1} is chosen through a test statistic, i.e., we set $T(\*u, \theta) = \widetilde T(\*z', \theta)$ where $\*z' = G_z(\*u, \btheta)$ is a copy of data 
generated by the given $\btheta$ and $\widetilde T(\*z', \theta)$ is a test statistic of a hypothesis testing problem. 
In general, the repro samples approach is broader than the classical N-P procedure.
In particular, the nuclear mapping $T(\*u, \btheta)$ does not need to be a test statistic since it is not necessarily a function of  data $\*z'.$ 
Allowing the nuclear mapping to be a function of $\*u$ provides more flexibility for developing inference procedures than the conventional testing approach.
We illustrate this point using 
a simple example where a single $\text{Bernoulli}(\theta)$ observation $y_{obs}$ is generated by $Y= I(U<\theta_0)$, $U \sim \text{Uniform}(0,1).$
To get a confidence set for $\theta_0$, we can let 
$T(u, \theta)= u$, 
and define $ B_\alpha(\theta) = 
    \begin{cases}
    [1-\alpha, 1] & \theta  > 0.5 \\
    [\frac{1-\alpha}{2},\frac{1 + \alpha}{2} ] & \theta =0.5 \\
    [0, \alpha] & \theta  < 0.5
    \end{cases}.
$ 
Then, following~\eqref{eq:G1}, we construct the repro samples confidence set  
 $   \Gamma_\alpha(y_{obs}) = 
    \begin{cases}
    [1-\alpha, 1] & y_{obs}= 1\\
    [0, \alpha] & y_{obs} =0
    \end{cases}.
$ 
Apparently, $U$ is not a function of $Y$, since it is impossible to solve for $U$ in the structure equation $Y= I(U<\theta)$ when given $(Y, \theta)$. 
In this case, $T(U, \theta)= U$ is not a test statistic. 
Section~\ref{sec:T-1} provides further discussion and additional examples comparing the repro samples methods and the classical N-P method. We show that the repro samples method can always produce an either equal or smaller confidence set than the traditional N-P method when the nuclear mapping $T(\*u, \theta)$ is chosen through a test statistic. There are also options in which the nuclear mapping function $T(\*u, \theta)$ is not a test statistic.

{\bf Remark 2.2} (comparison with GFI and ABC methods) In modern statistical practice, R.A. Fisher’s fiducial method is understood as an inversion method
that solves the structural equation $\*z_{obs} = G_z({\btheta}, \*u^*)$ for parameter ${\btheta}$ with any $\*u^* \sim \*U$
\citep[cf.,][]{brenner1983, Hannig2016, Thornton2022}.  The matching equation $\*z_{obs} = G_z({\btheta}, \*u^*)$ in (\ref{eq:G1}) plays a key role in the GFI development 
\citep{Hannig2016}. 
\cite{Hannig2016} and  \cite{Thornton2022} have also explored the connection of GFI to the Bayesian ABC method, noting that trying to match $\*z_{obs}$ with $\*z^*(\btheta) = G_z({\btheta}, \*u^*)$ is a  key step in both GFI and ABC procedures. However, as stated in 
\cite{Hannig2016}, an
exact matching of $\*z_{obs}  = \*z^*(\btheta)$ for any given $\*u^*$ is difficult and sometime even impossible. Both GFI and ABC methods adopts a tuning threshold value $\epsilon$ to judge whether an artificial sample $\*z^* = G_z({\btheta}, \*u^*)$
is close to the observed $\*z_{obs}$.  
 Unfortunately, the use of approximating threshold $\epsilon$ has piratical issues. 
 For instance, as shown in
 \cite{Li2016}, 
the validity of the ABC method requires that $\epsilon \to 0$ at a fast rate, but this fast rate in turn leads to inefficient sampling with a high (often 100\%) rejection rate. 
How to choose an appropriate $\epsilon$ to balance  inferential validity and computational efficiency
is still an unaddressed question in both ABC and GFI.
Our repro sample approach avoids this issue of using a threshold $\epsilon$ by directly working with a set of $\*u^* \in \{T(\*u^*, \btheta) \in B_\alpha(\theta)\}$ and multiple copies of $\*z^* = G_z({\btheta}, \*u^*)$ for each given $\btheta$. 
As further explained in Appendix II
the repro samples method in effect compares $\*z_{obs}$ with multiple copies of $\*z^* = G_z({\btheta}, \*u^*)$ at each given $\btheta$ and thus enabling us to use a confidence (significant) level $\alpha$ to replace the arbitrary $\epsilon$. 
In addition, the repro sample method can provide finite sample inference whereas GFI does not. Finally, ABC is a Bayesian procedure that not only requires additional prior assumption, but also requires a sufficient summary statistic in its algorithm.
The repro samples approach compares favorably in that it does not have these constraints.  

{\bf Remark 2.3} (comparison with
IM method).  
Inferential model (IM) \citep{Martin2015} also works 
directly on $\*U$ to quantify inference uncertainty, but contains notable differences from the proposed repro samples method. First and foremost, the IM method exists to provide an epistemic  probabilistic inference for $\btheta_0$, a task not considered under the current frequentest framework; cf.,  \citet{Martin2015}. The repro samples on the other hand is a frquentist approach developed fully within the frequentist framework. Also importantly, the IM method uses 
random sets to quantify the uncertainty of the auxiliary term $\*U$. But the repro samples method quantifies inference uncertainty through a nuclear mapping function $T(\*U, \btheta)$ that involves a potential parameter $\btheta$ and an associated Borel set. Only one fixed level-$\alpha$ Borel set is needed.  Although IM can sometimes potentially provide some finer details in special cases (since Dempster-Shefer calculus automatically includes both upper and lower probabilities; i.e., plausible and belief functions), with additional constraints, the application of IM is more limited than repro samples method, especially under more complicated settings. Moreover, using a nuclear mapping function involving $\btheta$ greatly extends the scope and flexibility of the repro samples method. As a special case, we are able to directly work on auxiliary $\*U$ by setting $T(\*U, \btheta) = \*U$. To the best of our knowledge, the framework of IM method cannot be easily extended to incorporate varying $\btheta$. Furthermore, IM development requires Dempster-Shefer calculus and asymptotic inference is not available under the development. When both apply, the IM and repro samples methods provide either the same or comparable inference conclusions.
Finally, the repro samples development 
affords us the opportunity to explore the unique aspect of matching to obtain a data-driven candidate set and greatly reduce computational burden;
details are provided in Section 3.3. On the other hand, how to incorporate a data-driven candidate set with the IM method to reduce computational cost is not apparent.

\subsection{A general Monte-Carlo algorithm}

In Example 1, the distribution of $T({\*U}, \btheta)$ has an explicit form, enabling us to express the level-$\alpha$ confidence set $\Gamma_{\alpha}(\*y_{obs})$ in (\ref{eq:G1}) explicitly. In cases when the distribution of $T({\*U}, \btheta)$ is not explicit, we can use a Monte-Carlo method to obtain a level-$\alpha$ Borel set $B_\alpha(\btheta)$ in (\ref{eq:B}) and then construct the level-$\alpha$ confidence set $\Gamma_{\alpha}(\*y_{obs})$ in (\ref{eq:G1}) as stated in Algorithm 1 below:

\begin{algorithm} \caption{Monte-Carlo algorithm (general case)}
\label{alg:Ag}
\begin{algorithmic}[1]
\State Simulate many copies of  ${\*u^s} \sim \*U$; the collection of ${\*u^s}$ forms a set ${\cal V}$.
\State For each given (grid) value of $\btheta \in \Theta$, 
\newline 
(a) Compute  $\{T({\*u^s}, \btheta), {\*u^s} \in {\cal V}\}$, from which obtain an empirical distribution of $T(\*U, \btheta)$. Based on the empirical distribution, we get a level-$\alpha$ Borel set $B_{ \alpha}(\btheta)$. 
\newline (b) Check whether there exists a $\*u^* \in {\cal U}$ such that 
${\*z_{obs}} = G_z({\btheta}, \*u^*)$
and $T(\*u^*, \btheta)  \in B_{\alpha}(\btheta)$. If both of the above criteria are satisfied, keep the $\btheta$.
\State Collect all kept $\btheta$ to form the set $\Gamma_{\alpha}(\*y_{obs})$. 
\end{algorithmic}
\end{algorithm}

This general algorithm, 
grid search algorithm, searches the space of $\Theta$ and is
used only when an explicit expression of $T(\*U, \btheta)$'s distribution
is not available. 
Our remarks on the implementation of the algorithm are as follows.

{\bf Remark 2.4} 
To obtain the level-$\alpha$ Borel set $B_{\alpha}({\btheta})$ for a given $\btheta$ in Step 2(a), w often need to compute the (empirical) distribution of $T(\*U, {\btheta})$ based on the set of Monte-Carlo points ${{\cal S}_{\theta}} = \{T(\*u^s, {\btheta})$, $\*u^s \in {\cal V}\}$. Then according to the empirical distribution, we can always numerically calculate $B_{\alpha}({\btheta})$.  If $T(\*u, {\btheta})$ is a scalar, we just directly compute its empirical distribution $F_{\cal V}(t) = $ $\sum_{\*u^s \in {\cal V}} I\big(T(\*u^s, {\btheta}) \leq t \big)\big/|{\cal V}|$ and use the (upper/lower) quantiles of $F_{\cal V}(t)$ to obtain a level-$\alpha$ interval as the Borel set $B_{\alpha}({\btheta})$. When $T(\*u, {\btheta})$ is a vector, 
say in ${\cal T} \subseteq \RR^q$,
we can use data depth
and follow \citet[][\S 3.2.2]{Liu2021}
to construct 
the Borel set $B_{\alpha}({\btheta})$. Specifically, we let $D_{{\cal S}_{\theta}}(\*t)$ be the empirical depth function for any $\*t \in \cal T$ that is computed based on the Monte-Carlo points in ${\cal S}_{\theta} = \{T(\*u, {\btheta})$, $\*u \in {\cal V}\}$  and $F_{{\cal V}|D} (t) =$ $\sum_{\*u^s \in {\cal V}} I\big(D_{{\cal S}_{\theta}}(T(\*u^s, {\btheta})) \leq t \big)\big/|{\cal V}|$ be the empirical CDF of $D_{{\cal S}_{\theta}}(T(\*u^s, {\btheta})).$  
Then, following Lemma 1 (a) of \cite{Liu2021}, the level-$\alpha$ (empirical) central region 
$B_\alpha({\btheta}) = \{\*t: F_{{\cal V}|D}(D_{{\cal S}_{\theta}}(\*t)) \geq 1 - \alpha\}$ is a level-$\alpha$ Borel set on ${\cal T}$. See also Example 2 (B) of Section 2.3 for a numerical example, where the empirical depth function $D_{{\cal S}_{\theta}}(\cdot)$ is computed by R package {\it ddalpha} \citep{pokotylo_depth_2019}.

{\bf Remark 2.5} The algorithm needs to calculate the empirical distribution of $T(\*U, \btheta)$ for each $\btheta$.  
To improve the computing efficiency (at the cost of potentially losing the flexibility of the repro samples method), we can sometimes  choose a nuclear mapping function $T(\*U, \btheta)$ whose distribution of is free (or approximately free) of $\btheta$. 
In this case,  $B_\alpha(\btheta )$ in (\ref{eq:B}) is free (or approximately free) of $\btheta$ with $B_\alpha(\btheta ) \equiv B_\alpha$. 
A benefit of using this type of nuclear mapping function 
is that we can modify Algorithm~\ref{alg:Ag} to save computing time: 
in step 2(a), we only need to obtain $B_{\alpha}(\btheta)$ for a single value of $\btheta$ instead of for each grid value of $\btheta \in 
\Theta.$

{\bf Remark 2.6} 
Later in Section 3.3, we will discuss the technique of using a data-driven candidate set to significantly reduce the search space. From the outset, we would like to comment that 
the search algorithm works in most cases that one commonly encounter.
In many statistical analysis, for instance the commonly used likelihood-based inference, the regular conditions on the parameter space include that $\Theta$ is a continuous space and the test statistic is also continuous (and often differentiable) in its inner space  $\Theta^o$. Similarly, 
if the nuclear mapping function $T(y_{obs}, \btheta)$ is continuous in $\btheta$ and $\Theta$ is a continuous space, one can in principle use a grid search method to carry out an analysis in this situation. What is new in the repro samples development is that $\Theta$ can either be a discrete space or a set of non-numerical subjects. In this case, even if $\Theta$ is (countable) infinite, we provide in Section~3.3 a way to take advantage of a many-to-one mapping structure inherited in the repro samples procedure to create a finite candidate set and carry out an efficient grid search method. See Section~3.3 for further details.

\subsection{A further extension to more general cases}

In this section, we further extend the framework developed above to settings where the matching equation \eqref{eq:Z} of the algorithmic model is not available. This includes some nonparametric inference problems that are not covered by the existing ABC or GFI methods. 
Consider the example of making  inference for a population quantile of a completely unknown distribution $F$, say, $\theta_0 = F^{-1}(\zeta)$, the $\zeta$-th quantile parameter of $F$. 
Since $I(Y < \theta_0) \sim \text{Bernoulli}(\zeta)$ for $Y \sim F$, we have equation
$I(Y < \theta_0) = U$, where $U \sim \text{Bernoulli}(\zeta)$. Now suppose we observe data $\*y_{obs} = (y_1 \ldots, y_n)^\top$, for a fixed $n$, then the equation becomes 
\begin{equation} \label{eq:quantile}
    \sum_{i=1}^{n} I(Y_i - \theta_0 <0) = \sum_{i=1}^n U_i 
\end{equation}
where $Y_i \overset{iid}{\sim} F$ and $U_i \overset{iid}{\sim} \text{Bernoulli}(\zeta)$. 
The corresponding realized version is
$\sum_{i=1}^{n} I(y_i - \theta_0 <0) = \sum_{i=1}^n u_i^{rel}$. Equation (\ref{eq:quantile}) can not be expressed in the form of either (\ref{eq:1}) or (\ref{eq:Z}). In fact, one cannot generate a repro sample $\*y^*$ based on (\ref{eq:quantile}) that is directly comparable to the observed sample $\*y_{obs}$. However, we can still use a repo sampling approach to provide an exact inference for such a nonparametrc inference problem.

In particular, we
consider the following generalized data generating equation: 
\begin{equation}
    \label{eq:AA}
    g(\*Z, \btheta_0, \*U) = \*0,
\end{equation}
where $g$ is a given mapping function from ${\cal Z} \times \Theta \times {\cal U}  \to \R^s$ and $\*Z = {\bf m}(\*Y)$ is a $m\times 1$ vector as defined at the start of Section 2.
The corresponding realization version is
  $g(\*z_{obs}, \btheta_0, \*u^{rel}) = \*0$.
Equation (\ref{eq:quantile}) 
is a special case of (\ref{eq:AA}) with $g(\*Y, \btheta_0, \*U) = \sum_{i=1}^{n} I(Y_i - \theta_0 <0) - \sum_{i=1}^n U_i$.

The level-$\alpha$ confidence set in (\ref{eq:G1}) is  modified to be
\begin{eqnarray}
\label{eq:G2}
\Gamma_{\alpha}(\*z_{obs}) = \big\{\btheta:  \exists 
\,\*u^* \in {\cal U} \mbox{ s.t. } 
 g(\*z_{obs}, \btheta, \*u^*) = 0, T(\*u^*, \btheta)  \in B_{\alpha}(\btheta)
\big\} \subset \Theta.
\end{eqnarray}
We have the following theorem for the set $\Gamma_{\alpha}(\*z_{obs})$ defined in (\ref{eq:G2}). The proof is similar to that of Theorem~\ref{thm:1} and is briefly described in Appendix I. 
\begin{theorem}\label{thm:2} Assume the generalized model equation (\ref{eq:AA}) holds. If the inequality (\ref{eq:B}) holds exactly, then for $\Gamma_{\alpha}(\*z_{obs})$ defined in (\ref{eq:G2}) the following inequality holds exactly, 
\begin{equation}
\label{eq:V2}
\P\left\{ {\btheta}_0 \in \Gamma_{\alpha}(\*Z)\right\} \ge \alpha 
\quad
 \hbox{for $0< \alpha <1$.}
\end{equation}
 Furthermore, if (\ref{eq:B}) holds approximately 
 with $\P \left\{T(\*U, \btheta)  \in B_{ \alpha}(\btheta) \right\}\ge \alpha\{1 + o(\delta^{'})\}$, then (\ref{eq:V2}) holds approximately 
 with $\P\left\{ {\btheta}_0 \in \Gamma_{\alpha}(\*Z)\right\} \ge \alpha\{1 + o(\delta^{'})\}$, for $0< \alpha <1$, where $\delta^{'} > 0$ is 
a small value 
that may 
depend on sample~size~$n$. 
\end{theorem}

As discussed in Sections 2.2,  
when the distribution of $T(\*u^*, \btheta)$ cannot be explicitly expressed, we may use a Monte-Carlo algorithm to carry out the construction of $\Gamma_{\alpha}(\*y_{obs})$. Specifically, we may still use Algorithm~\ref{alg:Ag} to get $\Gamma_{\alpha}(\*y_{obs})$ but replace the matching equation ${\*z_{obs}} = G_z({\btheta}, \*u^*)$ in the algorithm with $g(\*z_{obs}, \btheta, \*u^*) = 0$.

The example below provides an illustration of this extended framework. Part (A) is on nonparametric quantile inference and part (B) is on 
censored quantile regression where making a finite-sample (semi-parametric) inference for model parameters is new and previously unavailable in the literature. Both parts demonstrate that 
the repro samples method performs well and is compared favorably to the corresponding large sample bootstrap methods.

\begin{example}[Nonparametric and semi-parametric inference of quantiles] \label{ex:crq} (A) 
Assume $\*y_{obs} = (y_1, \ldots, y_n)^\top$ are from an unknown distribution $F$. We are interested in making inference about the population quantile $\theta_0 = F^{-1}(\zeta)$, for a given $0<\zeta <1$. 
Write
$g(\*Y, \theta, \*U) = \sum_{i = 1}^n I(Y_i \leq \theta) - \sum_{i =1}^n U_i$, where $Y_i \sim F$ and $U_i \sim \text{Bernoulli}(\zeta)$. 
We define $T(\*U, \theta) =  T(\*U) = \sum_{i =1}^n U_i$. It follows that $T(\*U) \sim \text{Binomial}(n, \zeta)$ and $B_\alpha = [a_L(\zeta), a_U(\zeta)]$, where $a_L(\cdot)$ and $a_U(\cdot)$ are defined in (\ref{eq:B-bounds}). By (\ref{eq:G2}), a finite sample level-$\alpha$ confidence set is
$\Gamma_{\alpha}(\*y_{obs}) = \big\{\theta: \sum_{i = 1}^n I(y_i \leq \theta) - \sum_{i =1}^n u_i = 0, 
\sum_{i =1}^n u_i \in [a_L(\zeta), a_U(\zeta)]
\big\} = \big\{\theta: a_L(\zeta) \leq \sum_{i = 1}^n I(y_i \leq \theta) \leq a_U(\zeta))
\big\} = \big[y_{(a_L(\zeta))}, y_{(a_U(\zeta)+1)}\big),
$
where 
$y_{(k)}$ is the $k$th order statistic
of the sample $\{y_1, \ldots, y_n\}$,  
$y_{(0)}$ and $y_{(n+1)}$ are the infimum and supremum of support of $F.$ 

Figure~\ref{fig:quantile} provides a numerical study to compare the empirical performance of $\Gamma_{.95}(y_{obs})$ 
against 95\% intervals obtained by the conventional bootstrap method under two types of $F$.
The repro samples method works~well even for $\zeta$'s close to $0$ or $1$, whereas the bootstrap method has coverage issues under these settings.

(B) Consider a censored regression model 
\citep{powell_censored_1986}, $Y_i = \max\{0,$ ${\*x}^\top_i \btheta_0 + \epsilon_i\}$, $i = 1, \ldots, n,$ where $\epsilon_i$ are independent error terms from an unknown distribution $F$ with 
medium $0$. Following \cite{bilias_simple_2000}, the estimating equation for performing a quantile regression is $\sum_{i=1}^n \*x_i \big\{I(Y_i - {\*x}_i^T \btheta  \leq 0) - \zeta\big\}  I(\*x^T_i\btheta > 0 ) = 0$,  
where $\zeta$ is the quantile level and 
the parameter of interest is $\btheta_0(\zeta) = (F^{-1}(\zeta), 0, \ldots,0)^\top + \btheta_0$. 
Since $I(Y_i -  {\*x}_i^T \btheta_0(\zeta)  \leq 0) \sim \text{Bernoulli}(\zeta)$, equation \eqref{eq:AA} can be expressed as
$
g(\*Y, \btheta, \*U) = \sum_{i=1}^{n} {\*x}_i I{({\*x}_i^T \btheta > 0)} \big\{I(Y_i -  {\*x}_i^T \btheta \leq 0)  -  I(U_i \leq \zeta)\big\},
$
where $U_i \sim U(0,1).$ We consider a nuclear mapping function $T(\*U, \btheta) = \big\{\sum_{i=1}^n \*x_i \*x_i^T I({\*x}_i^T \btheta > 0)\big\}^{-1}\sum_{i=1}^n {\*x}_i I{({\*x}_i^T \btheta > 0) I(U_i \leq \zeta)}$. 
Since the distribution of $T(\*U, \btheta)$ does not have an explicit expression, we use Algorithm 1 to obtain a level-$\alpha$ confidence set for $\btheta_0(\zeta)$. 
Here, $T(\*U, \btheta)$ is a vector (of the same length of $\btheta$) and 
$B_\alpha(\btheta)$ is obtained using the data-depth approach described in Remark 2.4. Table 2 provides a numerical study to compare $\Gamma_\alpha(\*y_{obs})$ in (\ref{eq:G2}) with the confidence sets obtained using the enhanced bootstrap method of \cite{bilias_simple_2000}. Again, the repro samples method provides level-$95\%$ confidence sets with finite sample coverage guarantee and it works well for this difficult inference problem even at a small $n = 75$, while the asymptotic bootstrap method 
has coverage issues under the same settings for $\zeta=0.5$. For $\zeta=0.9,$ both methods provide good coverage rates, but the bootstrap confidence sets are much larger.

\begin{table}[]
\centering
\resizebox{\textwidth}{!}{\begin{tabular}{llccc}
\hline
                 Quantile & Error Distribution & Repro Samples  &  Bootstrap & Relative Volume \\ \hline
\multirow{3}{*}{$\zeta = 0.50$} & (a) Normal  & 0.945(0.007)  & 0.885(0.010)         & 0.779(0.028)      \\
                  & (b) Mixture & 0.927(0.008)         & 0.891(0.010)          & 0.784(0.040)         \\
                  & (c) Heteroscedastic & 0.943(0.007)          & 0.893(0.010)          & 0.762(0.008)      \\
                  \hline
\multirow{3}{*}{$\zeta = 0.90$} & (a) Normal  & 0.945(0.007)          & 0.970(0.005)          & 17.518(2.723)         \\
                  & (b) Mixture &   0.949(0.007)          & 0.977(0.005)       & 214.962(9.595)         \\ 
                  & (c) Heteroscedastic &0.951(0.007)          & 0.980(0.004)         & 346.202(11.403)  \\ \hline
\end{tabular}}
\caption{Coverage rates and relative volumes of repro samples method versus the bootstrap method by \cite{bilias_simple_2000}. 
The study is based on model
   $y_i = \max\{0.5 + x_{i1} + x_{i2} + \epsilon_i, 0\},$
where $x_{i1} \sim \text{Bernoulli}(0.5)$ and $x_{i2} \sim N(0,1),$ and $n=75.$ The $\epsilon_i$ are generated from: (a) normal distribution $\epsilon_i \sim N(0,1),$ (b) mixture normal $\epsilon_i \sim w_iN(0,1) + (1-w_i) N(0,2), w_i \sim \text{Bernoulli}(0.75)$ and (c) heterocedastic distribution $\epsilon_i \sim (1+0.15x_{i2})N(0,1).$ 
The relative volume of the two three dimensional confidence sets $C_{repro}$ and $C_{boot}$ is defined as $\frac{\text{the proportion of points in $C_{repro}$ that are also in $C_{boot}$}}{\text{the proportion of points in $C_{boot}$ that are also in $C_{repro}$}}$. 
Each setting is repeated 1000 times.
}
\label{tab:crq_joint}
\end{table}

\end{example}

\begin{figure}[ht]
    \centering
    \begin{subfigure}[b]{\textwidth}
    \centering
      \includegraphics[width=\textwidth, height= 0.26\textheight]{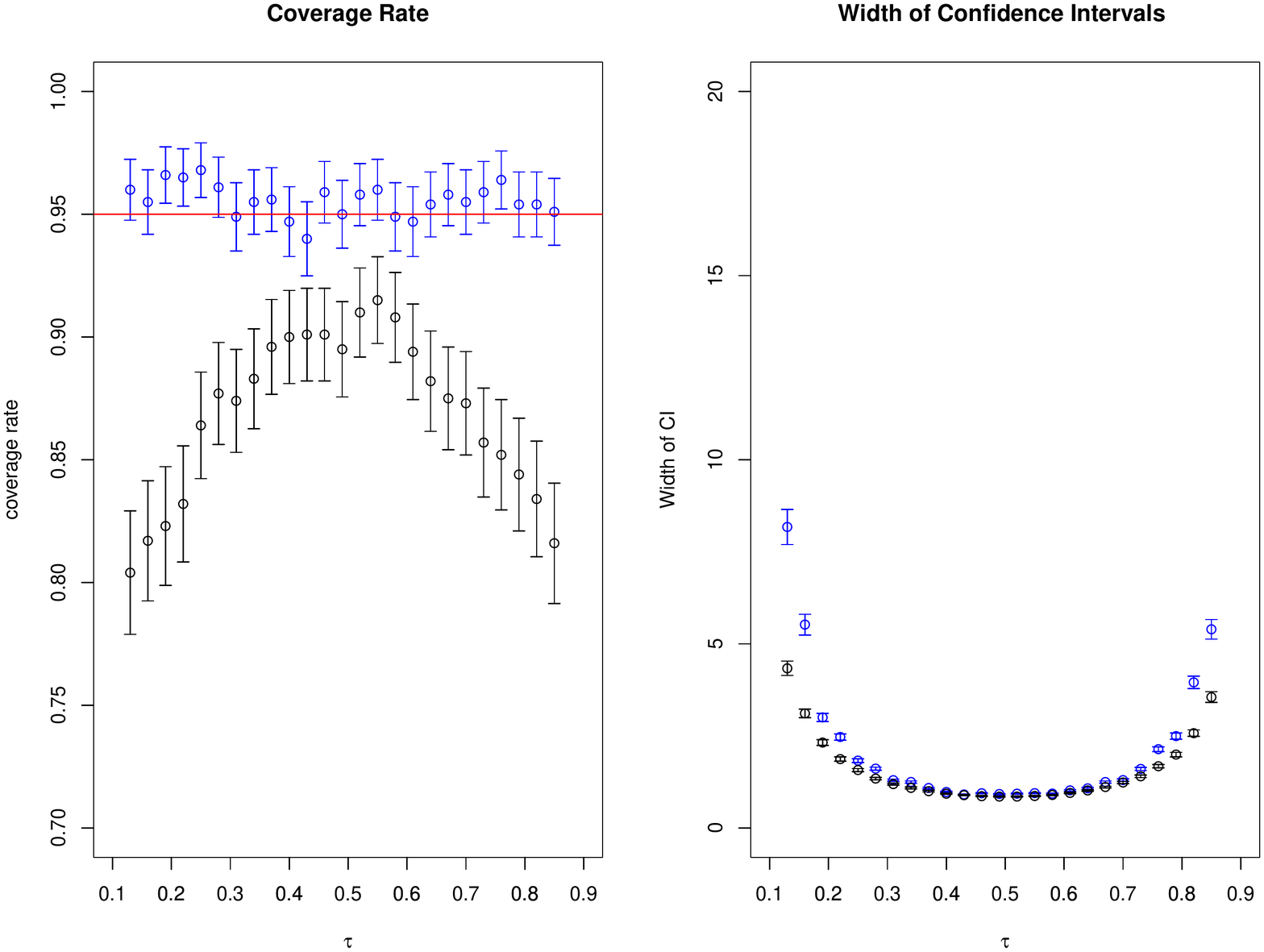}
      \caption{Cauchy Distribution}
    \end{subfigure}
   \begin{subfigure}[b]{\textwidth}
   \centering
    \includegraphics[width=\textwidth,  height= 0.26\textheight]{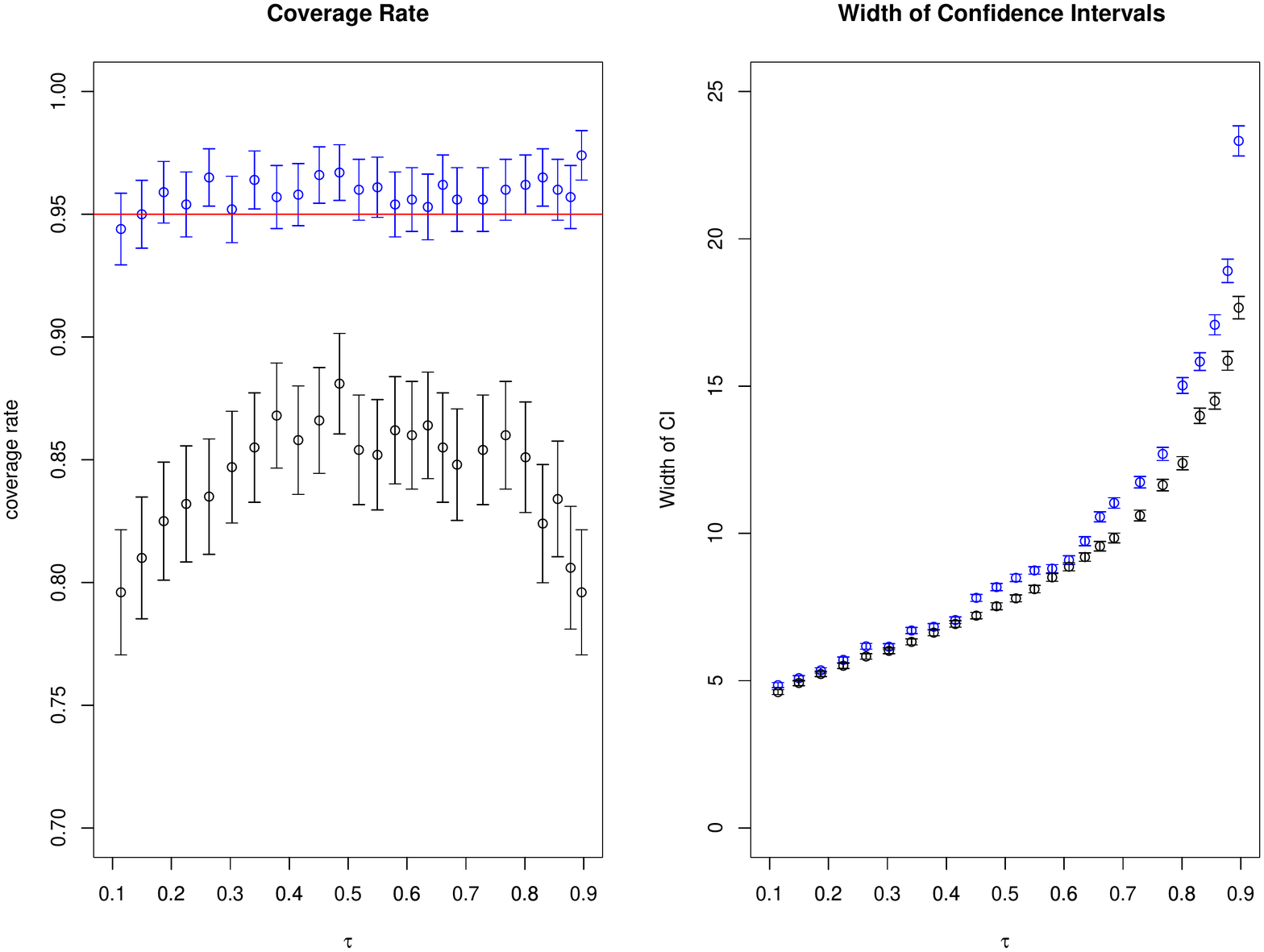}
    \caption{Negative Binomial Distribution}
   \end{subfigure}
    \caption{Coverage rates and widths of 95\% confidence intervals for $\theta_0 = F^{-1}(\zeta)$ of an (unknown) distribution $F$ for a range of percentile $\zeta$ values. Bars and dots colored in blue are produced by the repro samples method and those in black are 
    by the conventional bootstrap method. In (a) $F$ is $\text{Cauchy}(0,1)$ and in (b) $F$ is $\text{NB}(2, 0.1)$. Both cases have sample size $n=60.$ 
    The simulations are replicated for $1000$ times. The error bars display the mean~$\pm 2$~standard~errors.}
    \label{fig:quantile}
\end{figure}

\section{Test statistic, nuisance parameters and candidate set}\label{sec:T}

For any nuclear mapping function $T$, as long as we have a set $B_\alpha(\btheta)$ such that (\ref{eq:B}) holds, 
the set $\Gamma_{\alpha}(\*z_{obs})$ constructed accordingly is a level-$\alpha$ confidence set for $\btheta$. In other words, it is possible that one might have multiple choices of $T$ that all lead to valid (and perhaps different) confidence sets.
Therefore, the role of $T$ under the repro samples framework is similar to that of a test statistic under the classical (Neyman-Pearson) hypothesis testing framework. 
Accordingly, a good choice for $T$ is problem specific and we do not have a unique solution to all situations.
In this section, we investigate the choice of the nuclear mapping function $T$ in two specific yet common settings.
Specifically, in Section 3.1, we consider the special case of using a test statistic, including 
an (approximate) pivot quantity such as a likelihood ratio test statistic, as a nuclear mapping function. We provide comparisons to Neyman-Pearson testing methods.
In Section 3.2, we consider
the situation of existing nuisance parameters and provide a method to define a nuclear mapping function while minimizing the impact of the nuisance parameters. 

To improve the efficiency of 
computing Algorithms~\ref{alg:Ag},
we also discuss the use of a pre-screening candidate set in Section 3.3. 
We pay special attention  to 
discrete target parameters, a unique and also important development under the repro samples framework.

\subsection{Define a nuclear mapping through a test statistic}\label{sec:T-1}

In a classical testing problem $H_0: \btheta_0 = \btheta$ versus $H_1: \btheta_0 \not = \btheta$, one often constructs a test statistic, say $\widetilde T(\*z_{obs}, \btheta)$, and assume that we know the distribution of $\widetilde T(\*Z', \btheta)$ for $\*Z' = G_z(\btheta,\*U)$ generated under $H_0: \btheta_0 = \btheta$. 
Under the Neyman-Pearson framework, we can~derive  
a level-$\alpha$ acceptance region $A_\alpha(\btheta) = \big\{\*z_{obs} \big| \widetilde T(\*z_{obs}, \btheta) \in  B_\alpha(\btheta)\big\}$, where the set $B_\alpha(\btheta)$ satisfies $\P\big\{\widetilde T(\*Z', \btheta) \in  B_\alpha(\btheta)\big\} \ge \alpha$. By the property of test duality, a level-$\alpha$ confidence set is
\begin{equation*}
\widetilde \Gamma_\alpha(\*z_{obs}) =  \left\{\btheta:  \*z_{obs} \in  A_{\alpha}(\btheta) \right\}
= \left\{\btheta:  \widetilde T(\*z_{obs}, \btheta)  \in B_{\alpha}(\btheta) \right\}. 
\end{equation*}

Suppose the nuclear mapping function $T(\*u, \btheta)$ is defined through this test statistic:
\begin{equation} \label{eq:tT1}
T(\*U, \btheta) = \widetilde T(\*Z', \btheta) = \widetilde T(G_z({\btheta}, \*U), \btheta), 
\end{equation} 
where $\*Z' = G_z(\btheta,\*U)$. Since $\P\left\{T(\*U, \btheta) \in  B_\alpha(\btheta)\right\} = \P\big\{\widetilde T(\*Z', \btheta) \in  B_\alpha(\btheta)\big\} \ge \alpha$ for the same $B_\alpha(\cdot)$, a level-$\alpha$ confidence set of (\ref{eq:G1}) by the repro samples method is 
\begin{eqnarray*}
\Gamma_{\alpha}(\*z_{obs}) 
 =  \left\{\btheta: \exists \*u^* \in {\cal U} \mbox{ s.t. }  \*z_{obs} = G_z({\btheta}, \*u^*), \widetilde T(G_z({\btheta}, \*u^*), \btheta)  \in B_{\alpha}(\btheta)  \right\}.
\end{eqnarray*}
 We have the following lemma suggesting $\Gamma_{\alpha}(\*z_{obs}) \subseteq  \widetilde \Gamma_{\alpha}(\*z_{obs})$. See Appendix I for a proof. 
 \begin{lemma}\label{Lemma:1}
 If the nuclear mapping function is defined through a test statistic as in (\ref{eq:tT1}), then we have $\Gamma_{\alpha}(\*z_{obs}) \subseteq  \widetilde \Gamma_{\alpha}(\*z_{obs})$. The two sets are equal  
$\Gamma_{\alpha}(\*z_{obs}) = \widetilde \Gamma_{\alpha}(\*z_{obs})$, when $\widetilde \Gamma_{\alpha}(\*z_{obs}) \subseteq \big\{\btheta:  \*z_{obs} = G_z({\btheta}, \*u^*), 
\exists \, \*u^* \in {\cal U} \big\}$.
 \end{lemma}
 
It is possible that $\Gamma_{\alpha}(\*z_{obs})$ is strictly smaller than 
$\widetilde \Gamma_{\alpha}(\*z_{obs})$ with $\Gamma_{\alpha}(\*z_{obs}) \subsetneq \widetilde \Gamma_{\alpha}(\*z_{obs})$. The following is an example in which $\Gamma_{\alpha}(\*z_{obs})$ is strictly smaller than 
$\widetilde \Gamma_{\alpha}(\*z_{obs})$. 

\begin{example}\label{ex:3}
  Suppose $\*y_{obs} = (y_1, \ldots, y_n)^\top$ is a set of sample from the model $Y_i = \theta_0 + U_i$, $U_i \sim U(-1,1)$, $i = 1, \ldots, n$. Here, $\Theta = (-\infty, +\infty)$, ${\cal Y} = (-\infty, \infty)$. 
 The point estimator $\bar Y = \frac 1n \sum_{i=1}^n Y_i$ is unbiased (and also $\sqrt{n}$-consistent). Let's
 use the test statistic $\widetilde T(\*y_{obs}, \theta) = \bar y  - \theta$. Since $n \{(\bar Y  - \theta) + 1\}/2 = \sum_{i=1}^n \frac{U_i + 1}2$ follows a Irwin-Hall distribution when $\theta $ is the true value $\theta_0$,  
 by the classical testing method a level-$95\%$ confidence interval is $\widetilde \Gamma_{.95}(\*y_{obs}) = (\bar y - \frac2n q_{.975} + 1, \bar y + \frac2n q_{.975} - 1)$. Here, $q_{.975}$ is the $97.5\%$ quantile of the Irwin-Hall distribution. 
 If we use the repro samples method, our confidence set is then
 $\Gamma_{.95}(\*y_{obs}) = \big\{\theta:   y_{i} = \theta +  u_i^*,
 i = 1, \ldots, n; \, \bar y \in (\theta - \frac2n q_{.975} + 1, \theta + \frac2n q_{.975} - 1);
\exists  \*u^* \in (-1, 1)^n \big\} = \left\{\theta: \theta \in \cap_{i=1}^n (y_i -1, y_i+1) \right\} \cap \widetilde \Gamma_{.95}(\*y_{obs})$. Since
the constraint $\{\theta \in \cap_{i=1}^n (y_i -1, y_i+1)\}$ is non-trivial,
the strictly smaller statement $\Gamma(\*y_{obs}) \subsetneq \widetilde \Gamma(\*y_{obs})$ holds for  non-trivially many~realizations of $\*y_{obs}$.
 
As an illustration, consider the case $\theta = 0$ and $n = 3$, in which $q_{.975} = 2.4687$. Suppose we have a sample realization $\*y_{obs} = (-0.430, 0.049, 0.371)^\top$. By a direct calculation,  
$\widetilde \Gamma_{\alpha}(\*z_{obs}) =  (-.6458, .6458)$ and  $\Gamma_{\alpha}(\*y_{obs}) = \big\{\btheta:   y_{obs, i} = \theta +  u_i^*, i = 1, 2, 3; \,  
\exists  (u_1^*, u_2^*, u_3^*) \in (-1, 1) \times(-1, 1) \times (-1, 1) \big\} \cap \widetilde \Gamma_{\alpha}(\*y_{obs}) = \big\{(-1.430, .570) \cap (-.951, 1.049) \cap (-.629, 1.371)\big\} \cap   (-.6458, .6458)= (-.629, .570)$. Clearly,  $\Gamma_{\alpha}(\*z_{obs}) \subsetneq \widetilde \Gamma_{\alpha}(\*z_{obs})$. 
We repeated the experiment $1,000$ times. The coverage rates of $\widetilde \Gamma_{\alpha}(\*y_{obs})$ and $\Gamma_{\alpha}(\*y_{obs})$ are both $95.1\%$. Among the $1,000$ repetitions, $219$ times 
$\Gamma_{\alpha}(\*y_{obs}) = \widetilde \Gamma_{\alpha}(\*y_{obs})$ and 
$781$ times $\Gamma_{\alpha}(\*y_{obs}) \subsetneq \widetilde \Gamma_{\alpha}(\*y_{obs})$. The average lengths of $\Gamma_{\alpha}(\*y_{obs})$ and $\widetilde \Gamma_{\alpha}(\*y_{obs})$ are $0.915$ and $1.292$, respectively. 
\end{example}

As mentioned in Remark 2.2, a nuclear mapping function does not need to be a test statistic.
However, if it is defined through a test statistics by (\ref{eq:tT1}), there are some immediate implications from Lemma 1. 
First, the confidence set by the repro samples method is either smaller or the same as the one obtained using the test statistic by the classical testing approach. So the set $\Gamma(\*z_{obs})$ by the repro samples method is more desirable. Second, if the test statistic is optimal (in the sense that it leads to a powerful test), then the confidence set constructed by the corresponding nuclear mapping function is also optimal. In particular, we have the following corollary.
Here, a level-$\alpha$ {\it uniformly most accurate} (UMA) confidence set (also called {\it Neyman shortest}) refers to a level-$\alpha$ confidence set that minimizes the probability of false coverage (i.e., probability of covering a wrong parameter value) over a class of level-$\alpha$ confidence sets; cf., \citet[][\S 9.3.2]{tCAS90a}.  See Appendix I for a proof of the corollary.

\begin{corollary}\label{cor:UMA}
(a) If the test statistic $\widetilde T(\*z_{obs}, \btheta)$  corresponds to the uniformly most powerful test and  $\widetilde \Gamma(\*z_{obs})$ is a level-$\alpha$ UMA confidence set, then the set $\Gamma(\*z_{obs})$ by the corresponding repro sample method is also a level-$\alpha$ UMA confidence set. \\
(b) If the test statistic $\widetilde T(\*z_{obs}, \btheta)$ corresponds to the uniformly most powerful unbiased test and  $\widetilde \Gamma(\*z_{obs})$ is a level-$\alpha$ UMA unbiased confidence set, then the confidence set $\Gamma(\*z_{obs})$ by the corresponding repro sample method is also a level-$\alpha$ UMA unbiased confidence set.
\end{corollary}

Finally, we can define the nuclear mapping through a pivot or an approximately pivot quantity, which is often used as a test statistic. The probability distribution of a pivot quantity does not depend on the unknown parameters. Hence it is convenient to define the nuclear mapping function through a pivot or an approximately pivot quantity, i.e., set $T(\*u, \btheta) = \widetilde T(\*z', \btheta)$,
where $\widetilde T(\*z', \btheta)$ is 
a {\it pivot quantity} whose distribution satisfies $\P \{\widetilde T(\*Z', \btheta)  \in B \} =  \int_{{\bf t} \in B}
f({\bf t}) d {\bf t}$
or an {\it approximately pivot quantity} whose distribution satisfies
$\P \{\widetilde T(\*Z', \btheta)  \in B \} =  \int_{{\bf t} \in B}
f({\bf t}) d {\bf t} \,  \{1 + o(\delta^{'})\}.$ Here, $f({\bf t})$ is a density function that is free of the unknown parameter $\btheta$ and $B \subset \mathbb{R}^{d}$ is any Borel set, and $\delta^{'} > 0$ is 
a small number as defined in Theorems \ref{thm:1} and \ref{thm:2}. 
In this case, 
we have simplified in (\ref{eq:B}) that the set $B_\alpha(\btheta ) \equiv B_\alpha$ is free or approximately free of $\btheta$.

Let's consider the example of a likelihood ratio test (LRT) statistic, which is often regarded as the most commonly used asymptotic pivot quantity; 
cf., \cite{Reid2015}.  
Example~\ref{ex:4} below
compares the repro sample method works with the LRT test.

We both consider an example where the LRT is asymptotically chi-square distributed, and revisit Example~\ref{ex:3} where the regular asymptotic result does not hold. It demonstrates the flexibility of  the repro samples method and its ability to achieve improved performance over the commonly used likelihood inference especially in finite sample~settings.

\begin{example}
[Likelihood inference] \label{ex:4}
(a) Assume $L(\btheta| \*Z)$ is the likelihood function $L(\btheta| \*Z)$ corresponding to model (\ref{eq:Z}).
The LRT statistic for test
$H_0: \btheta_0 = \btheta$ vs $H_1: \btheta_0 \not = \btheta$
is $\lambda_{\theta}(\*z_{obs}) =  
\frac{L(\theta|\*z_{obs})}{\sup_{\theta'} L(\theta'|\*z_{obs})}$.
Suppose the nuclear mapping function is defined as
$T(\*u, \btheta) = -2 \log\{\lambda_{\theta}(\*z)\} = -2 \log\{\lambda_{\theta}(G_z(\*u, \theta))\}$.
A level-$\alpha$ confidence set by the repro sample method is then
$\Gamma_{\alpha}(\*z_{obs}) 
 =  \big\{\btheta:  \*z_{obs} = G(\btheta, \*u^*), -2 \log\{\lambda_{\theta}(\*z_{obs})\}  
 \in B_\alpha(\btheta), \exists \,\*u^* \in {\cal U}\big\},$  
where $B_\alpha(\btheta)$ is a set that satisfies (\ref{eq:B}). This $\Gamma_{\alpha}(\*z_{obs})$ is an exact level-$\alpha$ confidence set, if $B_\alpha(\btheta)$ can be obtained for each given $\btheta$
or directly using Algorithm~\ref{alg:Ag}. 
Frequently, as $n \to \infty$ and under some regularity conditions, $-2 \log\{\lambda_{\theta}(\*Z)\}$ is asymptotically $\chi^2$-distributed with ${p}$ degrees of freedom \cite[][p381]{tCAS90a}, thus asymptotically free of $\btheta$.
It follows that an asymptotic level-$\alpha$ confidence set is
$\Gamma_{\alpha}(\*z_{obs}) =\big\{\btheta:  \*z_{obs} = G(\btheta, \*u^*),  -2 \log\{\lambda_{\theta}(\*z_{obs})\} \leq q_{\chi^2_{p}}(\alpha), \exists \, \*u^* \in {\cal U}\big\}  \subseteq \widetilde \Gamma_{\alpha}(\*z_{obs}) \equiv 
\big\{\btheta:  -2 \log\{\lambda_{\theta}(\*z_{obs})\} \leq q_{\chi^2_{p}}(\alpha) \big\},
$ 
where $q_{\chi^2_{p}}(\alpha)$ is the $\alpha$-th quantile of the $\chi^2_p$ distribution 
and 
$\widetilde \Gamma_{\alpha}(\*z_{obs})$ is the confidence set obtained by the regular LRT method. 
By Corollary~\ref{cor:UMA}, $\Gamma_{\alpha}(\*z_{obs})$ preserves the optimality of $\widetilde \Gamma_{\alpha}(\*z_{obs})$ when $\widetilde \Gamma_{\alpha}(\*z_{obs})$ is optimal.

 (b) 
 We now revisit Example~3, where the likelihood function 
$L(\theta|\*y^{obs}) = \prod\limits_{i =1}^n
    \big[ I\{-1 < y_i - \theta \leq 1\}/2\big] 
    \propto I\big\{ \max\limits_{1\leq i \leq n}|y_i - \theta| < 1 \big\}.
$
The LRT statistic is $\lambda_{\theta}(\*z_{obs}) 
= I\big\{ \max\limits_{1 \leq i \leq n}|y_i - \theta| < 1 \big\}$ and $-2 \log\{\lambda_{\theta}(\*Y)\}$ 
does not converge to a $\chi^2$ distribution. 
However, $\lambda_{\theta}(\*y_{obs})$ is 
decreasing in $\widetilde T(\*y_{obs}, \theta) = \max\limits_{1 \leq i \leq n}|y_i - \theta|$.
So, the LRT rejects $H_0$ with a large $\widetilde T(\*y_{obs}, \theta)$, which leads to 
a level-$\alpha$ confidence set $\widetilde \Gamma_\alpha(\*y_{obs}) = \{\theta: \max\limits_{1 \leq i \leq n} |y_i - \theta| < \alpha^{1/n}\}$. Although we could directly define our nuclear mapping function through $\widetilde T(\*y_{obs}, \theta) =  \max\limits_{1 \leq i \leq n} |y_i - \theta|$ to obtain a repro sample confidence interval, we consider here an alternative choice 
of 
a vector nuclear mapping function $T(\*u, \theta) = (u_{(1)}, u_{(n)})$,
where $u_{(k)}$ is the $k$th order statistic of a sample from $U(-1,1)$. Let $c_\alpha \in (0,1)$ be a solution of
$({c_\alpha+1})^n - 2^{n-1}c_\alpha^n = 2^{n-1}{(1 - \alpha)}$. We can show $\P\{(U_{(1)}, U_{(n)}) \in B_\alpha\} = \alpha$ for $B_\alpha = (-1, -c_\alpha) \times (c_\alpha, 1)$. Thus, a level-$\alpha$ confidence set is $\Gamma_\alpha(\*y_{obs}) = \{\theta: y_{(n)} - 1 < \theta < y_{(1)} + 1;$ 
$ y_{(1)} - \theta < -c_\alpha, y_{(n)} - \theta > c_\alpha\} =
\big(\max\big\{y_{(1)} + c_\alpha, y_{(n)} - 1\big\}, \min\big\{y_{(n)} - c_\alpha, y_{(1)} + 1\big\}\big)$. We have conducted a numerical study with true $\theta_0 = 0$ and $n = 5, 20, 200$, respectively. In all cases, the coverage rates of both confidence intervals are right on target around $95\%$ in 3000 repetitions.  
The intervals by the repro sample method are consistently shorter than those obtained using the LRT method across all sample sizes, as evidenced in Figure~\ref{fig:lrt}. Although a LRT is uniformly most powerful for a simple-versus-simple test by the Neyman-Pearson Lemma, it is not the case for the two-sided test in this example. %In contrast, 
we are able to explore and use the repro samples method to obtain a better confidence interval.    
\end{example}

\begin{figure}
    \centering
      \includegraphics[width=\textwidth, height= 0.25\textheight]{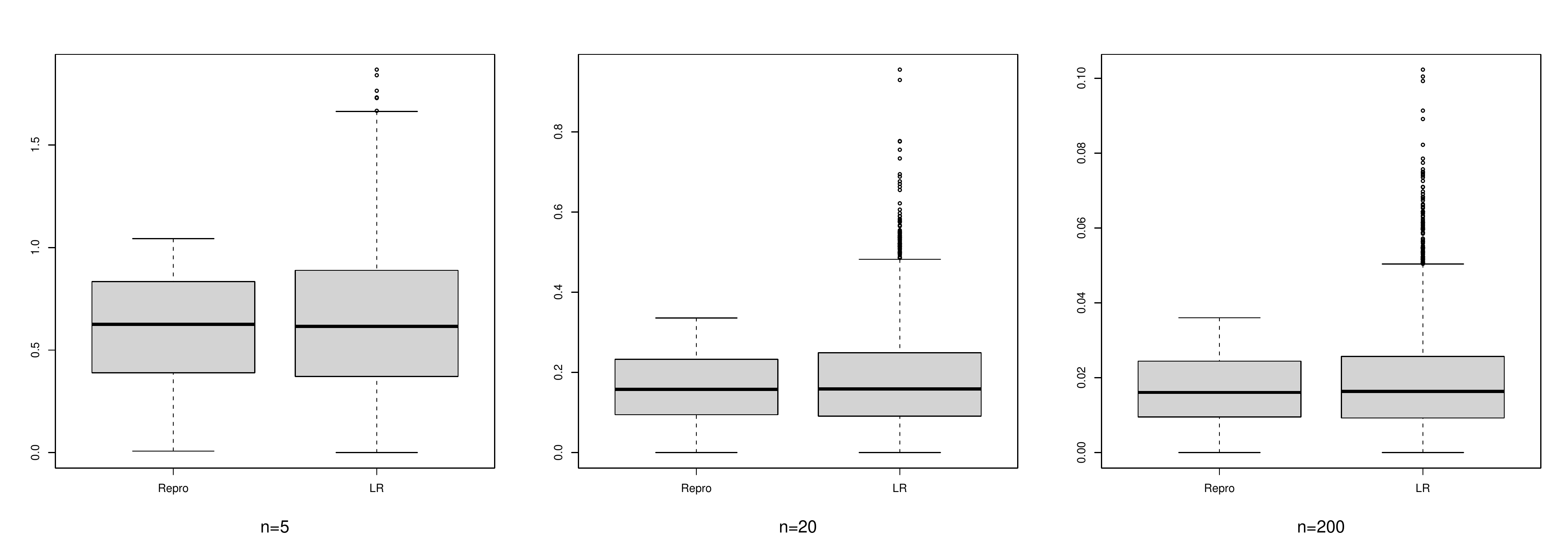}
    \caption{The side-by-side box plots comparing the widths of the two 95\% confidence intervals discussed in Example 4, for $n = 5, 20$ and $200$, respectively; 
    Number of repetitions $= 3000$.}
    \label{fig:lrt}
\end{figure}

\subsection{
Define a nuclear mapping and control nuisance parameters} 

Frequently researchers are 
interested in making 
inference only for a certain parameter, say $\eta$, and the remaining parameters $\bxi$ in $\btheta = (\eta, \bxi^\top)^\top $ are considered nuisance parameters.  
There is a large number of literature on this topic where
profiling and conditional inference remain to be two major techniques that handle nuisance parameters \citep[][and references  therein]{CoxReid1987, Davison2003}. 
These techniques typically 
lead to test statistics, say $\widetilde T(\*z, \eta)$, that do not involve the nuisance parameter $\bxi$ in expressions. 
Often times however the distributions of these test statistics $\widetilde T(\*z, \eta) = \widetilde T(G_z((\eta, \bxi^\top)^\top, \*u), \eta)$ still depend on the unknown nuisance parameter $\bxi$ and thus
not readily available. To obtain the distributions of $\widetilde T(\*z, \eta)$, researchers often invoke large samples 
and assume that there exists a consistent plug-in estimator of the nuisance parameter; e.g.,  \cite{Reid2015,Chuang2000}.
In this subsection, however, we use a similar profiling idea to provide an alternative approach. The advantage of our approach is that it can give us a confidence set with a guaranteed finite-sample coverage rate. 

Suppose $T_a(\*u, \btheta)$ is a mapping function from ${\cal U} \times \Theta \to \R^d$. 
For instance, we can often make 
$T_a(\*u, \btheta) = \widetilde T(\*z, \eta)$, a test statistic mentioned above using either a profile or a conditional approach. Or, more generally $T_a(\*u, \btheta)$ can be any nuclear mapping function $T(\*u, \btheta)$ used to make a joint inference about $\btheta = (\eta, \bxi^\top)^\top$. Suppose we have an artificial data point $\*z' = G_z(\*u, \btheta)$ and let $\Gamma_{\alpha'}(\*z')$ be a level-$\alpha'$ confidence set for $\btheta$,  where $\Gamma_{\alpha'}(\*z)$ 
has the expression of (\ref{eq:G1}) or (\ref{eq:G2}) but with 
$\*z_{obs}$ replaced by the artificial data $\*z' = G_z(\*u, \btheta)$ and the nuclear mapping function replaced by $T_a(\*u, \btheta)$. Now, beside $\btheta = (\eta, \bxi^\top)^\top$, let us consider another copy of parameter vector $\widetilde \btheta = (\eta, \widetilde \bxi^\top)^\top \in \Theta$ that has the same $\eta$ but perhaps a different $\widetilde \bxi \not= \bxi$. To explore the impact of a potentially wrong nuisance parameter $\widetilde \bxi \not= \bxi$, we define
\begin{equation}
    \label{eq:nu}
\nu(\*u; \eta, \bxi, \widetilde \bxi) = \inf_{\alpha'} \{\alpha': \widetilde\btheta \in \Gamma_{\alpha'}(\*z)\} = \inf_{\alpha'}  \{\alpha': \widetilde\btheta \in \Gamma_{\alpha'}(G_z(\*u, \btheta))\}. 
\end{equation}
From 
Corollary~\ref{col:test}, we can interpret $1-\nu(\*u; \eta, \bxi, \widetilde \bxi)$ as a `$p$-value' when we test the null hypothesis that the data $\*z$ is from $\widetilde \btheta = (\eta, \widetilde \bxi^\top)^\top$ but in fact $\*z = G_z(\*u, \btheta)$ with $\btheta = (\eta, \bxi^\top)^\top$. 
Our proposed {\it profile nuclear mapping function} 
for the target parameter $\eta$ is defined as
\begin{align}
\label{eq:T_p}
T^P(\*u, \btheta) = \min_{\widetilde\xi} \,\, \nu(\*u; \eta, \bxi, \widetilde \bxi). 
\end{align}
Although $T^P(\*u, \btheta)$ still depends on $\btheta = (\eta, \bxi^\top)^\top$, including the unknown nuisance parameter~$\bxi$, it
is dominated by a $U(0,1)$ random variable as stated in the lemma below. See Appendix I for a proof of the lemma.
\begin{lemma}
\label{lemma:nuisance}
Suppose $T^P(\*u, \btheta)$ is a profile nuclear mapping function from ${\cal U} \times \Theta \to [0,1]$ defined above. Then, we have $\P\left\{T^P(\*U, \btheta) \leq \alpha \right\} \ge \alpha.$ 
\end{lemma}
Based on the above lemma and following the construction outlined in Section 1, we define 
\begin{align}
\label{eq:CIeta}
& \Xi_\alpha(\*z_{obs}) = \bigg\{\eta: \exists \, \*u^* \in {\cal U} \, \text{and} \, \bxi, \nonumber \\
& \qquad
\text{ s.t. }  \begin{pmatrix}\eta  \\ \bxi
\end{pmatrix}\in \Theta,  \*z_{obs} = 
G_z\left(\begin{pmatrix}\eta \\ \bxi
\end{pmatrix}, \*u^*\right),
\,
T^P\left(\*u^*, \begin{pmatrix}\eta \\ \bxi
\end{pmatrix}\right) \leq \alpha   \bigg\}.     
\end{align}

\noindent
The next theorem states that $\Xi_\alpha(\*z_{obs})$
is a level-$\alpha$ confidence set of the target parameter~$\eta$. A proof is provided in Appendix. 
\begin{theorem}\label{the:nuisance}
Suppose that $\btheta_0 = (\eta_0, \bxi_0^\top)^\top$ is the true parameter with $\*Z = G_z(\*U, \btheta_0)$ and $\Xi_\alpha(\cdot)$ is defined in \eqref{eq:CIeta},  we have $\P\left\{\eta_0 \in \Xi_\alpha(\*Z)\right\} \ge \alpha$. 
\end{theorem}

In cases when the distribution of $T_a(\*U, \btheta)$ is known, the function $\nu(\cdot; \eta, \bxi, \widetilde \bxi)$ in (\ref{eq:nu}) and the nuclear mapping function $T^P(\*u, \btheta)$ have explicit formulas.
When the distribution of $T_a(\*U, \btheta)$ is not available, we can use a Monte-Carlo method to get $\nu(\cdot; \eta, \bxi, \widetilde \bxi)$ for each given $(\eta, \bxi, \widetilde \bxi)$ as described in the lemma below. See Appendix I for a proof.

\begin{lemma}
\label{lem:nuisance_depth}
Let $D_{{\cal S}_{\theta}}(\*t)$ be the empirical depth function that is computed based on the Monte-Carlo points in ${\cal S}_{\theta}$, where 
${{\cal S}_\theta} = \{T_a(\*u^s, \btheta), \*u^s \in {\cal V}\}$ is the Monte-Carlo set defined in Algorithm 1 for a given $\btheta = (\eta, \bxi^\top)^\top$. Let $ F_{\mathcal V|D}(s) = \frac{1}{|{\cal V}|} \sum_{\*u^s \in {\cal V}}I\big\{D_{{\cal S}_{\theta}}\big(T_a(\*u^s, \btheta) \big) \leq s \big\}$ be the empirical CDF of $D_{{\cal S}_{\theta}}\big(T_a(\*u^s, \btheta)\big),$ then
$$\nu(\*u; \eta, \bxi, \widetilde \bxi) = \inf_{\{\*u^*: G_z(\*u^*, \tilde\theta)= G_z(\*u, \tilde\theta)\}} \left\{1- F_{\mathcal V|D} \left(D_{{\cal S}_{\tilde\theta}}\big(T_a(\*u^*, \tilde\btheta) \big)\right)\right\},$$
where $\tilde\btheta = (\eta, \widetilde \bxi^\top)^\top. $
\end{lemma}

Using Lemma~\ref{lem:nuisance_depth}, we can evaluate $T^P(\*u, \btheta)$ via \eqref{eq:T_p}, and then obtain the confidence set in \eqref{eq:CIeta}. Note that, for any given $(\*u, \btheta)$, we can express $\nu(\*u; \eta, \bxi, \widetilde \bxi)$ as a function of $\widetilde \bxi$ only, so we can use a build-in optimization function in R (or another computing package) to obtain $T^P(\*u, \btheta)$ in \eqref{eq:T_p}.  
The following continuation of Example 2 illustrates constructing confidence interval \eqref{eq:CIeta} for a parameter of interests in presence of nuisance parameters. 
The case study example in Section 4 provides another illustration of a more difficult inference problem in a normal mixture model where the target parameter is the number of unknown components. 

\begin{example2b*}
\label{ex:crq_b}
We continue Example 2(b) to  
obtain a confidence interval for a single regression coefficient $\theta_k(\zeta)$ in the censored quantile regression model (with the remaining coefficients as nuisance parameters), $k = 0, \ldots, p$. We consider $T_a(\*u, \btheta)$ be the $(k+1)$-th element of the vector  $T(\*u, \btheta) = \big\{\sum\limits_{i=1}^n \*x_i \*x_i^T I({\*x}_i^T \btheta > 0)\big\}^{-1}\sum\limits_{i=1}^n {\*x}_i I{({\*x}_i^T \btheta > 0) I(u_i \leq \tau)}$ that is previously used in Example~\ref{ex:crq}(b) to make joint inference on the entire $\btheta$. 
Following Theorem 3 and Lemma 3, a 
confidence set 
for the single~parameter~$\theta_k(\zeta)$~is
\begin{align*}
    \Xi_{\alpha}(\*y_{obs}) = \left\{\theta_k:  \btheta = (\theta_k, \btheta_{(-k)}), g(\*y_{obs}, \btheta, \*u^*) = 0, 
      F_{\mathcal V|D}\big(T_a(\*u, \btheta)\big) \geq 1-\alpha,
 \exists 
\*u^* \in {\cal U} \right\}, 
\end{align*}
where $F_{\mathcal V|D}(s)$ is defined in Lemma 3. 
Table 3 reports the performance of the above confidence interval for $\theta_2(\zeta)$ from simulation studies on the same setting as in Example~\ref{ex:crq}(B). We observe that for $\zeta = 0.5$, repro samples achieves the desirable coverage while the bootstrap method by \cite{bilias_simple_2000} consistently under covers. The widths of the confidence intervals from the two approaches are similar. For $\zeta$ =0.9, both approaches well cover the true parameter, but the bootstrap confidence intervals appear to be much wider.

\begin{table}[]
\centering
\resizebox{\textwidth}{!}{\begin{tabular}{lc|cc|cc}
\hline
                 \multirow{2}{*}{Quantile} & \multirow{2}{*}{Error Distribution} & \multicolumn{2}{c|}{repro samples}  & \multicolumn{2}{c}{Bootstrap}  \\ 
                 & & Coverage & Width & Coverage & Width \\ \hline
                 \multirow{3}{*}{$\zeta = 0.50 $} & Normal  & 0.944(0.007)   & 0.747(0.006)       & 0.920(0.009)           & 0.727(0.007)    \\
                  & Mixture &  0.950(0.007)           & 0.832(0.007)        & 0.935(0.008)          & 0.839(0.008)            \\
                  & Heteroscedastic & 0.940(0.008)        & 0.826(0.007)        & 0.921(0.009)          & 0.844(0.009)   \\ \hline
\multirow{3}{*}{$\zeta = 0.90$} & Normal &0.979(0.005)           & 0.813(0.007)        & 0.983(0.004)          & 4.190 (0.424)               \\
                  & Mixture & 0.973(0.005)           & 1.407(0.013)        & 0.973(0.005)          & 3.260(0.312)     \\
                  & Heteroscedastic & 0.945(0.007)       & 1.149(0.009)       & 0.967(0.006)          & 4.124(0.439)      \\ 
                  \hline
               \end{tabular}}
               \caption{ 
The model and error distributions are identical to those in Table~\ref{tab:crq_joint}. Here we compare the coverage rates and widths of confidence intervals for the regression coefficient of $x_{i2}$ by the repro samples approaches and the bootstrap approach \citep{bilias_simple_2000}. The coverage rates and widths are computed using 1000 repetitions.}
\end{table}

\end{example2b*}

\subsection{Finding candidate set for target parameter}
\label{sec:candidate_general}

Sometimes when an explicit expression of the confidence set $\Gamma_\alpha(\*z_{obs})$ is not available,  
one can use Algorithm 1 to search through the parameter space $\Theta$ and construct $\Gamma_\alpha(\*z_{obs})$. Finding a suitable candidate set to limit the search can greatly improve the computing efficiency of the algorithm. Here, a candidate set, say $\widehat \Theta = \widehat \Theta(\*z_{obs})$ refers to a subset of the parameter space $\Theta$
that is pre-screened using the observed data $\*z_{obs}$. 
In this section, we first provide a general result of the impact of using a data-dependent candidate set on the coverage. We then discuss the special case of a discrete target parameter in which we can take advantage of a many-to-one mapping that is unique in a repro samples setup
to create an effective pre-screen candidate set to significantly reduce our computing cost. 

The candidate set $\widehat \Theta = \widehat \Theta(\*z_{obs})$ may be obtained using different
pre-screening methods.
The key is that we need to overcome the issue of a ``double usage" of the data but still ensure an at least approximate coverage of the final output confidence set 
\begin{align}
\label{eq:G1_candidate}
& \Gamma_{\alpha}'(\*z_{obs})   = \widehat \Theta(\*z_{obs}) \cap \Gamma_{\alpha}(\*z_{obs}) \nonumber \\
& \qquad = \big\{\btheta:  \exists \, \*u^* \in {\cal U}, \btheta \in \widehat\Theta,  \, \hbox{s.t.} \, \*z_{obs} = 
G_z({\btheta}, \*u^*),
\, 
T(\*u^*, \btheta)  \in B_{\alpha}(\btheta)  \big\},
\end{align}
we impose a condition that 
\begin{equation} \label{eq:cand-cond}
    \P\big\{ \btheta_0 \in \widehat \Theta(\*Z)\big\} = 1 - o(\delta'),
\end{equation}
where $\delta' > 0$ is a small number as described in Theorem~\ref{thm:1}. Here, the ``double usage" refers to the fact that we use the observed data in constructing both the candidate set $\widehat \Theta(\*z_{obs})$ and confidence set $\Gamma_{\alpha}(\*z_{obs})$. 
The corollary below contains a result  showing that the output set $\Gamma_{\alpha}'(\*z_{obs})$
is still an (approximate) level-$\alpha$ confidence set. 
See Appendix I for a proof of the corollary.

\begin{corollary}
\label{cor:cand_cs_coverage}
Let $\*Z= G_z(\*U, \btheta_0)$ and $\widehat\Theta = \widehat\Theta(\*Z)$ is a candidate set that satisfies (\ref{eq:cand-cond}).  
Then, we have 
$\P\big\{\btheta_0 \in \Gamma_{\alpha}'(\*Z)\big\} \geq \alpha - o(\delta').$
\end{corollary}

When the parameter space $\Theta$ is continuous and the mapping function $T(\*u, \btheta)$ is continuous in $\btheta$, we may use a grid search method to cover the entire $\Theta$, as discussed in Remark~2.6. A data-dependent candidate set $\widehat\Theta$ 
can help further narrow down to the grid search;  cf.,  Example~\ref{ex:crq_b} and also  \cite{michael_exact_2019-1}.

An intriguing yet unique situation in the repro samples development is when the parameter space $\Theta$ is discrete. 
We can take advantage of a ``many-to-one" mapping that often exists in the setup to get a data-dependent candidate set $\widehat \Theta$ that is much smaller than $\Theta$. 
We describe the above key idea and a 
generic approach below. Section~4 includes a case study with with greater details to provide a specific illustration on how to obtain this kind of  candidate set $\widehat \Theta.$

For convenience of discussion, we rewrite (\ref{eq:Z}) as $\*Z= G_z(\kappa_0, \bxi_0, \*U)$, where $\kappa_0$ is the true value of a discrete parameter with a parameter space ${\it \Upsilon}$ and, if applicable, $\bxi_0$ are the values of other (nuisance) parameters. Write $\btheta = (\kappa, \bxi^\top)^\top$. 
Our goal is to find a candidate set $\widehat {\it \Upsilon} = \widehat {\it \Upsilon}(\*z_{obs})$ such that $\widehat {\it \Upsilon}(\*z_{obs})$ is significantly smaller than ${\it \Upsilon}$ and $\widehat {\it \Upsilon}$ satisfies (\ref{eq:cand-cond}).

Let us first examine a ``many-to-one'' mapping often inherited in a repro sample setup involving a discrete parameter $\kappa$. In particular, based on the equation $\*z_{obs} = G_z(\kappa_0, \bxi_0,  \*u^{rel}),$ we can rewrite $\kappa_0$ as a solution of an optimization problem
$\kappa_0 = \arg \min\limits_{\kappa} \min\limits_{\xi}\|\*z_{obs} - G_z(\kappa, \bxi, \*u^{rel})\|$ or, more generally, 
\begin{equation}
    \kappa_0 = \arg \min\limits_{\kappa} \min\limits_{\xi}L\big(\*z_{obs}, G_z(\kappa, \bxi, \*u^{rel})\big),
    \label{eq:tau0}
\end{equation}
where $L\big(\*z, \*z'\big)$ is a continuous loss function that measures the difference between two copies of data $\*z$ and $\*z'$. However, we do not know (observe)  $\*u^{rel}$. By replacing $\*u^{rel}$ with a repro copy $\*u^*$ in (\ref{eq:tau0}), we get
\begin{equation}
   \kappa^* = \arg \min\limits_{\kappa} \min\limits_{\xi}L\big(\*z_{obs}, G_z(\kappa, \bxi, \*u^*)\big),
    \label{eq:tau-star} 
\end{equation}
which maps a value $\*u^* \in {\cal U}$ to a $\kappa^* \in {\it \Upsilon}$. Typically, ${\cal U}$ is uncountable and ${\it \Upsilon}$ is countable, so the mapping in (\ref{eq:tau-star}) is ``many-to-one''; That is, many $\*u^* \in {\cal U}$ correspond to one $\kappa^* \in {\it \Upsilon}$. We are particularly interested in the 
subset $S =\big\{\*u^*: \kappa_0 = \arg \min\limits_{\kappa} \min\limits_{\xi}L\big(\*z_{obs}, G_z(\kappa, \bxi, \*u^*)\big)\big\}$  $\subset {\cal U}$, where for any $\*u^* \in S$ the mapping in (\ref{eq:tau-star}) produces $\kappa^* = \kappa_0$. Since 
$\*u^{rel} \in S$, $S \not = \emptyset$. We assume $S \subset {\cal U}$ is a nontrivial set containing many elements that can not be ignored.  

Now suppose $\P\big\{\*U \in S \big | \*z_{obs}\big\} > 0$ has a nontrivial probability. 
Then, we can use a Monte-Carlo method to obtain a candidate set $\widehat {\Upsilon} = \widehat {\Upsilon}(\*z_{obs})$
such that $\P\big\{ \kappa_0 \not \in \widehat {\it \Upsilon}(\*Z)\big\} = o(\delta')$. 
The idea is to simulate a sequence of $\*u^c \sim \*U$, say, 
${\cal V}_c = \big\{\*u_1^c, \ldots, \*u_N^c \big\}$. When $N = |{\cal V}_c|$ is large enough, ${\cal V}_c\cap S \not = \emptyset$ with a high (Monte-Carlo) probability. 
That is $\kappa_0 \in \widehat{\it \Upsilon}_{{\cal V}_c}(\*z_{obs})$ with a high (Monte-Carlo) probability, where $\widehat{\it \Upsilon}_{{\cal V}_c}(\*z_{obs})$ is our candidate set  
\begin{equation}
\label{eq:cand-hat}
    \widehat{\it \Upsilon}_{{\cal V}_c}(\*z_{obs}) = \left\{\kappa^* = \arg \min\limits_{\kappa} \min\limits_{\xi}L\big(\*z_{obs}, G_z(\kappa, \bxi, \*u^c)\big) \,\, \big| \,\, \*u^c \in {\cal V}_c \right\}.
\end{equation}
We show next that 
$\P\big\{ \kappa_0  \in \widehat{\it \Upsilon}_{{\cal V}_c}(\*Z)\big\}) = 1 - o(\delta')$, where $\delta' = \delta'({\footnotesize|{\cal V}_c|}) \to 0$ as $|{\cal V}_c| \to \infty$. 

Let 
$\*U^*$ be an independent copy of $\*U$.  Formally, we assume 
the following condition: 
\begin{itemize}
    \item[(N1)] 
For 
any $\*u^{rel} \in \mathcal U$, there exists a 
neighborhood of $\*u^{rel}$, say $S_{\cal N}(\*u^{rel})$, 
such that
\begin{equation}
    \label{eq:N1}
    \resizebox{.84\hsize}{!}{$S_{\cal N}(\*u^{rel})\subseteq S =\big\{\*u^*: \kappa_0 = \arg \min\limits_{\kappa} \min\limits_{\theta}L\big(G_z(\btheta_0, \kappa_0, \*u^{rel}), G_z(\btheta, \kappa, \*u^*)\big)\big\}$}
\end{equation}
and    $
   \P_{(\*U^*,\*U)}\left\{\*U^* \in S_{\cal N}(\*U) \right\}
   \geq P_{\cal N} > 0$,
   for a positive number $P_{\cal N}$.
\end{itemize}

Although the neighborhood $S_{\cal N}$ and the number $P_{\cal N}$ are problem specific,  Condition (N1) is typically satisfied (or nearly satisfied)
when $L(\cdot, \cdot)$ is a continuous loss function and $|{\cal U}|$ is uncountable. 
The following lemma suggests that, 
under Condition (N1) and as $|{\cal V}_c| \to 0$, 
$\widehat{\it \Upsilon}_{{\cal V}_c}(\*z_{obs})$ 
in (\ref{eq:cand-hat}) satisfy condition (\ref{eq:cand-cond}). Here, the probability $\P_{\*U, {\cal V}_c}(\cdot)$ refers to the joint distribution of $\*U$ and $|{\cal V}_c|$ Monte-Carlo $\*U$'s in ${\cal V}_c$. See Appendix I for a proof of the lemma.

\begin{lemma}
\label{lemma:p_bound_C_D_general} 
Suppose Condition (N1) holds and $\widehat{\it \Upsilon}(\*z_{obs})$ is defined in \eqref{eq:cand-hat}. 
Then, 
\begin{align*}
    \P_{\*U, {\cal V}_c}\left\{\kappa_0 \not\in \widehat{\it \Upsilon}_{{\cal V}_c}(\*Z)\right\} \leq (1-P_{\cal N})^{|{\cal V}_c|}.
\end{align*}
\end{lemma}

Based on the lemma, we propose the following theorem. Here, the probability  $\P_{\*U|{\cal V}_c}(\cdot)$ refers to the condition distribution of $\*U$ given the set ${\cal V}_c$. See Appendix I for a proof. 

\begin{theorem}
\label{thm:p_bound_C_D_general}
 Suppose Condition (N1) holds and $\widehat{\it \Upsilon}(\*z_{obs})$ is defined in \eqref{eq:cand-hat}. Then, there exists a positive number $c$ such that we have, in probability, 
 	(a) $\P_{\*U|{\cal V}_c}\left\{\kappa_0 \in \widehat{\it \Upsilon}_{{\cal V}_c}(\*Z)\right\} = 1 - o_p\left(e^{-c|{\cal V}_c|}\right)
 	$
 	and (b) $\P_{\*U|{\cal V}_c}\big\{\btheta_0 \in \Gamma_{\alpha}'(\*Z)\big\} \geq \alpha -o_p\left(e^{-c|{\cal V}_c|}\right) 
 	$.
\end{theorem}

In the theorem, the small $o_p\left(e^{-c|{\cal V}_c|}\right)$ term tends to $0$ as the size of the Monte-Carlo set $|{\cal V}_c| \to \infty$, and the result holds for any fixed (data) sample size $n < \infty$. So the theorem is a small sample result. Also, the ``in probability" statement regards the Monte-Carlo randomness of ${\cal V}_c$. Based on Theorem 4, we have a clean 
interpretation for the frequentist
coverage statement: 
if we use the proposed Monte-Carlo method to obtain the candidate set $\widehat{\it \Upsilon}_{{\cal V}_c}(\*Z)$, we have a near $100\%$ chance (i.e., in probability statement in ${\cal V}_c$) that the coverage provability of the  obtained set $\Gamma_{\alpha}'(\*Z)$ is
no less than $\alpha$ (i.e., frequency statement in $\*U$). 

We will further discuss and provide an extension of Condition (N1) and  Theorem~\ref{thm:p_bound_C_D_general} in Section~4.3.

\section{Case study example: normal mixture with unknown number of components} 
\label{sec:mixture}

Mixture distributions are a commonly used model in data science.
However, a challenging task 
remains to be 
quantifying the uncertainty and making inference for the unknown number of components, say $\tau_0$, in a finite mixture distribution. 
Although  several existing procedures  provide point estimates for $\tau_0$, the question of how to assess the uncertainty of estimation of  $\tau_0$ is still open.
This difficulty is due in large part to the fact that the unknown $\tau_0$ is a discrete natural number 
and  typical tools such as the large sample CLT are no long applicable to the point estimators of $\tau_0$.
In this section, we demonstrate how to use the repro samples method to solve this open question in a normal mixture model by constructing a finite sample level-$\alpha$ confidence set for the unknown $\tau_0$. 
This case study example also 
sheds light on how to address inference problems of discrete parameters and handle nuisance parameters in general cases. 

\subsection{Notations: normal mixture model}

Suppose $Y_i\sim N\big(\mu^{(0)}_k, (\sigma^{(0)}_k)^2\big)$, one of the $\tau_0$ models in ${\cal M} = \{N\big(\mu^{(0)}_k, (\sigma^{(0)}_k)^2\big),$  $k \in \{1, \ldots, \tau_0\}\}$, for $i = 1, \ldots, n.$ Write ${\*Y} = (Y_1, \ldots, Y_n)^{\top}$. Here, the model parameters are ${\btheta} = (\tau,\*M_\tau,  {\*\mu}_{\tau}, \*\sigma_\tau^2)^{\top}$, where ${\*M}_{\tau} = ({\*m}_1, \ldots, {\*m}_n)^{\top}$ is the unknown (unobserved) $n \times \tau$ membership assignment matrix,  where ${\*m}_i = 
(m_{i1}, \ldots, m_{i\tau})^{\top}$ and  $m_{ik}$ is the indicator of whether $Y_i$ is from the $k$th distribution $N(\mu_k, \sigma^2)$, $1 \leq i \leq n$, $1 \leq k \leq \tau$. We denote by the corresponding true values $\tau_0$, $\*M_0  = ({\*m}^{(0)}_1, \ldots, {\*m}^{(0)}_n)^{\top}$ $=$ $\big((m^{(0}_{11}, \ldots, m^{(0)}_{1\tau})^{\top}, \dots, (m^{(0)}_{n1},$ $\ldots,$ $m^{(0)}_{n\tau})^{\top}\big)^\top$,
 ${\*\mu}_0 = (\mu_1^{(0)}, \ldots, \mu_{_{\tau_0}}^{(0)})^{\top}$, $\*\sigma_0 = (\sigma^{(0)}_1, \dots, \sigma^{(0)}_{\tau_0})^{\top}$  and ${\*\theta}_0$ $= $ $(\tau_0, \*M_0, {\*\mu}_0,$ $\*\sigma_0)^{\top}$, respectively.  Because we are interested in inference of the unknown number of components $\tau<n,$
${\*\mu}_\tau = (\mu_1, \ldots, \mu_\tau)^{\top}$ and $\*\sigma_\tau = (\sigma_1, \dots, \sigma_\tau)^{\top}$ are therefore unknown nuisance parameters of the $\tau$ components' means and variances.

We can re-express $Y_i = \sum_{k =1}^\tau m^{(0)}_{ik} (\mu_k^{(0)} + \sigma^{(0)}_{k}U_{i})$, where $U_i \sim N(0,1)$, in a vector form, 
\begin{align}
\label{eq:mixture_matrix}
     {\*Y} = {\*M}_{0} {\*\mu}_{0} + \diag(\*U)\*M_{0}  {\*\sigma_0} = {\*M}_{0} {\*\mu}_{0} + \diag(\*M_{0}  {\*\sigma_0}) \*U, 
\end{align}
where $\*U = (U_1, \ldots, U_n)^{\top}$. 
The realized version of (\ref{eq:mixture_matrix}) is 
${\bm y_{obs}} = {\bm M}_{0} {\bm \mu}_{0} +  \diag({\bm  u}^{rel}){\bm M}_{0} \*\sigma_0, 
$ 
where $\*M_0$ is the realized membership assignment matrix that generates $\*y_{obs}$. 
A typical 
setup $Y_i \sim \sum_{k =1}^\tau p_{k} N(\mu_k, \sigma_k^2)$ further assumes that ${\*m}_i= (m_{i1}, \ldots, m_{i\tau})^{\top}$ is a random draw from multinomial distribution $\text{Multinomial}(\tau; p_1, \ldots, p_{\tau})$ with unknown proportion parameters $(p_1, \ldots, p_{\tau})$, $\sum_{k =1}^\tau p_k = 1$.  

The literature confirms the identifiability issue in a mixture model that multiple sets of parameters could possibly generate the same data \citep{LiuShao2003}. To account for this identifiability issue, we re-define the true parameters as those that correspond to the ones with the smallest $\tau$, i.e., 
\begin{align}
\label{eq:C1}
(\tau_0, \*M_0, {\*\mu}_0, \*\sigma_0^2) = \underset{\{(\tau, \*M_{\tau}, {\*\mu}_{\tau}, \*\sigma^2)|  {\*M}_{\tau} {\*\mu}_{\tau} = {\*M}_{0} {\*\mu}_{0}, \,   {\*M}_{\tau} {\*\sigma}_{\tau} = {\*M}_{0} {\*\sigma}_{0}\} } {\arg \min} \tau,   
\end{align}
when the identifiability issue arises. 
Here, $\*M_0$ and ${\bm \mu}_0$ on the right hand side are the true parameter values that generates $\*y^{obs}.$ With a slight abuse of notation, we still use the same notations on the left hand side to denote the ones 
with the smallest $\tau$. 
We further assume that $(\tau_0, \*M_0, {\*\mu}_0, \*\sigma_0^2)$ defined in (\ref{eq:C1})
is unique.  We also assume that 
 $0 < \sigma_{\min} \leq \sigma_{\max} < \infty,$
    where $\sigma_{\min} = \|\*\sigma_0\|_0$ and $\sigma_{\max} = \|\*\sigma_0\|_\infty.$
That is, among the $\tau_0$ components, none of the variance $\sigma_k^2$ can completely dominate another $\sigma_{k'}^2$, for any $k \not= k'$. The inference we develop is on the set of parameters defined in (\ref{eq:C1}) with the smallest $\tau_0$.

The inference problem for the target parameter $\tau_0$ is complicated. Here, we have both types of nuisance parameters of varying sizes: discrete nuisance parameters $\*M_\tau$ that are potentially high-dimensional and continuous nuisance parameters $(\*\mu_\tau, \*\sigma_{\tau})$ whose lengths change with $\tau$. 
To have an effective repro samples approach that is computationally efficient, 
we handle the two types of nuisance parameters in two different ways. For the discrete nuisance parameter $\*M_\tau$ we take advantage of
the many-to-one mapping described in Section 3.3 to obtain a manageable candidate set. For the continuous nuisance parameters $(\*\mu_\tau, \*\sigma_{\tau})$, we utilize a specific feature of the normal mixture distribution, namely that, given $\tau$, we can find sufficient statistics for $(\*\mu_\tau, \*\sigma_{\tau})$
and avoid searching through all possible values of $(\*\mu_\tau, \*\sigma_{\tau})$. In Section 4.2, we define a nuclear mapping function $T^P(\*u, \btheta)$ that minimizes the impact of the nuisance parameters.
In Section 4.3, we construct an effective candidate set of $(\tau, \*M_\tau)$ over the corresponding discrete parameter space to greatly improve our computing efficiency.

\subsection{Nuclear mapping function
for 
given $(\tau, \*M_\tau)$} 
In most derivations of this subsection, $\tau$ is fixed. For notation simplicity, we suppress the subscripts in  $(\*\mu_\tau, \*\sigma_\tau)$ and simply write $(\*\mu, \*\sigma)$. Note that the vectors $\*\mu$ and $\*\sigma$ have the length of $\tau$ that changes when $\tau$ varies. Also, let $\*H_\tau = {\*M}_{\tau} ({\*M}_{\tau}^\top {\*M}_{\tau})^{-1} {\*M}_{\tau}^\top$, $\*H_\tau^{(k)} = \*H_\tau[\mathcal I_k, \mathcal I_k]$ be the submatrix of $\*H_\tau$ corresponding to the $k$th component with $\mathcal I_k=\{i: m_{ik}=1\}$ being the index set corresponding component $k$,  and $\*Y^{(k)} = \*Y[\mathcal I_k]$ be the subvector of $\*Y$ that is corresponding to the $k$th component.

Corresponding to
$\*Y = \*M_{0} \*\mu_0 + \diag(\*M_0  {\*\sigma}_0) \*U$ in \eqref{eq:mixture_matrix}, we can write $\*Y'$ as
$$
\*Y' =\*M_{\tau} \*\mu + \diag(\*M_\tau  {\*\sigma}) \*U =  \*M_{\tau} \A_\tau(\*Y') + \diag\left\{\*M_\tau\B_\tau(\*Y')\right\}\C_\tau(\*Y'),
$$
where $\A_\tau(\*Y') = 
(\*M_{\tau}^\top \*M_{\tau})^{-1}\*M_{\tau}^\top\*Y'$,
$$
\resizebox{.9\hsize}{!}{$\B_\tau(\*Y') = \begin{pmatrix} \|(I - \*H_{\tau}^{(1)})\*Y^{'(1)}\| \\ \vdots \\ \|(I - \*H_{\tau}^{(\tau)})\*Y^{'(\tau)}\|
\end{pmatrix} \,\, \hbox{and}\,\,
\C_\tau(\*Y') = \begin{pmatrix}\frac{(I - \*H_{\tau}^{(1)})\*Y^{'(1)}}{\|(I - \*H_{\tau}^{(1)})\*Y^{'(1)}\|} \\ \vdots \\ \frac{(I - \*H_{\tau}^{(\tau)})\*Y^{'(\tau)}}{\|(I - \*H_{\tau}^{(\tau)})\*Y^{'(\tau)}\|}
\end{pmatrix} = \begin{pmatrix}\frac{(I - \*H_{\tau}^{(1)})\*U^{(1)}}{\|(I - \*H_{\tau}^{(1)})\*U^{(1)}\|} \\ \vdots \\ \frac{(I - \*H_{\tau}^{(\tau)})\*U^{(\tau)}}{\|(I - \*H_{\tau}^{(\tau)})\*U^{(\tau)}\|}
\end{pmatrix}.$}
$$
For the given $(\tau, \*M_\tau)$,  $\{\A_\tau(\*Y'),  \B_\tau(\*Y')\}$ is the sufficient statistics for $(\*\mu, \*\sigma)$ and $\C_\tau(\*Y')$ is the remaining normalized ``noises" 
whose distribution is free of nuisance parameters $(\*\mu, \*\sigma)$.  
By its expression, $\C(\*Y')$  depends only on $\*U$, so we also write $\C(\*Y') = \C(\*U)$.
To develop a nuclear mapping for $\tau$ while removing the impact of the nuisance parameters $(\*\mu, \*\sigma)$, we consider below a conditional probability given $\{\A_\tau(\*Y'), \B_\tau(\*Y')\}$, the sufficient statistics~for~$(\*\mu, \*\sigma)$. 

Let $\hat\tau(\*Y') $ be any reasonable estimates of $\tau$ based on $\*Y'.$
For instance, we can use a BIC-based estimator
\begin{align*}
    \hat \tau(\*Y') = \arg\min_{\tau \in {\it \Upsilon}}\left[\left\{-2 \max_{(\*M_\tau, {\bm \mu}, \*\sigma) } \log \ell(\*M_\tau, {\*\mu}, \*\sigma|\*Y')\right\} + 2\tau \log(n)\right],
\end{align*}
where $\ell(\*M_\tau, {\bm \mu}, \*\tau|\*Y')$ is the log-likelihood function of $(\*M_\tau, {\bm \mu}, \*\tau)$ given $\*Y'.$ 
We also denote the conditional probability $ \mathcal P(k|\*a,\*b) = \P\big\{\hat \tau(\*Y') = k \big| \A_\tau(\*Y') = \*a, \B_\tau(\*Y')=\*b\big\}$ and let 
\begin{align*}
    \mathcal F_{\theta}(w|\*a,\*b) = \sum_{\{k: \mathcal P(k|\*a,\*b)  > \mathcal P(w|\*a,\*b) \}} \mathcal P(k|\*a,\*b).
\end{align*}
Conditional on $\big\{A_\tau(\*Y') = \*a, B_\tau(\*Y') = \*b \big\}$, $\mathcal F_{\theta}(w|\*a,\*b)$ does not involve often unknown $(\*\mu, \*\sigma)$, as stated by the lemma below.  
See Appendix I for a proof of the lemma.

\begin{lemma}
\label{lem:F_tau_M}
The function defined above $\mathcal F_{\theta}(w|\*a,\*b)$  is free of $(\*\mu, \*\sigma)$ and so we can express $\mathcal F_{\theta}(w|\*a,\*b) = \mathcal F_{(\tau, M_\tau)}(w|\*a,\*b).$  
\end{lemma}

Next we follow Section 3.2 to remove the impact of the remaining latent discrete nuisance parameter $\*M_\tau$.  
Specifically, let $T_a(\*u, \btheta)
= (\A_\tau(\*y'), \B_{\tau}(\*y'),$ $\C_{\tau}(\*y')),$ where $\*y' =\*M_{\tau} \*\mu + \diag(\*M_\tau  {\*\sigma}) \*u$,  and  $T_a(\*u, \btheta)$ it is a mapping from ${\cal U} \times \Theta \to {\cal T}$. We define a set in $\cal T$
\begin{align*}
     B_\theta(\alpha) & = \left\{(\*a,\*b,\*c): 
    \mathcal F_{\theta}(\hat \tau(\*y^*) |\*a,\*b) \leq \alpha \,\, \text{for} \,\, \*y^* = \*M_\tau \*a+\diag\{ \*M_\tau \*b\} \*c \right\} \nonumber\\
     & = \left\{(\*a,\*b,\*c):  \mathcal F_{(\tau, \*M_\tau) }(\hat \tau(\*y^*) |\*a,\*b) \leq \alpha \,\, \text{for} \,\, \*y^* = \*M_\tau \*a+\diag\{ \*M_\tau \*b\} \*c \right\} \\ & \overset{def}{=} B_{(\tau. \*M_\tau)}(\alpha).
\end{align*}
This set only depends on the given $(\tau, \*M_\tau)$ and is free of nuisance parameters $(\*\mu, \*\sigma)$. Theorem~\ref{the:T_a_mixture} below states that $T_a(\*U, \btheta)$ is inside $B_{(\tau. \*M_\tau)}(\alpha)$ with a probability greater than $\alpha$, both conditionally and unconditionally. 
See Appendix I for a proof. 

\begin{theorem}
\label{the:T_a_mixture}
For any $(\*a,\*b)$, the conditional probability 
\begin{align}
\label{eq:conditional_borel_prob}
    \P\{T_a(\*U, \btheta)  
\in B_{(\tau. \*M_\tau)}(\alpha) | \A_\tau(\*Y') = \*a, \B_\tau(\*Y') = \*b\} \geq \alpha.
\end{align}
Therefore, 
$
\P\{
T_a(\*U, \btheta)
\in B_{(\tau. \*M_\tau)}(\alpha)\} \geq \alpha.     
$
\end{theorem}

For $\*y' = G(\*u, \tau, \*M_\tau, \*\mu, \*\sigma)$ and any $(\tau, \widetilde{\*M_\tau})$, following \eqref{eq:nu}, we define 
\begin{align*}
%\label{eq:nu_mixture}
     \nu(\*u, \tau, \*M_{\tau}, \*\mu, \*\sigma, \widetilde{\*M_\tau})= \tilde \nu(\*y', \tau, \*M_{\tau}, \widetilde{\*M_\tau})  =  \inf_{\alpha'}\{\alpha': (\tau, \widetilde{\*M}_\tau) \in \Gamma_{\alpha'}(\*y')\}, 
\end{align*}
where $\Gamma_\alpha(\*y') = \{(\tau, \*M_\tau): T_a(\*u, \*\theta)  \in B_{(\tau. \*M_\tau)}(\alpha)\} = \{(\tau, \*M_\tau): (\A_\tau(\*y'),$ $\B_{\tau}(\*y'), \C_{\tau}(\*y'))  \in B_{(\tau. \*M_\tau)}(\alpha)\}$. 
Then the (profile) nuclear mapping function for $\tau$ is 
\begin{align*}
    T^P(\*u, \tau, \*M_\tau, \*\mu, \*\sigma) = \widetilde T^P(\*y', \tau) = \min_{\widetilde{\*M}_\tau} \tilde \nu(\*y', \tau, \*M_{\tau}, \widetilde{\*M_\tau}).
\end{align*}

According to Lemma~\ref{lemma:nuisance} and  \eqref{eq:CIeta},
we define
\begin{align}
    \Xi_\alpha(\*y_{obs})  & =  \big\{\tau: \exists \*u^* \in {\cal U}\mbox{ and $(\*M_\tau, \*\mu, \*\sigma)$ s.t. }  \*y_{obs} =  \*\mu + \diag\{\*M_\tau \*\sigma\}\*u^*, \nonumber \\ & \qquad\qquad T^P(\*u^*, \tau, \*M_\tau, \*\mu, \*\sigma) \leq \alpha\big\}
   \label{eq:mix4}
\end{align}
From Theorem~\ref{the:nuisance} and a direct calculation,  the corollary below states that $\Xi_\alpha(\*y_{obs})$ can be further simplified and  is a level $\alpha$ confidence set for $\tau_0$. 
See Appendix I for a proof. 
\begin{corollary}
\label{cor:mix4}
 The set defined in (\ref{eq:mix4}) can be further expressed as
\begin{align}
\Xi_\alpha(\*y_{obs})  = 
       \{\tau: \exists \widetilde{\*M}_\tau, \mathcal F_{\tau, \widetilde{\*M}_\tau}(\hat\tau(\*y^{obs})|\widetilde\A_{\tau}(\*y_{obs}), \widetilde\B_{\tilde\tau}(\*y_{obs})) \leq \alpha \},
      \label{eq:mix5}
       \end{align}
where $\widetilde\A_{\tau}(\*y_{obs})$ and $\widetilde\B_{\tau}(\*y_{obs})$ are defined by replacing $\*M_\tau$ with  $\widetilde{\*M}_\tau$ in the definition of $\A_\tau(\*y_{obs})$ and $\B_\tau(\*y_{obs}).$ Also, 
$\P\big\{\tau_0 \in \Xi_\alpha(\*Y) 
\big\} \ge \alpha$, so $\Xi_\alpha(\*y_{obs})$ is a level $\alpha$ confidence set for $\tau_0$.
\end{corollary}

Corollary~\ref{cor:mix4} suggests that, to construct $\Xi_\alpha(\*y_{obs}),$  we only need to calculate the function $\mathcal F_{\tau, \widetilde{\*M}_\tau}(\hat\tau(\*y^{obs})|\*a, \*b)$ in (\ref{eq:mix5})
for each pair of $(\tau, \widetilde{\*M}_\tau)$ and $\{\*a, \*b\} = \{\widetilde\A_{\tau}(\*y_{obs}), \widetilde\B_{\tilde\tau}(\*y_{obs})\}$.  This  can be achieved with Monte-Carlo approach as follows.
Given $(\tau, \widetilde{\*M}_\tau)$ and $(\*a, \*b),$ we simulate a large number of $\*u^s \sim \*U$ and the collection of $\*u^s$'s forms a set $\mathcal V$.  Let 
\begin{align*}
    \*y^s_{\widetilde{\*M}_\tau, \*a, \*b} = \widetilde{\*M}_\tau\*a+\diag\{\widetilde{\*M}_\tau\*b\} \C_\tau(\*u^s).
\end{align*}
The following theorem validates that we can approximate the conditional distribution of $\*Y$ given $\A_\tau (\*Y) = \*a$ and $\B_\tau(\*Y) = \*b$ with the Monte-Carlo sampling distribution of $\*y^s_{\widetilde{\*M}_\tau, \*a, \*b}.$ See Appendix I for a proof. 
\begin{lemma}
\label{lemma:Y_s_mixture}
Given a $(\tau, \widetilde{\*M}_\tau),$ let $\widetilde{\*Y} =  \widetilde{\*M}_\tau {\*\mu}_{\tau} + \diag(\*U)\widetilde{\*M}_\tau {\*\sigma_\tau}$ and $\*Y^s_{\widetilde{\*M}_\tau, \*a, \*b}= \widetilde{\*M}_\tau\*a+ \diag\{\widetilde{\*M}_\tau\*b\}\*C_\tau(\*U^s),$ where $\*U^s \sim \*U,$ then 
\begin{align*}
    \left\{  \widetilde{\*Y} \middle| \widetilde{\A}_\tau (\widetilde{\*Y}) = \*a, \widetilde{\B}_\tau(\widetilde{\*Y}) =\*b  \right\} \quad \sim \quad \*Y^s_{\widetilde{\*M}_\tau, \*a, \*b}.
\end{align*}
\end{lemma}
When $|\cal V|$ is large enough, the function $\mathcal F_{\tau, \widetilde{\*M}_\tau}(w|\*a,\*b)$ in (\ref{eq:mix5}) can be approximated by the Monte-Carlo version 
\begin{align*}
     \mathcal F_{\tau, \widetilde{\*M}_\tau}^s(w|\*a,\*b) \overset{def}{=} \sum_{\{k: {\mathcal P}_{\mathcal V}^s(k|\*a,\*b)  > {\mathcal P}_{\mathcal V}^s(w|\*a,\*b) \}} {\mathcal P}_{\mathcal V}^s(k|\*a,\*b),
\end{align*}
since based on Lemma~\ref{lemma:Y_s_mixture}, $\mathcal F_{\tau, \widetilde{\*M}_\tau}^s(w|\*a,\*b) \to \mathcal F_{\tau, \widetilde{\*M}_\tau}(w|\*a,\*b),$ as $|\cal V| \to \infty$. 
Here, $\mathcal P_{\mathcal V}^s(k|\*a,\*b) = \frac{1}{|\mathcal V|}\sum_{\*u^s \in \mathcal V}I(\hat \tau(\*y^s_{\widetilde{\*M}_\tau, \*a, \*b}) =k)$ and $\*y^s_{\widetilde{\*M}_\tau, \*a, \*b} = \widetilde{\*M}_\tau\*a+\diag\{\widetilde{\*M}_\tau\*b\} \C_\tau(\*u^s)$, for
$\*u^s \in {\cal V}$. 

\subsection{Finding candidates for $(\tau_0,\*M_0)$} \label{sec:finding_candidates}
We follow  Section~\ref{sec:candidate_general} to construct a data-driven candidate set $\widehat {\it \Upsilon}(\*y_{obs})$ for $\kappa_0 = (
\tau_0, \*M_0)$ based on the observed data $\*y_{obs}.$ 
Corresponding to (\ref{eq:tau-star}), we use  a modified BIC as our lost function and define 
our many-to-one mapping for 
$\kappa^* = (\tau^*, \*M^*_\tau)$ as
\begin{equation}
    \label{eq:modefied_BIC}
   \resizebox{.91\hsize}{!}{$(\tau^*, \*M^*_{\tau^*}) =\arg\min\limits_{(\tau, \*M_\tau)} \min\limits_{(\*\mu, \*\sigma)}  \left\{n\log \bigg(\frac{\|\*y_{obs} - {\*M}_{\tau} {\*\mu}_{\tau} - \diag(\*M_{\tau}  {\*\sigma_\tau}) \*u^*\|^2
+1}{n}\bigg) + 2\lambda\tau\log(n)\right\}.$}
\end{equation}
By \eqref{eq:cand-hat}, a data-dependent candidate set for $\kappa_0 = (\tau_0, \*M_0)$ is 
\begin{align*}
  \resizebox{1.0\hsize}{!}{$\widehat{\it \Upsilon}_{{\cal V}_c}(\*y_{obs}) 
  =  \bigg\{ 
  (\tau^*, \*M^*_{\tau^*})  
  ={\footnotesize \arg\min\limits_{(\tau, \*M_\tau)} \min\limits_{(\*\mu, \*\sigma)}  \bigg\{n\log \left(\frac{\|\*y_{obs} - {\*M}_{\tau} {\*\mu}_{\tau} - \diag(\*u^c)\*M_{\tau}  {\*\sigma_\tau}\|^2
+1}{n}\right) +
\lambda\log(n)(2\tau)
  \bigg\}, 
 \*u^c \in {\cal V}_c
  \bigg\}},$}
\end{align*}
where ${\cal V}_c = \{\*u_1, \ldots, \*u_N\}$ is a collection of simulated $\*u_j^c\sim N(\*0, \*I)$, $j = 1, \ldots, N$.

Then based on (\ref{eq:G1_candidate}) and (\ref{eq:mix5}),
we define the confidence set
\begin{align}
\label{eq:cs_mixture}
\resizebox{.91\hsize}{!}{$\Xi'_{\alpha}(\*y_{obs}) = \left\{\tau: \exists \widetilde{\*M}_\tau, \mbox{s.t. } (\tau, \widetilde{\*M}_\tau) \in \widehat{\it \Upsilon}_{{\cal V}_c}(\*y_{obs}),  \mathcal F_{\tau, \widetilde{\*M}_\tau}(\hat\tau(\*y^{obs})|\widetilde\A_{\tau}(\*y_{obs}), \widetilde\B_{\tilde\tau}(\*y_{obs})) \leq \alpha \right\}.$}
\end{align}

One complication under the mixture model setting is that the equation (\ref{eq:N1}) 
in (N1) may not hold when $\*u^{rel}$ hits certain directions (i.e., the two vectors $\diag(\*M_0\sigma_0) \*u^{rel}$ and $\{I-\*M_{\tau'}(\*M_{\tau'}^\top\*M_{\tau'})^{-1} \*M_{\tau'}^\top\}\*M_0\*\mu_0$ are on the same direction for some $\*M_{\tau'} \not = \*M_0$, $\tau' \leq \tau_0$) or $\|\*u^{rel}\|$ is extremely large. Fortunately, the probabilities of these events are arbitrarily small, resulting in the arbitrarily small difference $\delta$ between the probability bound in the following Theorem~\ref{thm:bound_of_C_d_mixture} and that in Theorem~\ref{thm:p_bound_C_D_general}.
In the theorem, $C_{\min} = \underset{\small \{\*M_\tau: \, |\tau| \leq |\tau_0|, \tau \neq \tau_0\}}{\min} \|(\*I - \*M_\tau)\*M_0\*\mu_0\| > 0.$
See Appendix I for a proof.

\begin{theorem}
\label{thm:bound_of_C_d_mixture}
Suppose $n-\tau_0>4.$ Let the distance metric between $\*U$ and $\*U^*$ be $\rho(\*U, \*U^*) = \frac{\|(\*I - \*H_{\*U^*})\*U\|}{\|\*U^*\|},$ then for any arbitrarily small $\delta >0,$ there exist $\gamma_\delta>0, c>0$ and also an interval of positive width $\Lambda_{\delta} =[\frac{\gamma_{\delta}^{1.5}}{\log(n)/n},$ $\frac{\log\{0.5C^2_{\min}\sigma_{\max}\gamma_{\delta}/\sigma_{\min} + 1\}}{2 \tau_0 \log(n)/n}]$ such that when $\lambda \in \Lambda_{\delta},$ (a) $\P_{\*U | {\cal V}_c}\{(\tau_0,\*M_0)\in \widehat{\it \Upsilon}_{{\cal V}_c}(\*Y)\}$ $\geq 1- \delta - o_p(e^{-c|\mathcal V^c|})$, (b) $\P_{\*U | {\cal V}_c}\{(\tau_0,\*M_0)\in \Xi'_{\alpha}(\*Y)\}  \geq \alpha- \delta - o_p(e^{-c|\mathcal V^c|}),$
where $\Xi'_{\alpha}(\cdot)$ is defined in \eqref{eq:cs_mixture}.
 \end{theorem}

{\bf Remark 4.1}
Compared to the standard BIC, the modified BIC many-to-one mapping
(\ref{eq:modefied_BIC}) has an extreme term that involves the repro copy $\*u^*$. We modified the term inside the log function so that the mapping (\ref{eq:modefied_BIC}) can give us $(\tau_0, M_0)$ when $\*u^*$ equals
$\*u^{rel}$. 
We also 
allow the penalty on the model size to change by using $\lambda\log(n)(2\tau)$ instead of  $\log(n)(2\tau)$. 
This change adds some flexibility to help control the finite sample coverage probability as shown in Theorem~\ref{thm:bound_of_C_d_mixture}.
In practice, we usually still use $\lambda=1,$ since it is often in its required range for 
a small $\delta$. For example, to achieve a coverage error $\delta=O(e^{-n}),$ we can take $\gamma_\delta =O(1/n)$.
Considering $C_{\min}$ is usually $O(n),$ it is likely that $\frac{\gamma_\delta^{1.5}}{\log(n)/n}<1<\frac{\log\{0.5C^2_{\min}\gamma_\delta \sigma_{\max}/\sigma_{\min} + 1\}}{2 \tau_0 \log(n)/n}$.

\section{Numerical Studies: Real and simulated data}
\subsection{Real data analysis: red blood cell sodium-lithium countertransport (SLC) data}

In this section we consider a data set that consists of Red blood cell sodium-lithium countertransport (SLC) activity measurements from 190 individuals \citep{dudley_assessing_1991}. 
SLC measurement is known to be correlated with blood pressure, and perceived as an essential cause of hypertension by some researchers. Moreover, SLC is usually easier to study than blood pressure that could be influenced by numerous environmental and genetic factors  \citep{dudley_assessing_1991}.  \cite{roeder_graphical_1994} and \cite{chen_inference_2012} both analyzed the SLC data set of 190 individuals using a normal mixture model, and perform a hypothesis testing on the number of components $H_0: \tau_0 = k$ vs $H_1: \tau_0>k$ for a small integer $k.$ \cite{roeder_graphical_1994} concluded that $\tau_0=3$ is the smallest value of $\tau_0$ not rejected by the data assuming variance of each component is the same.  
\cite{chen_inference_2012} argued instead that dropping the assumption of equal variances, the hypothesis of $\tau_0=2$ would not be rejected and 
thus a good fit for the data.

We use the repro samples approach developed in Section 4 
to re-analyze the SLC data. Following \cite{roeder_graphical_1994, chen_inference_2012}, we assume that the SLC measurements follow a normal mixture distribution with unknown number of components $\tau_0$ and obtain a 95\% confidence set for $\tau_0$ that is $\{2, 3, 4\}.$  After that, we fit the normal mixture model with $\tau=2,3,4$ on the SLC data with an EM algorithm, respectively.  We provide the estimated mean, standard deviation and weight of each component in Table~\ref{tab:SLC}. 
The confidence set obtained by our repro samples approach suggests that both hypothesis $\tau_0=2$ and $\tau_0=3$ are plausible. Since, in genetics, a two-component mixture distribution corresponds to a simple dominance model and a three-component corresponds to an additive model \citep{roeder_graphical_1994},
evidently we can not rule out either of the two competing genetic models based on the SLC data alone. 

An advantage of the proposed repor samples method is that 
its
confidence set provides both a lower bound and an upper bound for plausible values of $\tau_0$ supported by the data set, whereas inverting one-sided tests proposed by \cite{roeder_graphical_1994} and \cite{chen_inference_2012} can only lead to one sided confidence set. One formidable challenge of finding an upper bound is the fact that any distribution can be arbitrarily well approximated by a normal mixture distribution, therefore $H_0: \tau_0=k$ for a large $k$ can never be rejected. Here we manage to tackle this challenge by using repro samples to find a proper candidate set as described in Section~\ref{sec:finding_candidates}. The upper bound of the confidence set for $\tau_0$ we find for the SLC data is $4.$  In Figure~\ref{fig:SLC_mixture}, we plot the histogram of the SLC measurements together with estimated normal mixture densities with number of components $\tau =2,3,4$ using EM algorithm. It appears that only with $\tau=4$, the normal mixture model could capture the three spikes in the histogram, while still smoothly fitting the rest of the histogram. Indeed, a normal mixture model with four components should demand further investigation, since it effectively represents the data without over fitting it.

\begin{table}[ht]
\centering
\resizebox{\textwidth}{!}{\begin{tabular}{lccc}
  \hline
$\tau_0$ & Component Mean & Component SD & Weight \\ 
  \hline
2 & (0.2206,0.3654) & (0.0571,0.1012) & (0.7057,0.2943) \\ 
  3 & (0.1887,0.4199,0.2809) & (0.0414,0.0886,0.0474) & (0.4453,0.168,0.3866) \\ 
  4 & (0.1804,0.3351,0.2556,0.4403) & (0.0362,0.0359,0.0268,0.086) & (0.4018,0.1742,0.2941,0.1299) \\ 
   \hline
\end{tabular}}
\caption{Estimated mean, standard deviation and weight of each component of the normal mixture distributions for the SLC data.}
\label{tab:SLC}
\end{table}

\begin{figure}[ht]
    \centering
     \includegraphics[width=12cm, height=12cm]{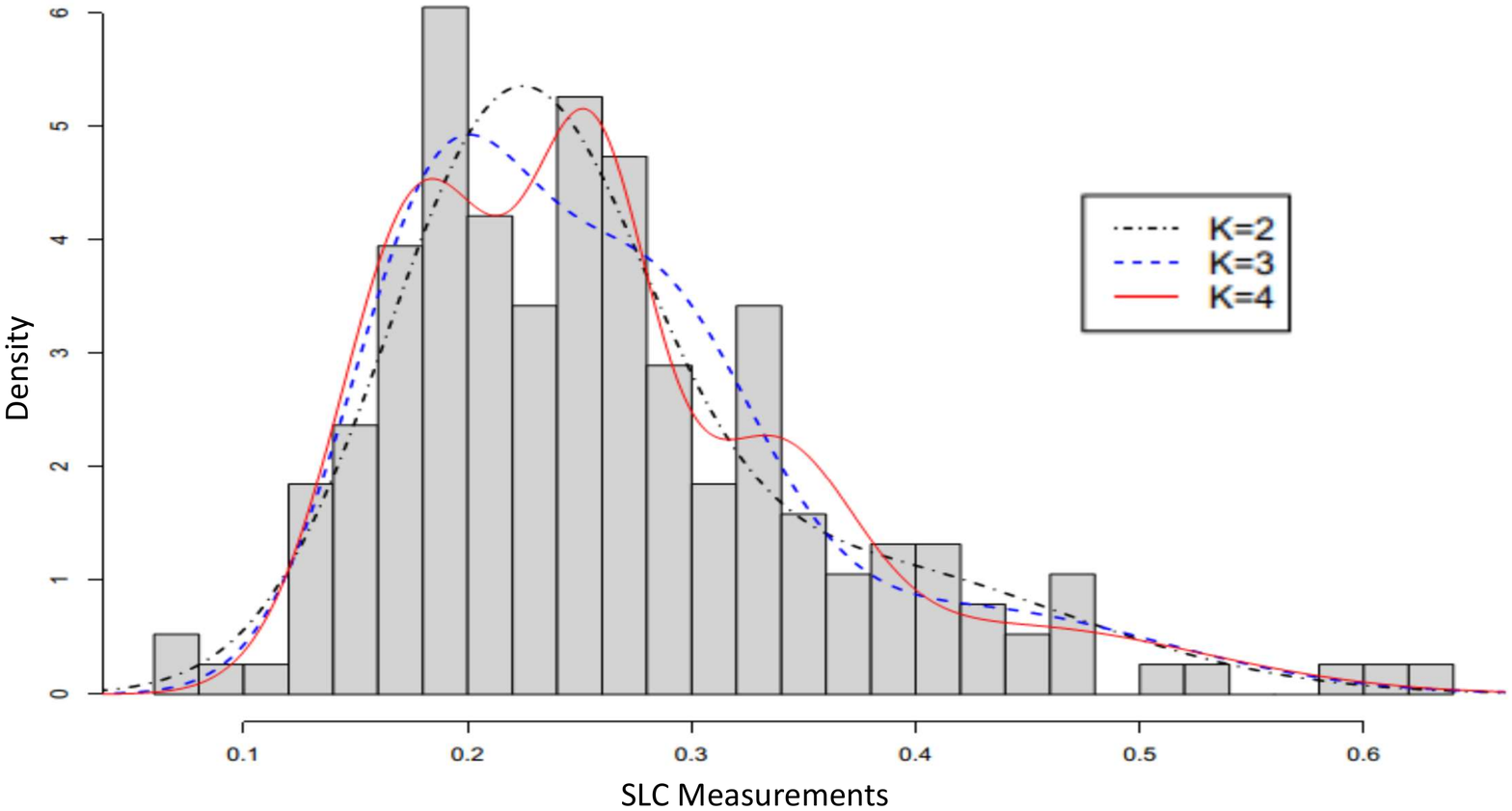}
     \vspace{-1.5cm}
    \caption{Histogram of 190 SLC measurements, including three estimated normal mixture density curves with number of components $\tau_0=2, 3, 4.$}
    \label{fig:SLC_mixture}
\end{figure}

We also conducted a Bayesian analysis on the SLD data
using 
R package {\it mixAK} \citep{komarek2014capabilities}, where 
the Bayesian inference procedure on Gaussian mixture model with unknown number of components is discussed in \cite{richardson_bayesian_1997}.
We implement four different priors on $\tau$ in our analysis: uniform, truncated Poisson(0.4), truncated Poisson(1) and truncated Poisson(5) over $\tau =1, \dots, 10.$ The priors used on $\*\mu$ and $\*\sigma$ are the default priors (i.e., Gaussian and inverse Gamma priors) in {\it mixAK}. 
The posterior distributions for the four priors are plotted in Figure~\ref{fig:SLCposteriors}, with corresponding 95\% credible sets $\{2,3\}$, $\{2\}$, $\{2,3\}$, $\{2,3,4\}$, respectively. 
It turns out 
the posterior distributions and their corresponding 95\% credible intervals (sets) of $\tau$ are quite different for different priors. Even in the cases with the same credible set using uniform and truncated Poisson(1) priors, their posterior distributions are quite different. 
It is clear that the
Bayesian inference of $\tau$ is very sensitive to the choice of prior distribution. A simulation study in the next section affirms this observation. In addition, the simulation study further suggests that the Bayesian procedure is not suited for the task of recovering the true value $\tau_0$ (under repeated experiments) and Bayesian outcomes are also affected greatly by the default priors on $\*\mu$ and~$\*\sigma$.

\begin{figure}[ht]
    \centering
    \includegraphics[width=\textwidth, height= 10cm]{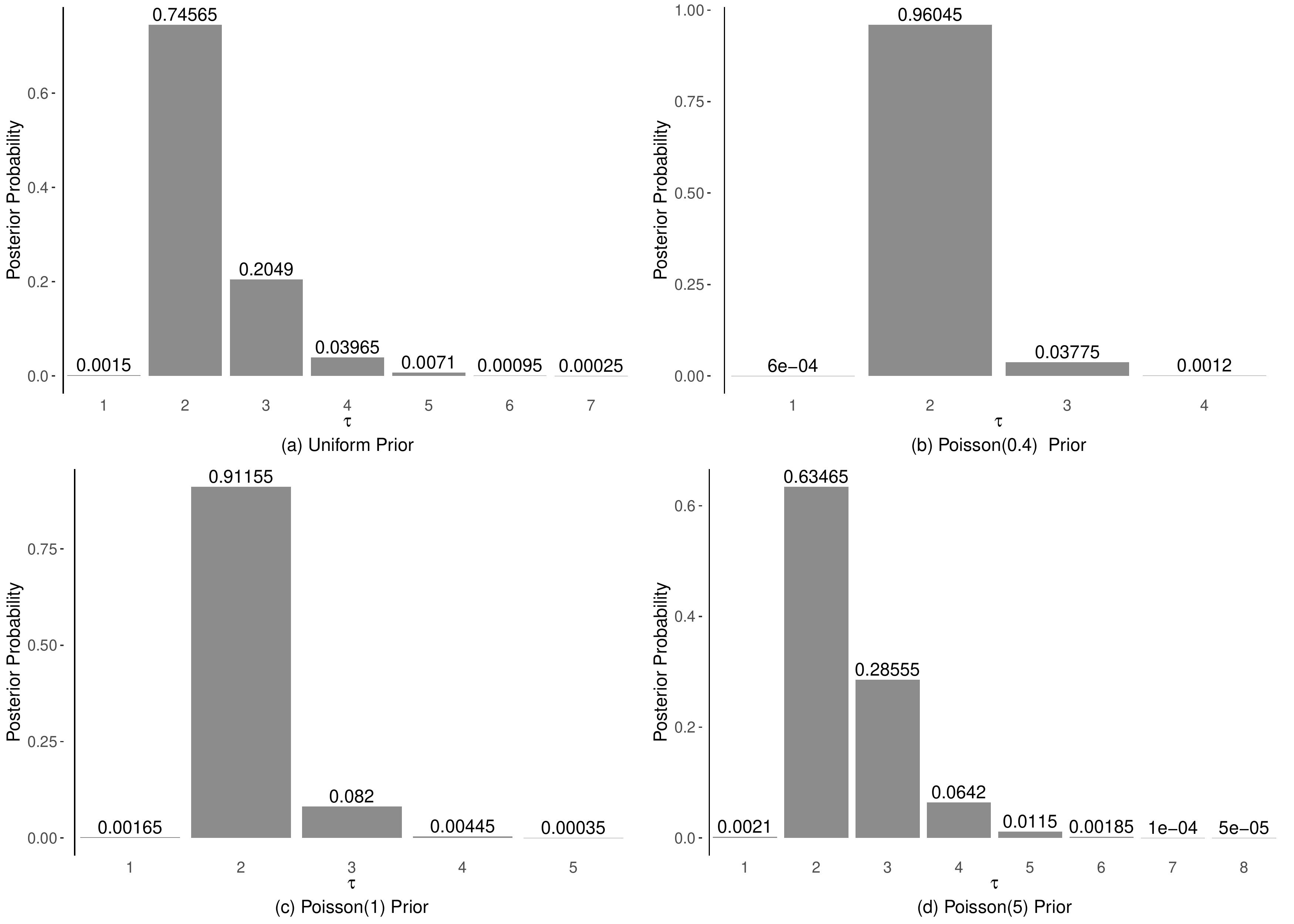}
    \caption{Posterior distributions of $\tau$ under four different priors: (a) uniform; (b) Poisson($\mu = 0.4$); (c) Poisson($\mu = 1$); (d) Poisson($\mu = 5$); All priors are truncated outside the set $\{\tau: \tau = 1, \dots, 10\}$.
    }
    \label{fig:SLCposteriors}
\end{figure}

\subsection{Simulation Study: Inference and Recovery of True $\tau_0$}
To further demonstrate the performance of the proposed confidence set, we carry out simulation studies in which the simulated data are generated using the estimated normal mixture model with $\tau_0=3$ and $\tau_0=4$ in Table~\ref{tab:SLC}, respectively. We keep the sample size at $n = 190$. For each simulated data, we apply the proposed repro samples method. 
We repeat the simulation and analysis 200 times, and subsequently summarize the 200 level-95\% confidence sets using a bar plot in  Figure~\ref{fig:simulation_mixture}, for $\tau_0 =3$ and $\tau_0 = 4$, respectively.
\begin{figure}
    \centering
     \includegraphics[width=0.85\textwidth, height=8.757cm]{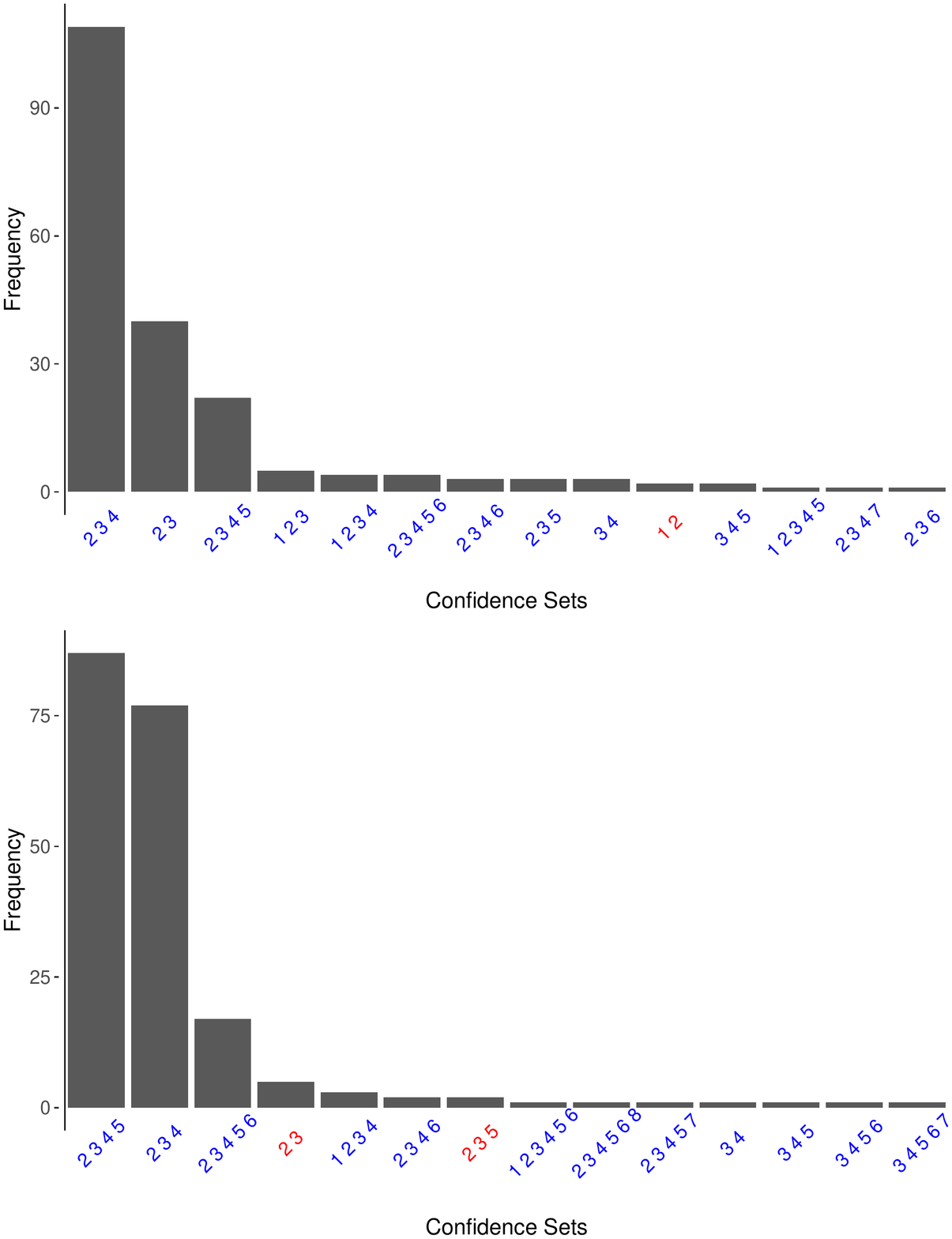}
    \caption{Bar plots of the 200 level 95\% repro samples confidence sets obtained in the simulation study (200 repetitions). 
    The upper panel is for the truth $\tau_0=3$; lower for the truth $\tau_0=4.$}
    \label{fig:simulation_mixture}
\end{figure}
It appears that for both settings of $\tau_0=3$ and $\tau_0=4$, the proposed confidence set for $\tau_0$ achieves  coverage rates of 99\% and 96.5\% respectively,  both greater than 95\%. Note that the parameter space for $\tau_0$ is discrete, therefore,  unlike confidence intervals for continuous parameters, it is not always possible to achieve an exact confidence level of 95\%. For simulated data with $\tau_0=3,$ the average size of the proposed confidence set is about 3, with a standard error of 0.05. In fact, we observe from Figure~\ref{fig:simulation_mixture} that for the majority of the 200 simulations, our approach produces a confidence set of $\{2,3,4\}.$ As for $\tau_0=4,$ more than 80\% of the times, we obtain a confidence set of either $\{2,3,4\}$ or $\{2,3,4,5\}$, making the average size of the confidence set around 3.65, with a standard error of 0.05. 

We compare the performance of the existing methods under the same simulation settings. First, the classical BIC method recovers the true value of $\tau_0 = 3$ only 13 times (6.5\%) out of the 200 repetitions.  
The number of times 
drops to $0$ (out of 200 repetitions) when the true $\tau_0=4.$ The results are not surprising, since with penalty terms point estimators of $\tau_0$ are often biased towards smaller $\tau$ values in practice, even though BIC estimator is a consistent estimator. 
These two extremely low rates of correctly recovering the true cluster sizes highlights the challenge of the traditional point estimation method in this simulation example, and more generally the risk of using just a point estimation in a normal mixture model. This result also suggests the necessity of  quantifying the uncertainty of the estimation using a confidence set with desired coverage rate. 

Additionally, when it comes to make an inference for $\tau_0$, we also illustrate the advantage of our proposed confidence set over the existing one-sided hypothesis testing approach in the literature. In practice, researchers would usually favor the smallest model that can not be rejected, as discussed in \cite{roeder_graphical_1994} and \cite{chen_inference_2012}. However, we observe from our simulation study that such practices could very well erroneously lead us to an overly small model, resulting in incorrect scientific interpretation. In particular, we find in our simulation study that, when the true $\tau_0=3$ and we perform the PLR test of \cite{chen_inference_2012} on $H_0: \tau_0=2$ versus 
$H_1: \tau_0 = 3$,
only 74\% times do we fail to reject $H_0: \tau_0=2$. When the true $\tau_0=4$ and we perform the PLR test on $H_0: \tau_0=3$ versus $H_1: \tau_0 = 4$,
we fail 93\% times  to reject $H_0: \tau_0=3$. These results  suggest that most of time the PLR test would favor a wrong model over the true one. Furthermore, inverting either of the one-sided tests does not produce a useful confidence set, since the upper bound of the confidence set is $\infty$. In conclusion, we believe that the proposed confidence set by the repro samples method is a more intact approach of conducting inference on $\tau_0$ compared to the one that  uses only hypothesis testing.

\begin{figure}[ht]
    \centering
    \includegraphics[width=\textwidth, height= 6cm]{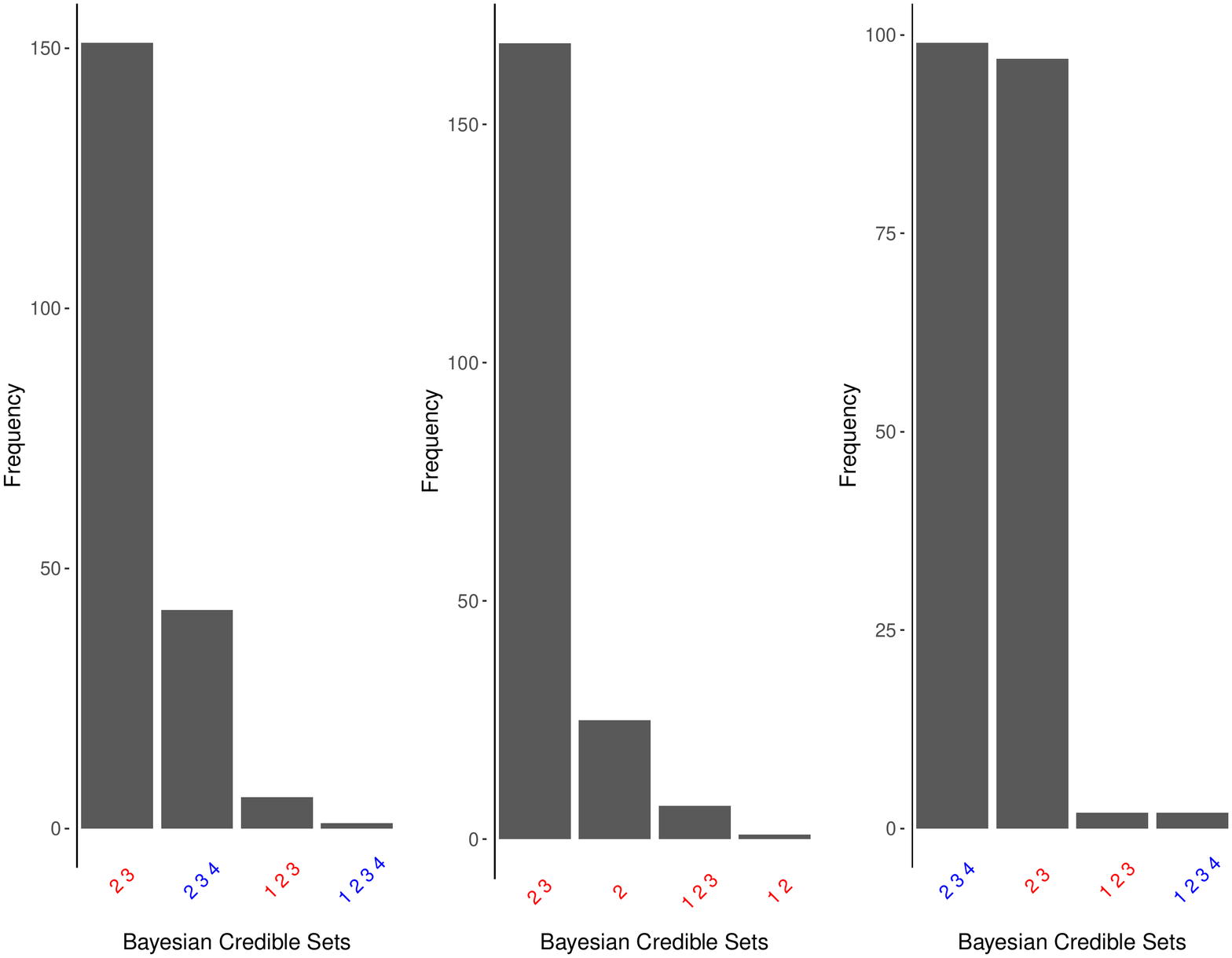}
    (a)\hspace{4.8cm}  (b)  \hspace{4.8cm} (c)
    \caption{Bar plots of the 200  95\% Bayesian credible sets obtained in the simulation study (200 repetitions) with $\tau_0 = 4$. Prior choices are (a) uniform; (b) Poisson($\mu = 1$); (c) Poisson($\mu = 5$); All priors truncated over $\{\tau=1, \dots, 10\}$.}
    \label{fig:bayes_simulation}
\end{figure}

Finally, we compare the performance of the Bayesian procedure described in \cite{richardson_bayesian_1997} on analysis of the Gaussian mixtures with unknown number of components. Figure~\ref{fig:bayes_simulation} summarizes the results from the Bayesian inference procedure produced by R package {\it mixAK} \citep{komarek2014capabilities}. The data is generated with the same setting as the lower panel of Figure~\ref{fig:simulation_mixture} with $\tau_0=4.$ The left panel of Figure~\ref{fig:bayes_simulation} illustrates the results by imposing uniform prior on $\tau = 1, \dots, 10.$ The middle panel displays results from Poisson(1) prior, which prefers smaller $\tau$, and the right panel from Poisson(5) prior, which favors the truth $\tau_0=4$ since the mode of Poisson(5) is 4. 
We can see that the Bayesian level-$95\%$ credible sets consistently underestimate and fail to cover $\tau_0$. With the uniform prior, the credible sets only covers $\tau_0$ 43 times out of the 200 simulations (21.5\%). While with a truncated Poisson(1) prior that favors smaller $\tau,$ the Bayesian approach completely fails to cover $\tau_0$ with none of the credible intervals including $\tau_0.$  Even with an ``oracle" truncated Poisson(5) prior that sides with the truth, the Bayesian approach  manages to cover the truth only $50.5\%$ of the times. An additional issue with the Bayesian method is that the outcomes are sensitive to the specification of the prior on $\tau,$ as evidenced by the differences among the three panels of Figure~\ref{fig:bayes_simulation}.

The underestimation of $\tau_0$ by the Bayesian method with a uniform prior for $\tau$ is a little surprising at the first appearance. However, further investigation reveals that the priors on $\*\mu$ and $\*\sigma$ also have a large impact on the estimation of $\tau$ and 
the underestimation is attributable to the shrinkage effect of implementing a multilevel hierarchical model (see \cite{richardson_bayesian_1997}  for details). Indeed, it is well known that shrinkage is ubiquitous when parameters are modeled hierarchically in Bayesian analysis. Here both the means and the variances of each component would share a common prior, thus over promoting smaller $\tau$ as observed across all four plots in Figure~\ref{fig:bayes_simulation}. 

In summary, we can conclude from the simulation study that the proposed repro samples method is effective and can always recover the true values of $\tau_0$ using confidence sets with guaranteed coverage rates. All existing methods, including the point estimation, existing hypothesis testing method and the Bayesian procedure, however regularly miss the true $\tau_0$, with outcomes almost always biased toward smaller $\tau$ values. Furthermore, credible sets by the Bayesian method do not have the repeated coverage rates and its outcomes are also very sensitive to the choices of priors on $\tau$ and $(\*\mu, \*\sigma)$.

\section{Further comments and discussion}

This article develops a general and effective framework, named repro samples method, to making inference for a wide range of problems, including many difficult open questions for which solutions were previously unavailable or could not be easily
obtained.
It utilizes artificial sample sets generated 
by mimicking the observed data to quantify inference uncertainty. 
Our development has explicit mathematical solutions in many settings, while others it uses a Monte-Carlo algorithm to achieve numerical solutions. The repro samples method 
has several theoretical guarantees. Theoretical discussions in the article include guaranteed frequentist coverage in various situations in both finite and large sample settings, relationship and advantages to the classical hypothesis testing procedure and associated optimality results, effective handling of nuisance parameters, and use of a data-driven candidate set to improve computing efficiency, etc.
We also discuss relationships to other 
artificial sample based approaches across Bayesian, fiducial and frequentist inferences and to that of the probabilistic IM method, with additional details provided in Appendix II.  
A distinct advantage of the proposed framework is that it does not need to rely on likelihood or large sample CLT to develop inference. It is particularly effective for a broad~collection of 
challenging inference problems, especially those involving discrete~target~parameters.

The repro samples method is broadly applicable. This article includes a number of examples to illustrate several different aspects of the development.  
It further includes a case study example to construct a finite-sample confidence set for the unknown number of components in a mixed normal model, which is a well-known unsolved problem in statistical inference. Although it appears simple, the case study example requires a number of procedures to deal with the discrete target parameter and also a large number of nuisance parameters of varying lengths. 
Beside the examples in this article, we have also applied the repro samples method to solve another difficult problem on post-selection inference in high dimensional models where the model selection uncertainty are quantified and also incorporated in the inference of regression coefficients. Due to space limitation, we summarize the results of the post-selection inference problem in a separate paper, in \cite{wang2022highdimenional}.  
Additional research involving discrete parameters (e.g., identification of root node(s) of a network, network membership, unknown number of clusters/mixtures, etc.) and also rare events data are  currently underway. The research results will be reported in separate articles.

Indeed, the repro samples method can greatly extend the scope and reach of statistical inference by considering broad types of parameters and models. 
Besides the regular type of  parameter space $\Theta$ that is continuous, the parameter space $\Theta$ in a repro samples method can be discrete or even mixed types. It sidesteps some regularity conditions commonly imposed in a traditional approach (e.g., the true $\*\theta_0$ is an inner point of parameter space $\Theta$;  \citealp[Chapters 7 and 8]{tCAS90a}). 
Furthermore, the algorithmic model (\ref{eq:1}) and its extension (\ref{eq:AA}) are very general. In addition to covering the traditional Fisherian model specifications of using density functions (as illustrated throughout the paper), 
they cover complex models such as those  specified in a sequence of differential equations in population genetics, astronomy, geology, etc., in which one  can simulate a data set from the model given the model parameter~$\btheta$. 
With the fast growth of data science, there is a pressing need to develop and expand the theoretical foundation for statistical inference and prediction. The development of the repro samples method is trans-formative and we anticipate that it will 
serve as a stepping stone to establish theoretical foundations for statistical inference and prediction 
in data science.

\bibliographystyle{agsm_nourl}
\bibliography{ref4all}

\appendix
\section*{Appendix I: Theoretical Proofs} 
\subsection*{I-A. Proofs of Theoretical Results in Section~\ref{sec:general}}
\begin{proof}[Proof of Theorem~\ref{thm:1}]

Since $\*Z = G_z( \btheta_0, \*U)$ always holds,  
Theorem~\ref{thm:1} follows from \eqref{eq:B} and 
    $\P\{\btheta_0 \in \Gamma_\alpha(\*Z)\} = \P\{ \*Z = G_z( \btheta_0, \*U),$ $T(\*U, \btheta_0) \in B_\alpha(\btheta_0)\} = \P\{T(\*U, \btheta_0) \in B_\alpha(\btheta_0)\}.$ 
\end{proof}

\begin{proof}[Proof of Corollary 1]

 Under $H_0$, we have $\theta_0 \in \Theta_0$ and thus
    \begin{align*} \P\{p(\*Z) \leq \gamma\}
    \leq  \P \left[\inf\left\{\alpha': \theta_0 \in \Gamma_{\alpha'}(\*Z)\right\} \geq 1 - \gamma \right]  
    = \P \left\{ \theta_0 \not\in \Gamma_{\alpha' }(\*Z),  %\mbox{ for any }  
    \forall \, \alpha' < 1 - \gamma \right\} \leq \gamma,
    \end{align*}
    for any significance level  $\gamma.$ Thus, the Type I error is controlled. 
\end{proof}

\begin{proof}[Proof of Theorem~\ref{thm:2}]

Since $g(\*Z, \btheta_0, \*U) =0$ always holds, Theorem~\ref{thm:2} follows from \eqref{eq:B} and 
%\begin{align*}
    $\P\{\btheta_0 \in \Gamma_\alpha(\*Z)\} = \P\{ g(\*Z, \btheta_0, \*U) =0, T(\*U, \btheta_0) \in B_\alpha(\btheta_0)\} = \P\{T(\*U, \btheta_0) \in B_\alpha(\btheta_0)\}.$ 
% \end{align*}
\end{proof}

\subsection*{I-B. Proofs of Theoretical Results in Section~\ref{sec:T}}

\begin{proof}[Proof of Lemma~\ref{Lemma:1}]
\begin{eqnarray}
\Gamma_{\alpha}(\*z_{obs}) 
& = & \big\{\btheta: \exists \*u^* \in {\cal U} \mbox{ s.t. }  \*z_{obs} = G_z({\btheta}, \*u^*), T(\*u^*, \btheta)  \in B_{\alpha}(\btheta) \big\}  \nonumber \\
& = &\big\{\btheta: \exists \*u^* \in {\cal U} \mbox{ s.t. }  \*z_{obs} = G_z({\btheta}, \*u^*), \tilde T(G_z({\btheta}, \*u^*), \btheta)  \in B_{\alpha}(\btheta) \big\} \nonumber
\\
& = &\big\{\btheta: \exists \*u^* \in {\cal U} \mbox{ s.t. }  \*z_{obs} = G_z({\btheta}, \*u^*), \tilde T(\*z_{obs}, \btheta)  \in B_{\alpha}(\btheta),  \big\} \nonumber % \label{eq:connect1}
\\
 & = &\big\{\btheta:  \*z_{obs} = G_z({\btheta}, \*u^*),  \exists \, \*u^* \in {\cal U} \big\} 
 \cap \big\{\btheta:  \tilde T(\*z_{obs}, \btheta)  \in B_{\alpha}(\btheta) \big\}
\nonumber
\\
& \subseteq & \big\{\btheta:  \tilde T(\*z_{obs}, \btheta)  \in B_{\alpha}(\btheta) \big\} \overset{denote}{=} \tilde \Gamma_{\alpha}(\*z_{obs}), \nonumber 
%\label{eq:connect}
\end{eqnarray}
The two sets  
$\Gamma_{\alpha}(\*z_{obs}) = \tilde \Gamma_{\alpha}(\*z_{obs})$, when $\tilde \Gamma_{\alpha}(\*z_{obs}) \subset \big\{\btheta:  \*z_{obs} = G_z({\btheta}, \*u^*), 
\exists \, \*u^* \in {\cal U} \big\}$.
\end{proof}

\begin{proof}[Proof of Corollary \ref{cor:UMA}]

 Let $\btheta_0$ be the true parameter with $\*z_{obs} = G_z(\btheta_0, \*u^{rel})$ (realized version) and $\*Z = G_z(\btheta_0, \*U)$ (random version). Let $\btheta' \not = \btheta_0$ be a false parameter value. 

(a) Since $\tilde
\Gamma_\alpha(\*Z)$ is an UMA confidence set, by the definition of UMA \citep[][Section 9.2.1]{tCAS90a}, we have 
$$
\P\left\{\btheta' \in \widetilde
\Gamma_\alpha(\*Z) \right\} \leq \P\left\{ \btheta' \in C_\alpha(\*Z) \right\},
$$
where $C_\alpha(\*z_{obs})$ is any level-$\alpha$ confidence set. However, by Lemma~\ref{Lemma:1}, we have $\Gamma_\alpha(\*Z) \subseteq \widetilde \Gamma_\alpha(\*Z)$. It follows that
$$
\P\left\{\btheta' \in 
\Gamma_\alpha(\*Z) \right\} \leq 
\P\big\{\btheta' \in \widetilde
\Gamma_\alpha(\*Z) \big\} \leq \P\left\{ \btheta' \in C_\alpha(\*Z) \right\}.
$$
So, the repro sample level-$\alpha$ confidence set $\Gamma_\alpha(\*z_{obs})$ is also UMA. 

(b) Since $\widetilde \Gamma_\alpha(\*Z)$ is unbiased, we have $\P\left\{\btheta' \in \widetilde \Gamma_\alpha(\*Z)\right\} \leq \alpha $, for any $\btheta' \not = \btheta_0$. By Lemma~\ref{Lemma:1}, we have $\Gamma_\alpha(\*Z) \subseteq \widetilde \Gamma_\alpha(\*Z)$, it follows that $\P\left\{\btheta' \in \Gamma_\alpha(\*Z)\right\} \leq \P\left\{\btheta' \in \widetilde \Gamma_\alpha(\*Z)\right\} \leq \alpha$, for any $\btheta' \not = \btheta_0$. So $\widetilde \Gamma_\alpha(\*Z)$ is also unbiased. The UMA part of proof is the same as in (a). 
\end{proof}

% \subsection*{Proof of Lemma~\ref{lemma:nuisance}}
\begin{proof}[Proof of Lemma~\ref{lemma:nuisance}] 

By \eqref{eq:nu}, we have $\P\left\{T^P(\*U, \btheta) \leq \alpha \right\} \ge \P\{\btheta \in \Gamma_{\alpha}(G_z(\*U, \btheta))\} \ge \alpha,$ 
where $\Gamma_{\alpha}(G_z(\*U, \btheta))$ is constructed with $T_a(\*u, \btheta) $ as in \eqref{eq:nu}.
\end{proof}

\begin{proof}[Proof of Theorem~\ref{the:nuisance}] 
By Lemma~\ref{lemma:nuisance},
    $\P\left\{\eta_0 \in \Xi_\alpha(\*Z)\right\} \geq \P\{T^P(\*U, \theta_0)\leq \alpha\} \geq \alpha.$
\end{proof}

\begin{proof}[Proof of Lemma~\ref{lem:nuisance_depth}]
Similar to Corollary 1, let $B_\alpha(\btheta) = \{\*t: F_{\mathcal V|D}\big(D_{{\cal S}_{\theta}}(\*t)\big) \geq 1-\alpha\}.$ Then $\Gamma_\alpha(\*z)$ in \eqref{eq:nu} is  
\begin{align*}
    \Gamma_\alpha(\*z) = \left\{\btheta: \exists \*u^*, \mbox{such that } z = G_z(\*u^*,\btheta),  F_{\mathcal V|D}\left(D_{{\cal S}_{\theta}}\big(T_a(\*u^*, \btheta)\big) \right)  \geq 1-\alpha \right\}.
\end{align*}
Therefore,
$\inf\limits_{\alpha'}\{\alpha': \widetilde\btheta \in \Gamma_{\alpha'}(G_z(\*u, \btheta))\}  =\inf\limits_{\alpha'}\big\{\alpha': \exists \*u^* \mbox{ s.t. }  G_z(\*u,\btheta) = G_z(\*u^*,\btheta),   \alpha \geq 1-F_{\mathcal V|D}\big(D_{{\cal S}_{\theta}}\big(T_a(\*u^*, \btheta)\big) \big) \big\} = \inf\limits_{\{\*u^*: G_z(\*u^*, \tilde\theta)= G_z(\*u, \tilde\theta)\}} \big\{1- F_{\mathcal V|D}\big(D_{{\cal S}_{\tilde\theta}}\big(T_a(\*u^*, \tilde\btheta) \big)\big)\big\}$.
\end{proof}

\begin{proof}[Proof of Corollary~\ref{cor:cand_cs_coverage}] 
By the fact that $\P\big\{T(\*U, \btheta_0) \in B_{\alpha}(\btheta_0)\big\} \geq \alpha,$ we have 
\begin{align*}
   \P\big\{T(\*U, \btheta_0) \in B_{\alpha}(\btheta_0)\big\} &  = \P\big\{T(\*U,\btheta_0) \in B_{\alpha}(\btheta_0), \btheta_0 \in \widehat\Theta\big\} +  \P\big\{T(\*U,\btheta_0) \in B_{\alpha}(\btheta_0), \btheta_0 \not\in \widehat\Theta\big\}  \\
   &  = \P\big\{\theta_0 \in \Gamma_{\alpha}'(\*Z)\big\} +  \P\big\{T(\*U,\btheta_0) \in B_{\alpha}(\btheta_0), \btheta_0 \not\in \widehat\Theta\big\} \geq \alpha.
\end{align*}
It then follows that
   $\P\big\{\theta_0 \in \Gamma_{\alpha}'(\*Z) \big\} \geq \alpha - \P\big\{T(\*U,\btheta_0) \in B_{\alpha}(\btheta_0), \btheta_0 \not\in \widehat\Theta\big\} \geq \alpha - \P\big\{\theta_0 \not\in \widehat\Theta\big\}$
and thus Corollary~\ref{cor:cand_cs_coverage}.
\end{proof}

\begin{proof}[Proof of Lemma~\ref{lemma:p_bound_C_D_general}]
Let $\P_{\*U, \mathcal V^*}$ denote the joint probability function with respect to $\*U$ and $\{\*U^*: \*U^* \in \mathcal V^*\},$ 
then  
{\small
\begin{align*}
& \P_{\*U, \mathcal V^*}\big\{\kappa_0 \not\in \widehat{\it \Upsilon}_{{\cal V}^*}(\*Z)\big\}  = \P_{\*U, \mathcal V^*}\big\{\kappa_0 \not\in \widehat{\it \Upsilon}_{{\cal V}^*}(\*Z), \bigcap_{\*U^* \in \mathcal V^*}\{\*U^* \not\in S_{\cal N}(\*U)\}\big\} \\ & \qquad \qquad  
 +  \P_{\*U, \mathcal V^*}\big\{\kappa_0 \not\in \widehat{\it \Upsilon}_{{\cal V}^*}(\*Z), \bigcup_{\*U^* \in \mathcal V^*}\{\*U^* \in S_{\cal N}(\*U)\}\big\}\\ & \qquad   = \P_{\*U, \mathcal V^*}\left\{\kappa_0 \not\in \widehat{\it \Upsilon}_{{\cal V}^*}(\*Z), \bigcap_{\*U^* \in \mathcal V^*}\{\*U^* \not\in S_{\cal N}(\*U)\}\right\} \leq \P_{\*U, \mathcal V^*}\left\{ \bigcap_{\*U^* \in \mathcal V^*}\{\*U^* \not\in S_{\cal N}(\*U)\}\right\}\\
& \qquad = E_{\*U}\left[\P_{\mathcal V^*|U}\left\{\bigcap_{\*U^* \in \mathcal V^*}\{\*U^* \not\in S_{\cal N}(\*U)\}\middle |\*U\right\}\right] =  E_{\*U}\left[\left\{1- \P_{\*U^*|\*U}\left(\*U^* \in S_{\cal N}(\*U)\middle |\*U\right)\right\}^{|\mathcal V^*|}\right]\\
& \qquad \leq \{1- E_{\*U}\P_{\*U^*|\*U}\left(\*U^* \in S_{\cal N}(\*U)\middle |\*U\right)\}^{|\mathcal V^*|} = \left\{1- \P_{\*U, \*U^*}\left(\*U^* \in S_{\cal N}(\*U)\right)\right\}^{|\mathcal V^*|} =  (1-P_{\cal N})^{|\mathcal V^*|}, \numberthis \label{eq: 1-P_delta}
\end{align*}
}
where the last inequality follows from Jensen's inequality. 
\end{proof}

\begin{proof}[Proof of Theorem~\ref{thm:p_bound_C_D_general}]
	By Lemma~\ref{lemma:p_bound_C_D_general} and Markov Inequality, 
	\begin{align*}
	\P_{\cal V^*}\left[\P_{\*U|\mathcal V^*}\{\kappa_0 \not\in \widehat{\it \Upsilon}_{{\cal V}^*}(\*Z)\} \geq (1-P_{\cal N})^{|\mathcal V^*|/2}\right] & \leq \frac{E_{\mathcal V^*}\P_{\*U|\mathcal V^*}\{\kappa_0 \not\in \widehat{\it \Upsilon}_{{\cal V}^*}(\*Z)\}}{(1-P_{\cal N})^{|\mathcal V^*|/2}}\\
	& \leq \frac{\P_{\*U,\mathcal V^*}\{\kappa_0 \not\in \widehat{\it \Upsilon}_{{\cal V}^*}(\*Z)\}}{(1-P_{\cal N})^{|\mathcal V^*|/2}} = (1-P_{\cal N})^{|\mathcal V^*|/2}.	
	\end{align*}
Theorem~\ref{thm:p_bound_C_D_general}(a) then follows by making $c<-\frac{1}{2}\log(1- P_{\cal N}).$ 

Theorem ~\ref{thm:p_bound_C_D_general}(b) follows immediately from Theorem~\ref{thm:p_bound_C_D_general}(a) Corollary~\ref{cor:cand_cs_coverage}.
\end{proof}
% \section*{Appendix III: Theoretical Proofs for Section~\ref{sec:mixture}}
\subsection*{I-C. Proofs of Theoretical Results in Section~\ref{sec:mixture}}
% \subsection*{Proof of Lemma~\ref{lem:F_tau_M}}

\begin{proof}[Proof of Lemma~\ref{lem:F_tau_M}]
Because given $(\tau, {\*M}_\tau)$,   $(\A_\tau(\*Y), \B_\tau(\*Y))$ are minimal sufficient for $(\*\mu, \*\sigma)$ and $\C_\tau(\*Y) = \C_\tau(\*U) $ is ancillary, it follows from the Basu's Theorem that 
$$\{\A_\tau(\*Y),  \B_\tau(\*Y)\} \perp \C_\tau(\*Y).$$ Also, when given  $\{\A_\tau(\*Y),  \B_\tau(\*Y)\} = (\*a, \*b),$ we have 
%\begin{align*}
 $\*Y %&  
% = \*M_\tau\A_\tau(\*Y) +  \diag\{\*M_\tau\B_\tau(\*Y)\}\C_\tau(\*Y)  
=\*M_\tau \*a+ \diag\{\*M_\tau \*b\}\C_\tau(\*Y) 
 %\\  & 
    = \*M_\tau \*a+ \diag\{\*M_\tau \*b\}\C_\tau(\*U).$
%\end{align*} 
Therefore,
%the conditional distribution of $\*Y|\A_\tau(\*Y)=\*a,\B_\tau(\*Y)=\*b$ is 
\begin{align}
\label{eq:Y_conditional_mixture}
   % \{\*M_\tau \*a+ \diag\{\*M_\tau \*b\}\C_\tau(\*U)|\A_\tau(\*Y)=\*a,\B_\tau(\*Y)=\*b\} 
   \*Y|\A_\tau(\*Y)=\*a,\B_\tau(\*Y)=\*b 
   \sim \*M_\tau \*a+ \diag\{\*M_\tau \*b\}\C_\tau(\*U),
\end{align}
which does not involve $(\*\mu, \*\sigma).$ We can then conclude that the conditional distribution $\hat\tau(\*Y)|\A_\tau(\*Y)=
\*a,\B_\tau(\*Y)=\*b$ is free of $(\*\mu, \*\sigma),$ from which Lemma~\ref{lem:F_tau_M} follows immediately. 
\end{proof}

% \subsection*{Proof of Theorem~\ref{the:T_a_mixture}}
\begin{proof}[Proof of Theorem~\ref{the:T_a_mixture}]
By the definition of $F_{(\tau, \*M_\tau)} (f| \*a,\*b)$, %for any $\*a,\*b$, 
we have 
%\begin{align*}  & 
   $\P\big\{\mathcal F_{(\tau, \*M_\tau)} (\hat\tau(\*Y)| \*a,\*b) \leq \alpha |\A_\tau(\*Y)=\*a, \B_\tau(\*Y) =\*b\big\} %\\& 
   = \sum_{\{f: \mathcal F_{(\tau, \*M_\tau)} (f| \*a,\*b) \leq \alpha \}} \P\left\{\hat\tau(Y) = f|\A_\tau(\*Y)=\*a, \B_\tau(\*Y) =\*b\right\}
    \geq \alpha,$
%\end{align*}
which implies \eqref{eq:conditional_borel_prob}. It then follows from the above that
$$\P\left\{\mathcal F_{(\tau, \*M_\tau)} (\hat\tau(Y)| \A_\tau(\*Y), \B_\tau(\*Y)) \leq \alpha |\A_\tau(\*Y), \B_\tau(\*Y) \right\} \geq \alpha.$$
Therefore, we have 
%\begin{align*}
 $   \P\big\{\tilde T_a(\*Y, \tau, \*M_\tau) \in B_{(\tau, \*M_\tau)}(\alpha)\big\} = \P\big\{\mathcal F_{(\tau, \*M_\tau)} (\hat\tau(Y)| \A_\tau(\*Y), \B_\tau(\*Y)) \leq \alpha \big\}  = \mathbb{E}_{\A_\tau(\*Y),\B_\tau(\*Y) }\big[\P\big\{\mathcal F_{(\tau, \*M_\tau)} (\hat\tau(Y)| \A_\tau(\*Y), \B_\tau(\*Y)) \leq \alpha |\A_\tau(\*Y), \B_\tau(\*Y) \big\}\big] \geq \alpha.
$ %\end{align*}
\end{proof}

\begin{proof}[Proof of Corollary~\ref{cor:mix4}]
By \eqref{eq:CIeta}, a level-$\alpha$ repro sampling confidence set based on the observed data $\*y_{obs}$ is
\begin{align}
    \Xi_\alpha(\*y_{obs})  & =  \big\{\tau: \exists \*u^* \in {\cal U}\mbox{ and $(\*M_\tau, \*\mu, \*\sigma)$ such that }  \*y_{obs} =  \*\mu + \diag\{\*M_\tau \*\sigma\}\*u^*, \nonumber \\ & \qquad\qquad T^P(\*u^*, \tau, \*M_\tau, \*\mu, \*\sigma) \leq \alpha\big\}\nonumber \\
    & = \left\{\tau: \widetilde T^P(\*y_{obs}, \tau) \leq \alpha\right\}
    \label{eq:tilde_t_y}\\
    &= 
    \left\{\tau: \min_{\widetilde{\*M}_\tau} \inf_{\alpha'}\{\alpha': (\tau, \widetilde{\*M}_\tau) \in \Gamma_{\alpha'}(\*y_{obs})\}
    \leq \alpha\right\}\nonumber\\
    & = \left\{\tau: \exists \widetilde{\*M}_\tau,  (\tau, \widetilde{\*M}_\tau) \in \Gamma_{\alpha}(\*y_{obs})\right\} \nonumber \\
     & =  \{\tau: \exists \widetilde{\*M}_\tau, \widetilde T_a(\*y_{obs}, \tau, \widetilde{\*M}_\tau) \in B_{\tau, \widetilde{\*M}_\tau}(\alpha)\} \nonumber\\
 & =   \{\tau: \exists \widetilde{\*M}_\tau, \mathcal F_{\tau, \widetilde{\*M}_\tau}(\hat\tau(\*y^{obs})|\widetilde\A_{\tau}(\*y_{obs}), \widetilde\B_{\tilde\tau}(\*y_{obs})) \leq \alpha \}.\nonumber
%  & \approx   \{(\tau, M_\tau): \mathcal F_{\mathcal V}(\hat\tau(\*y^{obs})|\A_\tau(\*y_{obs}), \B_\tau(\*y_{obs})) \leq \alpha \},
\end{align}
The equation in \eqref{eq:tilde_t_y} holds, since $\big\{\tau: \exists \*u^* \in {\cal U}\hbox{ and $(\*M_\tau, \*\mu, \*\sigma)$ such that }  \*y_{obs} =  \*\mu + $ $\diag\{\*M_\tau \*\sigma\} \*u^* \big\} = \{1, \ldots, n\}$ and $T^P(\*u, \tau, \*M_\tau, \*\mu, \*\sigma) = \widetilde T^P(\*y^*, \tau) = \widetilde T^P(\*y_{obs}, \tau)$ under matching $\*y^* = \*y_{obs}$ where $\*y^* = \*\mu + \diag\{\*M_\tau \*\sigma\} \*u^*$. 
\end{proof}

% \subsection*{Proof of Theorem~\ref{the:Y_s_mixture}}
\begin{proof}[Proof of Lemma~\ref{lemma:Y_s_mixture}]
By \eqref{eq:Y_conditional_mixture}, 
\begin{align*}
              \{\widetilde{\*M}_\tau \*a+ \diag\{\widetilde{\*M}_\tau \*b\}\widetilde{\C}_\tau(\*U)|\widetilde{\A}_\tau(\widetilde{\*Y})=\*a,\widetilde{\B}_\tau(\widetilde{\*Y})=\*b\} \sim \widetilde{\*M}_\tau \*a+ \diag\{\widetilde{\*M}_\tau \*b\}\widetilde{\C}_\tau(\*U),
\end{align*}
from which Lemma~\ref{lemma:Y_s_mixture} follows immediately.
\end{proof}

% \subsection*{Proof of Lemma~\ref{lem:bound_single_ustar}}
In order to prove Theorem~\ref{thm:bound_of_C_d_mixture}, we will introduce several other technical lemmas first. 
\begin{lemma}
\label{lem::angle} Let $\rho(\*v_1, \*v_2) = \frac{\*v^\top_1\*v_2}{\|\*v_1\|\|\*v_2\|}$ and  $\rho_{\*M_\tau^\bot}(\*v_1, \*v_2) = \frac{\*v_1^\top (I-\*H_\tau) \*v_2}{\|(I-\*H_\tau)\*v_1\|\|(I-\*H_\tau) \*v_2\|}$ for any vectors $\*v_1, \*v_2$. Suppose $\tau < n$. Let $\*W = \diag(\*U)\*M_0 \*\tau_0,$ then for any $0 \leq \gamma_1, \gamma \leq 1,$ 
\begin{align}
\label{eq: angle_U_M0}
    P_{\*U}\left\{\rho_{\*M_\tau^\bot}(\*U, \*M_0\*\mu_0)< \gamma_1^2\right\}>  1-   2 \{\arccos (\gamma_1)\}^{n-|\tau|-1},
\end{align}
and 
\begin{align}
\label{eq:angle_U_Ustar}
    P_{(\*U, \*U^*)}\{\rho(\*W, \*U^*) > 1-\gamma^2\} > \frac{\gamma^{n-2}\arcsin (\gamma)}{n-1}. 
\end{align}
Further more $\rho_{\*M_\tau^\bot}(\*U, \*M_0\*\mu_0)$ and $\rho(\*W, \*U^*)$ are independent. 
\end{lemma}

\begin{proof}[Proof of Lemma~\ref{lem::angle}]
The inequality \eqref{eq: angle_U_M0} follows from Lemma 2 of \cite{wang2022highdimenional}.
 
For \eqref{eq:angle_U_Ustar}, we carry out a spherical transformation on $\*U^*.$ Then conditional on $\*W = \*w$,  we can show that 
\begin{align}\label{eq:u_epsi}
    & %P_{ \epsilon}\left\{\|(I-P_{\*u})\epsilon\|^2 <\gamma^2 \|\epsilon\|^2 \bigg| \*u \right\}
    %\nonumber\\& \qquad = 
    P_{\*U^*|\*W}\left\{{\|\*W^T \*U^*\|}\big/{(\|\*W\|\| \*U^*\|)} > \sqrt{1 - \gamma^2} \bigg| \*W^*=\*w \right\}\nonumber\\
    & \qquad =P_{ \psi}\left\{  |\sin(\psi)|  < \gamma  \right\} \nonumber \\
    & \qquad = \frac{2}{c_1}\int_0^{\arcsin \gamma} \sin^{n-2}(s) ds  \nonumber \\
    & \qquad >  \frac{2}{c_1}\int_0^{\arcsin \gamma} (\frac{s \gamma }{\arcsin \gamma})^{n-2} ds> \frac{\gamma^{n-2}\arcsin \gamma}{n-1}, 
\end{align}
where $\psi = \psi(\*U^*, \*w)$ (or $\pi - \psi$) is the angle between $\*U^*$ and $\*w$ and the normalizing constant $c_1 = \int_0^{\pi} \sin^{n-2} (\psi) d\psi = 2 \int_0^{\frac \pi 2} \sin^{n-2} (\psi) d\psi$. 
Because \eqref{eq:u_epsi} does not involve $\*U^*$ and $\*w$, we have
\begin{align*}
  P_{(\*U, \*U^*)}\big\{ \rho(\*U, \*W) > 1- \gamma^2 \big\} & =  E_{\*W}\left[ P_{\*U^*|\*W}\left\{{\|\*W^T \*U^*\|}\big/{(\|\*W\|\| \*U^*\|)} > \sqrt{1 - \gamma^2} \bigg| \*W^*\right\}\right]
  \\ & > \frac{\gamma^{n-2}\arcsin \gamma}{n-1}.
\end{align*} 
The above statement also suggests that $\rho(\*U^*,\*W)$ and $\*U$ are independent, hence $\rho(\*U^*,\*W)$ and   $\rho_{\*M_\tau^\bot}(\*U, \*M_0\*\mu_0)$ are independent.

% \begin{align*}
%     P_{ \epsilon}\left\{ \rho(\*U, \epsilon) > 1- \gamma^2 \bigg| \, \*U \in A(\gamma^2_1, \tau)  \right\} = P_{(\*U, \epsilon)}\big\{ (\*U, \epsilon) \in B(\gamma^2) \big\}, 
% \end{align*}
% thus $\big\{ (\*U, \epsilon) \in B(\gamma^2) \big\}$ and $\big\{ \*U \in A(\gamma^2_1, \tau) \big\}$ are independent. 

% {\color{blue}Further by the spherical transformation of $\epsilon$, $\|\epsilon\|^2$ and $\psi$ in \eqref{eq:u_epsi} are independent. Because $\big\{ (\*U, \epsilon) \in B(\gamma^2) \big\} \bigcap \big\{ \*U \in A(\gamma^2_1, \tau) \big\}$ are only decided by $\psi$, $\|\epsilon\|^2$ are independent with $\big\{ (\*U, \epsilon) \in B(\gamma^2) \big\} \bigcap \big\{ \*U \in A(\gamma^2_1, \tau) \big\}$.} 
\end{proof} 

\begin{lemma}\label{lemma:u_d_sufficent}
Let $g_{\*M_\tau}(\*U)= \frac{\|(I-H_{\*M_\tau})\*U\|}{\|\*U\|},$  $\tilde\gamma_1=(1-\sqrt{\tilde\gamma})\gamma_1 - \sqrt{2-2\sqrt{1-\tilde\gamma^2}}$ and $ \tilde\gamma = \frac{\sigma_{\max}}{\sigma_{\min}} \gamma.$  If $g_{\*M_\tau}(\*U) > \sqrt{\tilde\gamma}, \rho_{\*M_\tau^\bot}(\*U, \*M_0\*\mu_0)< \tilde \gamma_1^2$ and $\rho(\*W, \*U^*) > 1-\gamma^2,$ then
\begin{align*}
    \rho_{\*M_\tau^\bot}(\*U^*, \*M_0\*\mu_0)<  \gamma_1^2.
\end{align*}

% \begin{align*}
% C^1_{\gamma_1, \gamma} = \left\{\rho_{\*M_\tau^\bot}(\*U, \*M_0\*\mu_0)< \tilde \gamma_1^2,\rho(\*W, \*U^*) > 1-\gamma^2, \right\}
% \end{align*}

% \begin{align*}
%     \{\rho_{\*M_\tau^\bot}(\*U, \*M_0\*\mu_0)< \tilde \gamma_1^2,\rho(\*W, \*U^*) > 1-\gamma^2 \} \subset \{\rho_{\*M_\tau^\bot}(\*U^*, \*M_0\*\mu_0)< \gamma_1^2\}
% \end{align*}    

%     Let $$E(\gamma_1, \gamma)= \left\{ \max_{\*M_\tau \neq \*M_0, \tau\leq \tau_0} \rho_{\*M_\tau^\bot}(\*U, \*M_0\*\mu_0) < \gamma_1^2, \rho(\*U, \*w) > 1- \gamma^2  \right\},$$
% for any $0<\gamma_1, \gamma <1.$ Then under Condition C2., for $\*U_i \sim N(0,\*I), i=1, \dots, d$ and $\*w \sim N(0,\diag(\*M_0\*\tau_0^{\circ 2})),$ 
% \begin{align}
% \label{eq:bound_E_compliment}
% \begin{split}
%       & P_{({\cal U}^d, \*w)}\left( \bigcap_{i=1}^d \left[\left\{(\*U_i, \*w) \in E(\gamma_1, \gamma)\right\}^C\right] \right)\\
%       &   \leq  \left\{1- \frac{\gamma^{n-2}\arcsin (\gamma)}{n-1}\right\}^d  + 4(\arccos \tilde\gamma_1)^{n-\tau_0-1} n^{\tau_0} + 2\left(\frac{\sigma^2_{\max}\gamma\tau_0}{\sigma^2_{\min}(n-\tau_0) }\right)^{n-\tau_0-1} n^{\tau_0} \\
%       & = \bar P(\gamma_1, \gamma,d).
% \end{split}
% \end{align}
% where  $\tilde\gamma_1=(1-\sqrt{\gamma})\gamma_1 - \sqrt{2-2\sqrt{1-\gamma^2}}.$ 
\end{lemma}

\begin{proof}[Proof of Lemma~\ref{lemma:u_d_sufficent}]
% With a slight abuse of notation, for any $\*M_\tau$, we let $\rho(\*w, \*M_\tau) = \frac{\|H_{\*M_\tau}\*w\|^2}{\|\*w\|^2}$ be the square of cosine of the angle between $\*w$ and the space spanned by $\*M_\tau. $

% We first show that for any $\*U \sim N(0,1),$ we have 
% \begin{align}
% \label{eq:tilde}
%       P\left( \left\{(\*U, \*w) \in E(\gamma_1, \gamma)\right\}^C \middle|\*w, \max_{\*M_\tau \neq \*M_0, \tau\leq \tau_0} \rho_{\*M_\tau^\bot}(\*U, \*M_0\*\mu_0) < \tilde\gamma_1^2, \rho(\*w, \tau) < 1- \gamma^2 \right) \nonumber \\
%       =  1-  P\left\{\rho(\*U, \*w) > 1- \gamma^2\right\}
%     \leq 1- \frac{\gamma^{n-2}\arcsin (\gamma)}{n-1}.
% \end{align}
First, by  $\rho(\*W, \*U^*) > 1-\gamma^2,$
\begin{align*}
    \gamma^2 > \frac{\|(\*I- \*H_{\*W})\*U^*\|^2}{\|\*W\|^2} > \frac{\sigma^2_{\min} \|(\*I -\*H_{\*U})\*U^*\|^2}{\sigma^2_{\max}\|\*U\|^2} = \frac{\sigma^2_{\min}}{\sigma^2_{\max}}\{1 - \rho(\*U, \*U^*)\}.
\end{align*}
Therefore $\rho(\*U, \*U^*) > 1 - \frac{\sigma^2_{\max}}{\sigma^2_{\min}} \gamma^2 = 1- \tilde\gamma^2.$

 It then follows from the above that
\begin{align*}
     & \frac{1}{\|(I-H_{\*M_\tau}){\*U^*}\|}{\*U^*}^{T}(I-H_{\*M_\tau})\*M_0\*\mu_0\\ 
     & = \frac{1}{\|(I-H_{\*M_\tau})\*U\|}\*U^T(I-H_{\*M_\tau})\*M_0\*\mu_0 +  \left(\frac{1}{\|\*U\|} - \frac{1}{\|(I-H_{\*M_\tau})\*U\|}\right)(\*U^T(I-H_{\*M_\tau})\*M_0\*\mu_0\\
     & \hspace{2cm} + \left({\*U^*}^T/\|{\*U^*}\|-\*U^T/\|\*U\|\right)(I-H_{\*M_\tau})\*M_0\*\mu_0   \\
     & \hspace{4cm} +  \left(\frac{1}{\|(I-H_{\*M_\tau}){\*U^*}\|}-\frac{1}{\|{\*U^*}\|}\right){\*U^*}^T(I-H_{\*M_\tau})\*M_0\*\mu_0\\
     & \leq \frac{\|(I-H_{\*M_\tau})\*U\|}{\|\*U\|}\frac{1}{\|(I-H_{\*M_\tau})\*U\|}\*U^T(I-H_{\*M_\tau})\*M_0\*\mu_0  \\
     & \hspace{2cm} +  \left\|\frac{{\*U^*}^T}{\|{\*U^*}\|}-\frac{\*U^T}{\|\*U\|}\right\|\|(I-H_{\*M_\tau})\*M_0\*\mu_0\| \\
     &  \hspace{4cm} + \frac{\|{\*U^*}\|-\|(I-H_{\*M_\tau}){\*U^*}\|}{\|{\*U^*}\|}\frac{1}{\|(I-H_{\*M_\tau}){\*U^*}\|}{\*U^*}^T(I-H_{\*M_\tau})\*M_0\*\mu_0 \\
     &  \leq g_{\*M_\tau}(\*U)\tilde\gamma_1\|(I-H_{\*M_\tau})\*M_0\*\mu_0 \|  +  \sqrt{2-2\frac{{\*U^*}^t \*U}{\|{\*U^*}^t\| \|\*U\|}}\|(I-H_{\*M_\tau})\*M_0\*\mu_0\| \\ 
     & \hspace{2cm} +(1-g_{\*M_\tau}({\*U^*}))\frac{1}{\|(I-H_{\*M_\tau}){\*U^*}\|}{\*U^*}^T(I-H_{\*M_\tau})\*M_0\*\mu_0\\
     &   \leq g_{\*M_\tau}(\*U) \tilde\gamma_1\|(I-H_{\*M_\tau})\*M_0\*\mu_0 \|  +  \sqrt{2-2\sqrt{1-\tilde\gamma^2}}\|(I-H_{\*M_\tau})\*M_0\*\mu_0\|\\
     & \hspace{2cm}+ (1-g_{\*M_\tau}({\*U^*}))\frac{1}{\|(I-H_{\*M_\tau}){\*U^*}\|}{\*U^*}^T(I-H_{\*M_\tau})\*M_0\*\mu_0.
     %&  = \gamma_1\|(I-H_{\*M_\tau})X_0 \beta_0\|,
\end{align*}
%where $g_{\*M_\tau}({\*U^*})= \|(I-H_{\*M_\tau}){\*U^*}\|/\|{\*U^*}\|$ and $g_{\*M_\tau}(\*U)= \|(I-H_{\*M_\tau})\*U\|/\|\*U\|.$  
It then follows that
\begin{align*}
    & \frac{1}{\|(I-H_{\*M_\tau}){\*U^*}\|}{\*U^*}^T(I-H_{\*M_\tau})\*M_0\*\mu_0\\ 
    & \leq  \frac{g_{\*M_\tau}(\*U)}{g_{\*M_\tau}({\*U^*})} \tilde\gamma_1\|(I-H_{\*M_\tau})\*M_0\*\mu_0\|  + \frac{1}{g_{\*M_\tau}({\*U^*})} \sqrt{2-2\sqrt{1-\tilde\gamma^2}}\|(I-H_{\*M_\tau})\*M_0\*\mu_0\|.
\end{align*}
Further because $\|(I-H_{\*M_\tau}){\*U^*}\|  \leq \|(I-H_{\*M_\tau} P_{\*U}){\*U^*}\|$, if $\rho({\*U^*}, \*U)  > 1- \tilde\gamma^2,$ we have
\begin{align*}
    g_{\*M_\tau}({\*U^*}) & \leq \frac{\|(I-P_{\*U}){\*U^*}\|}{\|{\*U^*}\|} + \frac{\|(P_{\*U}-H_{\*M_\tau} P_{\*U}){\*U^*}\|}{\|{\*U^*}\|} \\
    & \leq \tilde\gamma + \frac{\|(I-H_{\*M_\tau})P_{\*U}{\*U^*}\|}{\|P_{\*U}{\*U^*}\|} = \tilde\gamma + g_{\*M_\tau}(\*U).
\end{align*}
Similarly, we can show that $g_{\*M_\tau}(\*U)\leq  g_{\*M_\tau}({\*U^*}) + \tilde\gamma.$
Then 
\begin{align*}
\frac{g_{\*M_\tau}({\*U^*})}{g_{\*M_\tau}(\*U)}\geq 1 - \frac{\tilde\gamma}{g_{\*M_\tau}(\*U)}.     
\end{align*}
It then follows that a sufficient condition for $\frac{1}{\|(I-H_{\*M_\tau}){\*U^*}\|}{\*U^*}^T(I-H_{\*M_\tau})\*M_0\*\mu_0 \leq \gamma_1$ is 
\begin{align*}
\tilde\gamma_1\leq \left(1 - \frac{\tilde\gamma}{g_{\*M_\tau}(\*U)}\right) \gamma_1 - \sqrt{2-2\sqrt{1-\tilde\gamma^2}}\leq  \frac{g_{\*M_\tau}({\*U^*})}{g_{\*M_\tau}(\*U)}\gamma_1 - \sqrt{2-2\sqrt{1-\tilde\gamma^2}}.     
\end{align*}
Because $g_{\*M_\tau}(\*U) = \sqrt{1-\rho^2(\*U, \tau) } > \sqrt{\tilde\gamma},$ the above inequality holds for $\tilde\gamma_1=(1-\sqrt{\tilde\gamma})\gamma_1 - \sqrt{2-2\sqrt{1-\tilde\gamma^2}}.$ 
\end{proof}

\begin{lemma}
\label{lem:bound_single_ustar}
Suppose $n-\tau_0>4$ and let $ \widetilde\gamma = \frac{\sigma_{\max}}{\sigma_{\min}} \gamma$, then for any $\gamma>0$ such that $2.1 \widetilde\gamma^{1/4}<1$, and $\lambda \in [\frac{\gamma^{1.5}}{\log(n)/n},\frac{\log\{0.5C^2_{\min}\widetilde\gamma + 1\}}{2 \tau_0 \log(n)/n}],$ %under Condition C1. and C2.
\begin{align}
\label{eq:prob_bound_C_D_simple}
       P\left\{(\tau_0, \*M_0) \not\in \widehat{\it \Upsilon}_{{\cal V}^*}(\*Y)\right\} &\leq \sum^{\tau_{\max}}_{\tau=\tau_0+1}
    \exp\left\{-\frac{2(\tau-\tau_0)}{3\sigma^2_{\min}\sqrt{\gamma}}    + 0.25n +  \tau\log(n) \right\} \nonumber \\
    & \quad \quad  +  \sum_{K=1}^{\tau_0} \exp\left\{-\frac{\sigma^{1/2}_{max}}{15\sigma^{5/2}_{\min}\sqrt\gamma} C_{\min}\ + 0.25n + \tau\log(n)\right\}\nonumber\\
     & \quad \quad  +  2n^{3\tau_0/2} \widetilde\gamma^{\frac{n-\tau_0}{2}-1}
     + 4(20 \widetilde\gamma)^{\frac{n-\tau_0-1}{4}} n^{\tau_0} + \left(1-\frac{\gamma^{n-1}}{n-1}\right)^{|\mathcal V^*|}. 
\end{align}
\end{lemma}

\begin{proof}[Proof of Lemma~\ref{lem:bound_single_ustar}]
Let $RSS(\*M_\tau, \*U^*) = \|(\*I-\*H_{\*M_\tau, \*U^*})\*Y\|^2$ and $RSS(\*M_0, \*U^*) = \|(\*I-\*H_{\*M_0, \*U^*})\*Y\|^2$. Then we have
\begin{align*}
    & P(\hat{\*M}^*  = \*M_\tau) \\
    & \leq P\left\{RSS(\*M_\tau, \*U^*) - e^{2\lambda(\tau_0-\tau)\log(n)/n} RSS(\*M_0, \*U^*) + 1 -e^{2\lambda(\tau_0-\tau)\log(n)/n}\leq 0 \right\}\\
    & = P\big\{ \|(\*I-\*H_{\*M_\tau, \*U^*})\*M_0 \*\mu_0\|^2 + \|(\*I-\*H_{\*M_\tau, \*U^*})\*W\|^2 + 2 (\*W)^T (\*I-\*H_{\*M_\tau, \*U^*})\*M_0 \*\mu_0 \\
    & \quad  \qquad  \qquad \qquad - e^{2\lambda(\tau_0-\tau)\log(n)/n} \|(\*I-\*H_{\*M_0, \*U^*})\*W\|^2 +1 -e^{2\lambda(\tau_0-\tau)\log(n)/n} < 0\big\}.
\end{align*}

Let $C_{\*M_0, \*M_\tau} = \|(\*I - \*M_\tau)\*M_0\*\mu_0\|.$ Then for $\tau > \tau_0,$ , when $\lambda \geq \frac{\gamma^{1.5}}{\log(n)/n},$ we have
\begin{align*}
     & P\big\{\hat{\*M}^* = \*M_\tau, \rho(\*W, \*U^*) > 1-\gamma^2\big\} \\
    & \leq P\big\{ - e^{2\lambda(\tau_0-\tau)\log(n)/n} \|(\*I-\*H_{\*M_0, \*U^*})\*W\|^2+1 -e^{2\lambda(\tau_0-\tau)\log(n)/n} < 0,\rho(\*W, \*U^*) > 1-\gamma^2 \big\}\\
    & \leq P\big\{ - e^{2\lambda(\tau_0-\tau)\log(n)/n} \|(\*I-\*H_{\*U^*})\*W\|^2+1 -e^{2\lambda(\tau_0-\tau)\log(n)/n} < 0,  \rho(\*W, \*U^*) > 1-\gamma^2 \big\}\\
    & \leq P\big\{ - e^{2\lambda(\tau_0-\tau)\log(n)/n}\gamma^2\|\*W\|^2 +1 -e^{2\lambda(\tau_0-\tau)\log(n)/n} < 0, \rho(\*W, \*U^*) > 1-\gamma^2 \big\}\\
    & = P\bigg\{ \|\*W\|^2 > \frac{1 -e^{2\lambda(\tau_0-\tau)\log(n)/n}}{ e^{2\lambda(\tau_0-\tau)\log(n)/n}\gamma^2} , \rho(\*W, \*U^*) > 1-\gamma^2 \bigg\}\\
    & \leq P\bigg\{ \|\*W\|^2 > \frac{1 -e^{2\lambda(\tau_0-\tau)\log(n)/n}}{ e^{2\lambda(\tau_0-\tau)\log(n)/n}\gamma^2} \bigg\}\\
    & \leq P\bigg\{ \|\*W\|^2 > \frac{ -2\lambda(\tau_0-\tau)\log(n)/n}{ \gamma^2} \bigg\}\\
    & \leq P\bigg\{ \|\*W\|^2 > \frac{ -2(\tau_0-\tau)}{ \sqrt{\gamma}} \bigg\}\\
    & =  P\left\{\|\*W\|^2  >  C^I_{\gamma, \tau}\right\}. \numberthis \label{eq:CI_gamma}
 \end{align*}
 
 The second-to-last inequality holds because $\frac{1-x}{x}\geq  -\log(x)$ for any $0 < x \leq 1.$

Because $\|\*W\|^2$ follows a weighted chi-square distribution, then for any $c>0$, we can bound the probability
\begin{align}
\label{eq:markov_ineuqality}
    P(\|\*W\|^2 > c)  \leq E \exp\{t\|\*W\|^2 -tc\} 
    = \prod_{i=1}^N\left(\frac{1}{1-2\sigma^2_it}\right)^{1/2}e^{-tc}.
\end{align}

Make $t = 1/(3\sigma_{\min}),$ then for any $\*M_\tau, \tau > \tau_0$
\begin{align*}
%\label{eq:bound_tau_large}
    P(\hat{\*M}^*  = \*M_\tau) \leq \exp\left\{-\frac{1}{3\sigma^2_{\min}} C^I_{\gamma, \tau} + 0.25n \right\}.
\end{align*}

For $\tau \leq \tau_0,$  we would need $\lambda$ to satisfy
% \begin{align*}
%     C^2_{\*M_0, \*M_\tau}   >  e^{2\lambda(\tau_0-\tau)\log(n)/n} -1. 
% \end{align*}
\begin{align}
\label{eq:lambda_upper}
    \lambda \leq \frac{\log\{C^2_{\min}(1- \gamma^2_1)/2 + 1\}}{2 \tau_0 \log(n)/n}.
\end{align}
Then on the event $\{\rho(\*W, \*U^*) > 1-\gamma^2, \rho_{\*M_\tau^\bot}(\*U^*, \*M_0\*\mu_0)<  \gamma_1^2\},$ we have
\begin{align*}
%\label{eq:bound_c_\tau_small}
%\begin{split}
    & P(\hat{\*M}^*  = \*M_\tau, \rho(\*W, \*U^*) > 1-\gamma^2, \rho_{\*M_\tau^\bot}(\*U^*, \*M_0\*\mu_0)<  \gamma_1^2) \\
    & \leq P\big\{\|(\*I-\*H_{\*M_\tau, \*U^*})\*M_0 \*\mu_0\|^2  + 2 (\*W)^T (\*I-\*H_{\*M_\tau, \*U^*})\*M_0 \*\mu_0  \\
    & \quad  \qquad  \qquad \qquad   - e^{2\lambda(\tau_0-\tau)\log(n)/n} \|(\*I-\*H_{\*M_0, \*U^*})\*W\|^2 +1 -e^{2\lambda(\tau_0-\tau)\log(n)/n} < 0,\\
        & \quad  \qquad  \qquad \qquad  \qquad \rho(\*W, \*U^*) > 1-\gamma^2, \rho_{\*M_\tau^\bot}(\*U^*, \*M_0\*\mu_0)<  \gamma_1^2\big\}\\
   & \leq P\big\{\|(\*I-\*H_{\*M_\tau, \*U^*})\*M_0 \*\mu_0\|^2  + 2 (\*W)^T (\*I-\*H_{\*M_\tau, \*U^*})\*M_0 \*\mu_0  \\
    & \quad  \qquad  \qquad \qquad   - e^{2\lambda(\tau_0-\tau)\log(n)/n} \|(\*I-\*H_{\*M_0, \*U^*})\*W\|^2 +1 -e^{2\lambda(\tau_0-\tau)\log(n)/n} < 0\big\}\\
     & \leq P\big\{(1 - \gamma^2_1)\|(\*I-\*H_{\*M_\tau})\*M_0 \*\mu_0\|^2  - 2 \|(\*I-\*H_{\*M_\tau, \*U^*})(\*W)\| \|(\*I-\*H_{\*M_\tau, \*U^*})\*M_0 \*\mu_0\|  \\
    & \quad  \qquad  \qquad \qquad   - e^{2\lambda(\tau_0-\tau)\log(n)/n} \|(\*I-\*H_{\*M_0, \*U^*})\*W\|^2 +1 -e^{2\lambda(\tau_0-\tau)\log(n)/n} < 0\big\}\\
    & \leq P\big\{(1 - \gamma^2_1)\|(\*I-\*H_{\*M_\tau})\*M_0 \*\mu_0\|^2  - 2 \gamma\|\*W\| \|(\*I-\*H_{\*M_\tau})\*M_0 \*\mu_0\|  \\
    & \quad  \qquad  \qquad \qquad   - e^{2\lambda(\tau_0-\tau)\log(n)/n} \gamma^2\|\*W\|^2 +1 -e^{2\lambda(\tau_0-\tau)\log(n)/n} < 0\big\}\\
     & \leq P\bigg\{{\|\*W\| \gamma}  > \\
     &  \frac{-C_{\*M_0, \*M_\tau}  + \sqrt{ C^2_{\*M_0, \*M_\tau}  +e^{2\lambda(\tau_0-\tau)\log(n)/n} \{(1- e^{2\lambda(\tau_0-\tau)\log(n)/n} )+ (1- \gamma_1^2) C^2_{\*M_0, \*M_\tau}\}   } }{e^{2\lambda(\tau_0-\tau)\log(n)/n} }\bigg\}\\
     & \leq P\bigg\{{\|\*W\|}  >  \frac{-C_{\*M_0, \*M_\tau}  + \sqrt{ C^2_{\*M_0, \*M_\tau}  + (1- \gamma_1^2) C^2_{\*M_0, \*M_\tau}/2   } }{\gamma }\bigg\}\\
     & \leq P\bigg\{\|\*W\| > C_{\min} \frac{-1 + \sqrt{1 + (1-\gamma_1^2)/2}}{\gamma}\bigg\}\\
     & \leq P\bigg\{\|\*W\| > C_{\min} \frac{(1-\gamma_1^2)}{5\gamma}\bigg\}\\
     & = P\left\{\|\*W\|  >  C^{II}_{\gamma_1,\gamma, \tau}\right\}\\
     & \leq \exp\left\{-\frac{1}{3\sigma^2_{\min}} C^{II}_{\gamma,\gamma_1, \tau}\ + 0.25n \right\}, 
%\end{split}
\end{align*}
where the last inequality follows from \eqref{eq:markov_ineuqality} and making $t = 1/(3\sigma_{\min}).$ The second to last inequality holds because $-1+\sqrt{1+x} \geq 0.4x$ for any $0 \leq x \leq 1.$ 

 %$\alpha_{\gamma_1}=(e^{2\lambda(\tau_0-\tau)\log(n)/n} -1)/(\sqrt{1-\gamma_1^2} C^2_{\*M_0, \*M_\tau}),$

 It then follows from the above and Lemma~\ref{lemma:u_d_sufficent} that for any $\*M_\tau, \tau \leq \tau_0$
 \begin{align*}
    & P\big\{\hat{\*M}^* = \*M_\tau, \rho(\*W, \*U^*) > 1-\gamma^2\} \\
    &  = P\big\{\hat{\*M}^* = \*M_\tau, \rho(\*W, \*U^*) > 1-\gamma^2, \rho_{\*M_\tau^\bot}(\*U^*, \*M_0\*\mu_0)<  \gamma_1^2\}\\ & \qquad + P\big\{\hat{\*M}^* = \*M_\tau, \rho(\*W, \*U^*) > 1-\gamma^2, \rho_{\*M_\tau^\bot}(\*U^*, \*M_0\*\mu_0) \geq  \gamma_1^2\}\\
    & \leq P\left\{\|\*W\|  >  C^{II}_{\gamma_1,\gamma, \tau}\right\} +   P\big\{\rho(\*W, \*U^*) > 1-\gamma^2, \rho_{\*M_\tau^\bot}(\*U^*, \*M_0\*\mu_0) \geq  \gamma_1^2\}\\
    & \leq P\left\{\|\*W\|  >  C^{II}_{\gamma_1,\gamma, \tau}\right\} +  P\left\{g_{\*M_\tau}(\*U) \leq \sqrt{\tilde\gamma}\right\} +  P\left\{\rho_{\*M_\tau^\bot}(\*U, \*M_0\*\mu_0) \leq  \tilde \gamma_1^2\right\}.
 \end{align*}
  To bound the probability  $P\left\{g_{\*M_\tau}(\*U) \leq \sqrt{\tilde\gamma}\right\},$
 we have
 \begin{align*}
   & P(g_{\*M_\tau}(\*U) \leq  \sqrt{\tilde\gamma})= P\left(\frac{\|(I-H_{\*M_\tau})\*U\|^2}{\|H_{\*M_\tau}\*U\|^2} \leq \frac{\tilde\gamma}{1-\tilde\gamma}\right)=  F_{n-\tau, \tau}\left(\frac{\tilde\gamma/(n-\tau)}{(1-\tilde\gamma)/\tau}\right).
\end{align*}
 When $n-\tau>4$ and $\tilde\gamma<0.6$,
\begin{align*}
    F_{n-\tau, \tau}\left(\frac{\tilde\gamma/(n-\tau)}{(1-\tilde\gamma)/\tau}\right) & = \frac{{\int}_0^{\tilde\gamma} t^{\frac{n-\tau}{2}-1} (1-t)^{\frac{\tau}{2}-1}dt}{\mathcal B(\frac{n-\tau}{2}, \frac{\tau}{2})}\leq \frac{\tilde\gamma^{\frac{n-\tau}{2}-1}}{\mathcal B(\frac{n-\tau}{2}, \frac{\tau}{2})} \leq \left(\frac{n-\tau}{2}\right)^{\frac{\tau}{2}}\tilde\gamma^{\frac{n-\tau}{2}-1},
\end{align*}
since the beta function $\mathcal B(\frac{n-|\tau|}{2}, \frac{|\tau|}{2}) \geq \left(\frac{n-|\tau|}{2}\right)^{-\frac{|\tau|}{2}}.$
Moreover, By Lemma~\ref{lem::angle}, we have
 \begin{align*}
     P\left\{\rho_{\*M_\tau^\bot}(\*U, \*M_0\*\mu_0)\right\} \leq 2(\arccos \tilde\gamma_1)^{n-\tau-1}.
 \end{align*}

Then for any $\*M_\tau, \tau \leq \tau_0,$ we now have the following probability bound
\begin{align*}
    & P\big\{\hat{\*M}^* = \*M_\tau, \rho(\*W, \*U^*) > 1-\gamma^2\}\\
    & \leq  \exp\left\{-\frac{1}{3\sigma^2_{\min}} C^{II}_{\gamma,\gamma_1, \tau}\ + 0.25n \right\} + \left(\frac{n-\tau}{2}\right)^{\frac{\tau}{2}}\tilde\gamma^{\frac{n-\tau}{2}-1} +   2(\arccos \tilde\gamma_1)^{n-\tau-1}.
\end{align*}
Therefore, 
\begin{align*}
     & \sum_{\{\*M_\tau, \tau \leq \tau_0\}}P\big\{\hat{\*M}^* = \*M_\tau, \rho(\*W, \*U^*) > 1-\gamma^2\}\\
     & \leq \sum_{\tau=1}^{\tau_0} {n \choose \tau} \left\{\exp\left(-\frac{1}{3\sigma^2_{\min}} C^{II}_{\gamma,\gamma_1, \tau}\ + 0.25n \right) + \left(\frac{n-\tau}{2}\right)^{\frac{\tau}{2}}\tilde\gamma^{\frac{n-\tau}{2}-1} +  2 (\arccos \tilde\gamma_1)^{n-\tau-1}\right\}\\
     & \leq \sum_{\tau=1}^{\tau_0} n^\tau \left\{\exp\left(-\frac{1}{3\sigma^2_{\min}} C^{II}_{\gamma,\gamma_1, \tau}\ + 0.25n \right) + \left(\frac{n-\tau}{2}\right)^{\frac{\tau}{2}}\tilde\gamma^{\frac{n-\tau}{2}-1} +  2 (\arccos \tilde\gamma_1)^{n-\tau-1}\right\}\\
     & \leq  \sum_{\tau=1}^{\tau_0}\exp\left\{-\frac{1}{3\sigma^2_{\min}} C^{II}_{\gamma_1,\gamma, \tau}\ + 0.25n + \tau\log(n)\right\} + n^{3\tau_0/2} \tilde\gamma^{\frac{n-\tau_0}{2}-1}\sum_{t=0}^{\tau_0-1} \left(\frac{\tilde\gamma}{n^3}\right)^t \\
     & \qquad \qquad \qquad +  2(\arccos \tilde\gamma_1)^{n-\tau_0-1} n^{\tau_0} \sum_{t=0}^{\tau_0-1} \left(\frac{\arccos \tilde\gamma_1}{n}\right)^{t}\\
     & \leq \sum_{\tau=1}^{\tau_0}\exp\left\{-\frac{1}{3\sigma^2_{\min}} C^{II}_{\gamma_1,\gamma, \tau}\ + 0.25n + \tau\log(n)\right\} + 2n^{3\tau_0/2} \tilde\gamma^{\frac{n-\tau_0}{2}-1}
     + 4(\arccos \tilde\gamma_1)^{n-\tau_0-1} n^{\tau_0}.
\end{align*}

It then follows from \eqref{eq:CI_gamma} and the above that 
\begin{align*}
\P\{\hat{\*M}^* \neq \*M_0\} &  = P\{\hat{\*M}^* \neq \*M_0,  \rho(\*W, \*U^*) > 1-\gamma^2\} + P\{\hat{\*M}^* \neq \*M_0,  \rho(\*W, \*U^*) \leq  1-\gamma^2\}\\ 
    & \leq P\{\hat{\*M}^* \neq \*M_0,  \rho(\*W, \*U^*) > 1-\gamma^2\} +  \P\{\rho(\*W, \*U^*) \leq  1-\gamma^2\}\\
    &\leq \sum^{K_{\max}}_{\tau=\tau_0+1}
    \exp\left\{-\frac{1}{3\sigma^2_{\min}} C^I_{\gamma, \tau} + 0.25n +  \tau\log(n) \right\} \\
    & \quad \quad  +  \sum_{\tau=1}^{\tau_0} \exp\left\{-\frac{1}{3\sigma^2_{\min}} C^{II}_{\gamma_1,\gamma, \tau}\ + 0.25n + \tau\log(n)\right\}\\
     & \quad \quad  +  2n^{3\tau_0/2} \tilde\gamma^{\frac{n-\tau_0}{2}-1}
     + 4(\arccos \tilde\gamma_1)^{n-\tau_0-1} n^{\tau_0} + P\{\rho(\*W, \*U^*) \leq  1-\gamma^2\}. 
\end{align*}
The above is similar to Condition (N1) with the difference being the first four terms, all of which, we will show below, can be arbitrarily small. 
It then follows immediately from \eqref{eq: 1-P_delta} in the proof of  Lemma~\ref{lemma:p_bound_C_D_general}and Lemma~\ref{lem::angle} that
\begin{align}
\label{eq:finite_p_bound_C_D}
     P\{\*M_0 \not\in \widehat{\it \Upsilon}_{{\cal V}^*}(\*Y)\} &\leq \sum^{K_{\max}}_{\tau=\tau_0+1}
    \exp\left\{-\frac{1}{3\sigma^2_{\min}} C^I_{\gamma, \tau} + 0.25n +  \tau\log(n) \right\} \nonumber \\
    & \quad  +  \sum_{\tau=1}^{\tau_0} \exp\left\{-\frac{1}{3\sigma^2_{\min}} C^{II}_{\gamma_1,\gamma, \tau}\ + 0.25n + \tau\log(n)\right\}\nonumber\\
     & \quad +  2n^{3\tau_0/2} \tilde\gamma^{\frac{n-\tau_0}{2}-1}
     + 4(\arccos \tilde\gamma_1)^{n-\tau_0-1} n^{\tau_0} + \left(1-\frac{\gamma^{n-2}\arcsin (\gamma)}{n-1}\right)^{|{\cal V}^*|}. 
\end{align}

We now will try to simplify the bound by
first making $\gamma_1= \sqrt{1- \tilde\gamma^{1/2}},$ then $\tilde\gamma_1=(1-\sqrt{\tilde\gamma}^{1/2})\sqrt{1- \tilde\gamma} - \sqrt{2-2\sqrt{1-\tilde\gamma^2}} \geq 1 - 2\sqrt{\tilde\gamma} \geq 0$ for $0 \leq \tilde\gamma \leq 0.25,$ from which we have $\arccos \tilde\gamma_1 \leq  \arccos (1 - 2\sqrt{\tilde\gamma}) \leq 2.1 \tilde\gamma^{1/4}<1.$  Further $C^{II}_{\gamma_1,\gamma, \tau} = \frac{1-\gamma_1^2}{5\gamma}C_{\min} = \frac{\sigma^{1/2}_{\max}}{5\sigma^{1/2}_{\min}\gamma^{1/2}} C_{\min},$  and the bound for $\lambda$ in \eqref{eq:lambda_upper} is reduced to   $\lambda \leq \frac{\log\{C^2_{\min}(1- \gamma^2_1)/2 + 1\}}{2 \tau_0 \log(n)/n} = \frac{\log\{0.5C^2_{\min}\tilde\gamma^{1/2} + 1\}}{2 \tau_0 \log(n)/n} .$

Moreover, because $\arcsin x \geq x$, we have
\begin{align*}
    1-\frac{\gamma^{n-2}\arcsin (\gamma)}{n-1} \leq 1-\frac{\gamma^{n-1}}{n-1}.
\end{align*}

Then the bound in \eqref{eq:finite_p_bound_C_D} is simplified as
\begin{align*}
       P\{\*M_0 \not\in \widehat{\it \Upsilon}_{{\cal V}^*}(\*Y)\} &\leq \sum^{K_{\max}}_{\tau=\tau_0+1}
    \exp\left\{-\frac{2(\tau-\tau_0)}{3\sigma^2_{\min}\sqrt{\gamma}}    + 0.25n +  \tau\log(n) \right\} \nonumber \\
    & \quad \quad  +  \sum_{\tau=1}^{\tau_0} \exp\left\{-\frac{\sigma^{1/2}_{max}}{15\sigma^{5/2}_{\min}\sqrt\gamma} C_{\min}\ + 0.25n + \tau\log(n)\right\}\nonumber\\
     & \quad \quad  +  2n^{3\tau_0/2} \tilde\gamma^{\frac{n-\tau_0}{2}-1}
     + 4(20 \tilde\gamma)^{\frac{n-\tau_0-1}{4}} n^{\tau_0} + \left(1-\frac{\gamma^{n-1}}{n-1}\right)^{|{\cal V}^*|}. 
\end{align*}
\end{proof}

% \subsection*{Proof of Theorem~\ref{thm:bound_of_C_d_mixture}}

\begin{lemma}
\label{lemma:bound_of_C_d_mixture}
Suppose $n-\tau_0>4.$ Let the distance metric between $\*U$ and $\*U^*$ be $\rho(\*U, \*U^*) = \frac{\|(\*I - \*H_{\*U^*})\*U\|}{\|\*U^*\|},$ then for any $\delta >0,$ there exists a $\gamma_\delta>0$ and an interval of positive width $\Lambda_{\delta} =[\frac{\gamma_{\delta}^{1.5}}{\log(n)/n},\frac{\log\{0.5C^2_{\min}\sigma_{\max}\gamma_{\delta}/\sigma_{\min} + 1\}}{2 \tau_0 \log(n)/n}]$ such that when $\lambda \in \Lambda_{\delta},$
\begin{align*}
    \P_{\*U, {\cal V}^*}\{(\tau_0,\*M_0)\not\in \widehat{\it \Upsilon}_{{\cal V}^*}(\*Y)\} & \leq \delta + \left[1- \P\left\{\mathcal \rho(\*U, \*U^*) \leq \gamma_{\delta}\right\}\right]^{|\mathcal V^*|} \\
   &  \leq\delta + \left(1-\frac{\gamma_\delta^{n-1}}{n-1}\right)^{|\mathcal V^*|}.
   % \rightarrow \delta, 
   % \mbox{ as $|\mathcal V^*| \rightarrow \infty$. }
\end{align*}
    % C_{\*M_0, \*M_\tau} >0 $, where  $C_{\*M_0, \*M_\tau} = \|(\*I - \*M_\tau)\*M_0\*\mu_0\|,$
\end{lemma}

\begin{proof}[Proof of Lemma~\ref{lemma:bound_of_C_d_mixture}]
Because the first four terms of \eqref{eq:prob_bound_C_D_simple} converges to 0 as $\gamma$ goes to 0, for any $\delta >0$, there exist a $\gamma_{\delta}$ that the probability bound in  \eqref{eq:prob_bound_C_D_simple} reduces to $ \delta + \left(1-\frac{\gamma_\delta^{n-1}}{n-1}\right)^{|\mathcal V^*|}.$ Moreover, the interval $\Lambda_\delta$ is of positive length because $x^{1.5} \leq \frac{1}{\tau_0}\log(cx+1)$ in a neighbourhood of 0 for any constant $c>0.$ Hence Lemma~\ref{lemma:bound_of_C_d_mixture} follows immediately. 
 \end{proof}
 
 \begin{proof}[Proof of Theorem~\ref{thm:bound_of_C_d_mixture}]
 Let $P_\delta = \left(1-\frac{\gamma_\delta^{n-1}}{n-1}\right)$, by Lemma~\ref{lemma:bound_of_C_d_mixture} and Markov Inequality,
 \begin{align*}
      & \P_{{\cal V}^*}\left[\P_{\*U| {\cal V}^*}\{(\tau_0,\*M_0)\not\in \widehat{\it \Upsilon}_{{\cal V}^*}(\*Y)\} -\delta \geq (1-P_{\delta})^{|\mathcal V^*|/2}\right]\\
      & \leq  \frac{E_{\mathcal V^*}\P_{\*U|\mathcal V^*}\{(\tau_0,\*M_0)\not\in \widehat{\it \Upsilon}_{{\cal V}^*}(\*Y)\}}{(1-P_{\delta })^{|\mathcal V^*|/2}} \\
      & =  \frac{\P_{\*U,\mathcal V^*}\{(\tau_0,\*M_0)\not\in \widehat{\it \Upsilon}_{{\cal V}^*}(\*Y)\}}{(1-P_{\delta })^{|\mathcal V^*|/2}} 
      \leq {(1-P_{\delta })^{|\mathcal V^*|/2}},
 \end{align*}
from which (a) follows by making $c<-\frac{1}{2}\log(1- P_\delta).$ Then (b) follows from (a) and Corollary~\ref{cor:cand_cs_coverage}.
 
%  	\begin{align*}
% 	\P_{\cal V^*}\left[\P_{\*U|\mathcal V^*}\{\kappa_0 \not\in \widehat{\it \Upsilon}_{{\cal V}^*}(\*Z)\} \geq (1-P_{\cal N})^{|\mathcal V^*|/2}\right] & \leq \frac{E_{\mathcal V^*}\P_{\*U|\mathcal V^*}\{\kappa_0 \not\in \widehat{\it \Upsilon}_{{\cal V}^*}(\*Z)\}}{(1-P_{\cal N})^{|\mathcal V^*|/2}}\\
% 	& \leq \frac{\P_{\*U,\mathcal V^*}\{\kappa_0 \not\in \widehat{\it \Upsilon}_{{\cal V}^*}(\*Z)\}}{(1-P_{\cal N})^{|\mathcal V^*|/2}} = (1-P_{\cal N})^{|\mathcal V^*|/2}.	
% 	\end{align*}
\end{proof}

\newpage
\section*{Appendix II: A Brief Review and Comparison with Existing and Relevant Inference Procedures}

The repro samples method is an entirely frequentist approach, but
its development inherits some key ideas from several 
existing inferential procedures
across  Bayesian, fiducial and frequentist (BFF) paradigms.  
For readers' convenience, we first provide in Subsection II-A a brief review of several relevant BFF inferential approaches. We then provide in Subsection II-B a comparative discussion to highlight the strengths and uniqueness of the proposed repro samples method, and also describe a role that the  samples method can play in bridging across BFF inferences.

\subsection*{II-A. A Brief Review of Existing and 
Relevant BFF procedures}

The repro samples method utilizes artificial sample data sets for inference. There are several relevant artificial sample-based inferential approaches across the BFF paradigms. 

\vspace{2mm} \noindent
{\bf (i) Approximate Bayesian Computing (ABC)}
refers to a family of techniques used to approximate posterior densities of $\btheta$ by bypassing direct likelihood evaluations \citep[cf.][]{Rubin1984,Tavare1997, 
Peters2012}.  
A basic version of the so-called {\it rejection ABC algorithm} has the following [a] - [c] steps:
{\it 
   [a] Simulate an artificial model parameter $\btheta^*$ from a prior $\pi(\btheta)$ and, given $\btheta^*$, simulate an artificial dataset ${\bf y}^*$ 
(e.g., set ${\bf y}^* = G(\btheta^*, {\bf u}^*)$ for a simulated ${\bf u}^*$); 
[b] If 
$ {\bf y}^* \approx {\bf y}_{obs}$,
 we collect the artificial parameter $\btheta^*$; 
 [c] Repeat steps [a] \& [b] to obtain a large set of $\btheta^*$.}
ABC is in fact 
a Bayesian inversion method,
since the match ${\bf y}^* \approx {\bf y}_{obs}$ is equivalent to solving $G(\btheta, {\bf u}^*) \approx {\bf y}_{obs}$ for $\btheta^*$.

Operationally,
because it is difficult to match ${\bf y}^* \approx {\bf y}_{obs}$, ABC instead matches
$|S({\bf y}_{obs}) - S({\bf y}^*)|\leq \epsilon$ for a pre-chosen summary statistic $S(\cdot)$ and tolerance $\epsilon >0$. The kept $\btheta^*$ form  
{\it an ABC posterior},  $p_{\epsilon}(\btheta|\*y_{obs})$.  
If $S(\cdot)$
is a {\it sufficient statistic},  then $p_{\epsilon}(\btheta | \*y_{obs}) \to p(\btheta | \*y_{obs})$,  
the target posterior distribution,  as $\epsilon \to 0$ at~a~fast~rate~\citep{Barber2015,Li2016}.
To improve computing efficiency, this basic ABC algorithm has been extended to, for instance, more complex ABC-Markov Chain Monte Carlo (ABC-MCMC) type of algorithms, but the key matching idea remains. 
Note that, the ABC method seeks to match each single of artificial sample copy $\*y^*$ with the observed data $\*y_{obs}$ through a summary statistic and a preset tolerance level $\epsilon$ to form a rejection rule. 

Bayesian literature has 
 pointed out two remaining issues in ABC~\citep[cf.,][]{Li2016,Thornton2018}. First, When $S(\cdot)$
is not sufficient, $p_{\epsilon}(\btheta | \*y_{obs}) \not\to p(\btheta | \*y_{obs})$.  
In this case,
there is no guarantee that an ABC posterior is a Bayesian posterior thus 
the interpretation of the inferential result by the ABC posterior is unclear.  
However,  
deriving a sufficient statistic is not possible when the likelihood is intractable, thus it places a limitation on the practical use of ABC as a likelihood-free inference approach. 
Second, an ABC method requires the preset threshold $\epsilon \to 0$ at a fast rate,
leading to a degeneracy of computing efficiency 
(since the acceptance probability in Step [b] goes to $0$ for a very small $\epsilon \to 0$). 
In practice, there is  no clear-cut choice for an appropriate $\epsilon$ to balance the procedures' computational efficiency and inference validity
is an unsolved question \cite[e.g.,][]{Li2016}.

\vspace{2mm} \noindent
{\bf (ii) Generalized fiducial inference (GFI)} is a generalization of Fisher's fiducial method, which is understood in contemporary statistics  as {an inversion method} to solve a pivot equation for model parameter $\btheta$ \citep{Zabell1992, Hannig2016}. 
GFI extends the inversion from a pivot to actual data; i.e., solves for $\btheta^*$ from
${\bf y}_{obs} = G(\btheta, {\bf u}^*)$ for artificially generated ${\bf u}^*$'s. But, since a solution is not always possible, thus GFI 
considers an optimization under an $\epsilon$-approximation: 
\begin{equation}
 \btheta_\epsilon^* =  {\rm argmin}_{\theta \in \{\theta: \,\, || {\bf y}_{obs} - G(\theta, {\bf u}^*)||^2 \leq \epsilon \}}  || {\bf y}_{obs} - G({\theta}, {\bf u}^*)||^2
\label{eq:GFI}
\end{equation}
and let $\epsilon \to 0$ at a fast rate~\citep{Hannig2016}.
GLI often has a guaranteed statistical (frequentist) performance, provided $n \to \infty$ and $G$ is smooth/differentiable in $\btheta$~\citep{Hannig2009, Hannig2016}. Different than the ABC that uses rejection sample, the GFI instead uses optimizations to approximately solve the equation $\*y_{obs} = G(\btheta, \*u^*)$ within a $\epsilon$-constrained set in $\Theta$.
Operationally, 
GLI suffers the same computational issue as in ABC on the choice of the pre-specified $\epsilon$ since, 
as $\epsilon \to \infty$ at a fast rate, equation (\ref{eq:GFI}) may not have a solution. 

\vspace{2mm} \noindent
{\bf (iii) Efron's bootstrap and other related artificial sampling methods}, in which many copies of artificial data are generated, are a popular way to help quantify uncertainty in complicate estimation problems in statistics. Let $\hat \btheta$ be a point estimator of $\btheta$ and $\btheta^*$ be the corresponding bootstrap estimator
The key justification relies on the so-called {\it bootstrap central limit theorem} (CLT) $\btheta^* - \hat \btheta | \*y_{obs} \sim \hat \btheta -\btheta | \btheta$ \citep[e.g.,]{Singh1981,Bickel1981}, in which the (multinomial distributed) uncertainty in resampled artificial data asymptotically matches the  (often non-multinomial distributed) uncertainty in $\hat \btheta$ inherited from the sampling population. 
When CLT does not apply (e.g., the parameter space $\Theta$ is a discrete set), the methods are not supported by theory and often perform poorly.  

\vspace{3mm}\noindent
In addition to the above artificial-sample-based approaches, repro samples method is also related the so called {\it Inferential Model} (IM) that seeks to produce  (epistemic)  probabilistic inference for $\btheta_0$:  

\noindent
{\bf (iv) Inferential model (IM)}  is an attempt to develop a general framework for ``prior-free exact  probabilistic inference'' \citep{Martin2015}. As described in \cite{Martin2015} and under the model assumption (\ref{eq:1}), an IM procedure include three steps: 
{\it [\it A-step] Associate $(\*Y, \btheta)$ with the unobserved auxiliary statistic $\*U$, i.e., $\*Y = G({\btheta}, {\*U})$;
[\it P-Step] Predict the unobserved auxiliary statistic $\*U$ with a random set ${\cal S}$ with distribution~${\mathbb P}_{\cal S}$;
[\it C-step] Combine observed $\*y_{obs}$ with random set ${\cal S}$  into a new data-dependent random set 
$\Theta_{\*y_{obs}}({\cal S}) = \bigcup_{\*u \in {\cal S}} \left\{\btheta: \*y_{obs} = G(\btheta, \*u)\right\}$.} 
A novel aspect of the IM development that is beyond the typical fiducial procedures is that the uncertainty quantification of $\*Y$ is through the parameter free $\*U$ using a random set in [P-step]. This separation removes the impact of parameter $\*\theta$ and  makes the task of uncertainty quantification easier.  
In order to fully express the outcomes into probabilistic forms, IM development needs to use a system of imprecise probability known as {\it Dempster-Shafer Calculus}. 
The outcomes of an IM algorithm are the so-called plausible and belief functions (or lower and upper probability functions) for $\btheta_0$, which can be used for frequentist inference. The plausible and belief functions are akin to the so-called {\it upper or lower confidence distributions} \citep[e.g.,][]{Thornton2022}, although the IM method is not considered by most (including Martin and Liu) as a frequentist development.  

\subsection*{II-B. Repro samples method and Connections to BFF Inferential Approaches}

The repro samples method borrows two important ideas from the aforementioned BFF procedures. First, the repro samples method utilizes artificial data samples to help quantify the uncertainty of a statistical inference. This idea has been used in bootstrap, ABC, GFI, and Monte-Carlo hypothesis testing procedures, among others.  A key aspect of utilizing artificial data samples for inference in these procedures is the common attempt to match (in a certain way) the artificial data with the observed sample. This matching is done by directly comparing either the value or distribution of the artificial $\*y$ or its summary statistic with that of $\*y_{obs}$. The proposed repro samples method fully utilizes the artificial samples and the matching idea to both develop inference and solve computational questions. Second, the repro samples method borrows an idea from IM that we can quantify the inference uncertainty by first assessing the uncertainty of the unobserved $\*u^{rel}$. This allows us to develop both finite and asymptotic inferences under the frequentist framework without the need to use any likelihood or decision theoretical criteria. 
We consider that the repro samples method is a further development (generalization and refinements) of these existing BFF procedures. We provide some comparative discussions below. 

\vspace{2mm} \noindent
{\bf Repro samples method versus the classical Monte-Carlo hypothesis testing.}  In a few plain statements, the Monte-Carlo Algorithm~1 for the repro samples method can be described as follows: For any potential value $\btheta \in \Theta$, we generate multiple copies of $\*u^s$ (denote the set of their collection ${\cal V}$) and compute corresponding $T(\*u^s, \btheta)$'s. If there exists an $\*u^* \in {\cal U}$ such that $\*y_{obs} = G(\btheta, \*u^*)$ and the value $T(\*u^*, \btheta)$
is {\it conformal} with the multiple copies of $T(\*u^s, \btheta)$, $\*u^s \in {\cal V}$, at the level $\alpha$, then we keep the $\btheta$ in  $\Gamma_{\alpha}(\*y_{obs})$. Here, the matching of a single copy of $T(\*u^*, \btheta)$ with the multiple copies of $T(\*u^s, \btheta)$'s is evaluated by a {\it conformal measure} at level-$\alpha$. Here, the concept {\it conformal} is 
borrowed from recent developments of {\it conformal prediction} \citep{vovk2005algorithmic}. 

In the special case when the nuclear mapping is defined through a test statistic $T(\*u, \btheta) = \tilde T(\*y, \btheta)$ where $\*y = G(\btheta, \*u)$ (as discussed in Section 3.1) and suppose we can always find a $\*u$ such that $\*y_{obs} = G(\btheta, \*u^*)$ for any $\btheta$, the Monte-Carlo Algorithm~1 can be simplified to: for any potential value $\btheta \in \Theta$, we generate multiple copies of artificial data 
$\*y^s = G(\btheta, \*u^s)$ and compute the corresponding $\tilde T(\*y^s, \btheta)$, $\*u^s \in {\cal V}$. If the test statistic $\tilde T(\*y_{obs}, \btheta)$ is conformal with the many copies of $\tilde T(\*y^s, \btheta)$'s at level $\alpha$, then we keep this $\btheta$ to form the level-$\alpha$ confidence set $\Gamma_{\alpha}(\*y_{obs})$. 
This simplified algorithm is exactly an Monte-Carlo implementation of the classical hypothesis test of $H_0: \btheta_0 = \btheta$, and the conformal statement above corresponds to the statement of not rejecting $H_0$. Again, we see that the classical testing approach of our repro samples method. In the Monte-Carlo version, we utilize multiple copies of artificial sample $\*y^s$'s and compare them to a single copy of the observed data $\*y_{obs}$.

\vspace{2mm} \noindent
{\bf Repro samples method versus ABC and GFI.}
Both GFI and repro samples developments highlight the matching of equation $\*y_{obs} = G(\btheta, \*u^*)$ and the fact that, when $\*u^*$ matches or is close to $\*u^{rel}$, the solution of $\btheta$ from the equation is equal or close to the true $\btheta_0$. The difference is that the GFI approach compares each single copy of artificial sample $\*y^* = G(\btheta, \*u^*)$ with the single copy of the observed $\*y_{obs}$, so does the ABC method, while the repro samples method compares the single copy of the observed $\*y_{obs}$ with multiple copies of $\*y^s = G(\btheta, \*u^s)$
(or, more accurately, the realized $T(\btheta, \*u^{rel})$ with multiple copies of $T(\btheta, \*u^s)$'s) for $\*u^s \in {\cal V}$. By a single comparing observed sample $\*y_{obs}$ (or realized $\*u^{rel}$) with many copies of artificiallt generated $\*y^s$ (or $\*u^s \in {\cal V}$), the repro samples method has the advantage to use a level $\alpha$ to calibrate the uncertain and side-step the difficult problem on how to appropriately preset the
threshold $\epsilon$ in both GFI and ABC. Furthermore, the use of the nuclear mapping adds much flexibility to the repro samples method. Unlike the ABC method, the often inconvenient requirement of using a sufficient statistic is not needed to ensure the validity of the repro samples method. 

\vspace{2mm} \noindent
{\bf Repro samples method versus IM.}
The repro samples and IM methods both promote to first quantify the uncertainty of $\*U$ and use it to help address the overall uncertainty inherited in the sampling data. In order to produce a level $\alpha$ confidence set, the repro sample method use a single fixed Borel set and a $\btheta$-dependent nuclear mapping function to help quantify the uncertainty in $\*U$. The repro samples method is a fully frequentist approach
developed using only the standard probability tool. The IM on the other hand attempts to achieve a higher goal of producing a prior-free probabilistic inference for the unknown parameter $\btheta$. By doing so, it requires to use random sets and a complex imprecise probability system of Dempster-Shafer calculus. 
Furthermore, IM focuses on finite sample inference, while the repro samples can be used for both finite and large sample inferences. 
Finally, the repro samples method  stresses the actual matching of simulated artificial $\*u^*$ with the (unknown) realized $\*u^{rel}$ as well. The development of candidate set for discrete parameter $\kappa$ in Section 3.3 is rooted on the fact that, when $\*u^*$ is equal or in a neighborhood of $\*u^{rel}$ can give us the same $\kappa_0$ in the many-to-one mapping function (\ref{eq:tau-star}). This cannot be foreseen or easily adapted in an IM approach. 

\vspace{2mm} \noindent
{\bf Artificial-sample-based inference and a bridge of BFF paradigms}. There have been several recent developments on the foundation of statistical inference across Bayeisan, fiducial and frequentist (BFF) paradigms. \cite{Reid2022} and \cite{Thornton2022} provide comprehensive overviews on several recent BFF research developments at the foundation level. Based on the development of confidence distribution,  \cite{Thornton2022} also explores the aspect of making inference though matching artificial samples with observed data across BFF procedures. It argues that the matching of simulated artificial randomness with the sampling randomness inherited form a statistical model provides a unified
bridge to connect BFF inference procedures.
Here, the artificial randomness includes, for examples, bootstrap randomness in bootstrap, MCMMC randomness in a Bayesian analysis, the randomness of $\*U^*$ in fiducial procedures such as 
GLM and IM. 
The repro samples method is another development that seeks to align simulated artificial randomness with the sampling randomness. By doing so we can effective measure and quantify the uncertainty in our statistical inferential statements. 
The development can be used to address many difficult inference problems, especially those involving discrete parameter space and also those where large sample CLT does not apply. It provides a further advancement on the foundation of statistical inference to meet the growing need for ever-emerging data science.

\end{document}